\documentclass{article}
\usepackage{kr}

\usepackage{times}
\usepackage{soul}
\usepackage{url}
\usepackage[utf8]{inputenc}
\usepackage[small]{caption}
\usepackage{graphicx}
\usepackage{amsmath}
\usepackage{amsthm}
\usepackage{booktabs}
\usepackage{algorithm}
\usepackage[noend]{algorithmic}
\usepackage[hidelinks,breaklinks]{hyperref}

\def\nnot{\neg}
\def\accp{\mathtt{acc}}
\def\defp{\mathtt{def}}
\def\attackp{\mathtt{attack}}
\def\bp{\mathtt{b}}
\def\fact{\mathtt{a}}
\def\confp{\mathtt{confl}}
\def\prefp{\mathtt{pref}}
\def\nextp{\mathtt{next}}
\def\oddp{\mathtt{odd}}

\def\boundedconf{\textsc{BoundConf}\xspace}
\def\binaryconf{\textsc{BinConf}\xspace}
\def\polyCQ{\textsc{PolyBCQ}\xspace}
\def\polycons{\textsc{PolyCons}\xspace}

\newcommand{\mypar}[1]{\medskip\noindent\textbf{#1.}}

\def\Gmc{\mathcal{G}}

\def\cn{\ensuremath{\mathsf{N_{C}}\xspace}}
\def\rn{\ensuremath{\mathsf{N_{R}}\xspace}}
\def\inds{\ensuremath{\mathsf{N_{I}}\xspace}}

\def\dllite{DL-Lite\xspace}
\def\dllitecore{DL-Lite$_{\mn{core}}$\xspace}
\def\dlliter{DL-Lite$_{\mathcal{R}}$\xspace}
\def\dlliterhorn{DL-Lite$^{\mathcal{H}}_\mn{horn}$\xspace}

\newcommand{\tup}[1]{\langle #1\rangle}

\newcommand{\bravemodels}[1]{\models_{\text{brave}}^{#1}}
\newcommand{\armodels}[1]{\models_{\text{AR}}^{#1}}
\newcommand{\iarmodels}[1]{\models_{\text{IAR}}^{#1}}

\newcommand{\conflicts}[1]{\mi{Conf}(#1)}

\newcommand{\reps}[1]{\mi{SRep}(#1)}
\newcommand{\greps}[1]{\mi{GRep}(#1)}
\newcommand{\preps}[1]{\mi{PRep}(#1)}
\newcommand{\creps}[1]{\mi{CRep}(#1)}
\newcommand{\xreps}[1]{\mi{XRep}(#1)}

\def\defeat{\rightsquigarrow}
\def\sdefeat{\defeat_\succ}
\def\args{\mi{Args}}

\newcommand{\chf}[1]{\ensuremath{\Gamma_{#1}}}
\newcommand{\chff}[0]{\ensuremath{\Gamma_{F}}}

\newcommand\aczero{\ensuremath{\mathsf{AC}^{0}}\xspace}

\def\ptime{\textsc{P}\xspace}
\def\np{\textsc{NP}\xspace}
\def\conp{co\textsc{NP}\xspace}
\def\piptwo{\ensuremath{\Pi^{p}_{2}}\xspace}
\def\sigmaptwo{\ensuremath{\Sigma^{p}_{2}}\xspace}
\def\delayp{\textsc{DelayP}\xspace}
\def\incrdelayp{\textsc{IncP}\xspace}
\def\totalp{\textsc{TotalP}\xspace}

\def\true{\ensuremath{\mathsf{true}}}
\def\false{\ensuremath{\mathsf{false}}}

\newcommand{\mn}[1]{\ensuremath{\mathsf{#1}}}
\newcommand{\mi}[1]{\ensuremath{\mathit{#1}}}

\newcommand{\Amc}{\ensuremath{\mathcal{A}}}
\newcommand{\Bmc}{\ensuremath{\mathcal{B}}}
\newcommand{\Cmc}{\ensuremath{\mathcal{C}}}

\newcommand{\Imc}{\ensuremath{\mathcal{I}}}
\newcommand{\Kmc}{\ensuremath{\mathcal{K}}}

\newcommand{\Tmc}{\ensuremath{\mathcal{T}}}

\newcommand{\eg}{e.g.,~}

\newcommand{\ie}{i.e.,~}
\newcommand{\wrt}{w.r.t.~}

\usepackage{latexsym}
\usepackage{amsfonts, amssymb, stmaryrd}
\usepackage{xspace}
\usepackage{xcolor}
\usepackage{thm-restate} 
\usepackage{tikz}
\tikzset{>=latex}
\usepackage{subcaption}

\newtheorem{theorem}{Theorem}
\newtheorem{definition}[theorem]{Definition}
\newtheorem{example}[theorem]{Example}
\newtheorem{lemma}[theorem]{Lemma}

\newtheorem{remark}[theorem]{Remark}
\newtheorem{corollary}[theorem]{Corollary}

\usepackage{fancyhdr}
\fancypagestyle{firstpage}
{
    \fancyhead[L]{Long version (with proofs) of the KR'20 paper with the same title. Corrects the statement of Theorem 43 (missing hypothesis).}    
    \fancyhead[R]{}
}

\title{Querying and Repairing Inconsistent Prioritized Knowledge Bases:
\\ Complexity Analysis and Links with Abstract Argumentation}

\author{Meghyn Bienvenu$^1$\and Camille Bourgaux$^2$\\
\affiliations
$^1$ CNRS \& University of Bordeaux, France\\
$^2$ DI ENS, ENS, CNRS, PSL University \& Inria, Paris, France\\  
\emails
meghyn.bienvenu@u-bordeaux.fr, camille.bourgaux@ens.fr
}
 
\begin{document}

\maketitle
\thispagestyle{firstpage}
\begin{abstract}
In this paper, we explore the issue of inconsistency handling over prioritized knowledge bases (KBs), which consist of an ontology, a set of facts, and a priority relation between conflicting facts. In the database setting, a closely related scenario has been studied and led to the definition of  three different notions of optimal repairs (global, Pareto, and completion) of a prioritized inconsistent database. After transferring the notions of globally-, Pareto- and completion-optimal repairs to our setting, we study the data complexity of the core reasoning tasks: query entailment under inconsistency-tolerant semantics based upon optimal repairs, existence of a unique optimal repair, and enumeration of all optimal repairs. Our results provide a nearly complete picture of the data complexity of these tasks for ontologies formulated in common DL-Lite dialects. The second contribution of our work is to clarify the relationship between optimal repairs and different notions of extensions for (set-based) argumentation frameworks. Among our results, we show that Pareto-optimal repairs correspond precisely to stable extensions (and often also to preferred extensions), and we propose a novel semantics for prioritized KBs which is inspired by grounded extensions and enjoys favourable computational properties. Our study also yields some results of independent interest concerning preference-based argumentation frameworks.

\end{abstract}

\section{Introduction}
Ontology-mediated query answering (OMQA) improves data access through the use of an ontology,
which provides a convenient vocabulary for query formulation and captures domain knowledge
that is exploited during query evaluation \cite{DBLP:journals/jods/PoggiLCGLR08,DBLP:conf/rweb/BienvenuO15,DBLP:conf/ijcai/XiaoCKLPRZ18}. There is now a large literature on OMQA, 
with much of the work 
adopting description logics (DLs) \cite{DBLP:books/daglib/0041477}, or the closely related OWL (2) standard \cite{owl2}, 
as the ontology specification language. 
In particular, the DL-Lite family of lightweight DLs \cite{calvaneseetal:dllite,DBLP:journals/jair/ArtaleCKZ09} (which underpin the OWL 2 QL profile \cite{profiles})
has been a main focus of both theoretical and practical research 
due to its favourable computational properties.

While it is often reasonable to assume that the ontology has been debugged and contains only trusted knowledge, 
real-world datasets are plagued by data quality issues, which may render the KB inconsistent. 
This has led to a large body of work on how to handle data inconsistencies in OMQA,
with the proposal of several inconsistency-tolerant semantics to
provide meaningful answers to queries posed over inconsistent KBs
(see e.g.\ the surveys \cite{DBLP:conf/rweb/BienvenuB16} and \cite{DBLP:conf/dlog/Bienvenu19}). 
Many of these semantics are based upon some notion of repair, defined as 
an inclusion-maximal subset of the data that is consistent with the ontology. 
This is in particular the case for the AR, brave, and IAR semantics, 
which correspond respectively to a query answer holding w.r.t. all repairs, 
at least one repair, or the intersection of all repairs \cite{LemboLRRS10,DBLP:conf/ijcai/BienvenuR13}. 
The computational properties of these and other semantics are now quite well understood,
and some practical algorithms and implementations have begun to be developed, 
see e.g.\ \cite{DBLP:journals/ws/LemboLRRS15,DBLP:journals/jair/BienvenuBG19,DBLP:conf/ijcai/TsalapatiSSK16,DBLP:conf/aaai/TrivelaSV18}. 

In many scenarios, there is some information about the dataset that can be used to select the most relevant repairs (e.g.\ we might know which relations or sources are most reliable, or the likelihood of certain kinds of facts being correct). 
To exploit such information, variants of the preceding semantics have been considered 
in which we restrict attention to the 
most preferred repairs based upon cardinality, weight, or a stratification of the dataset into priority levels \cite{DBLP:conf/aaai/BienvenuBG14}. 
While the latter forms of preferences are quite natural, 
there are other relevant ways of defining preferred repairs that are worth exploring.
In particular, in the database area, 
the seminal work of \citeauthor{DBLP:journals/amai/StaworkoCM12} \shortcite{DBLP:journals/amai/StaworkoCM12} supposes that preferences are given in terms of a priority relation (i.e. acyclic binary relation) between conflicting facts.
Such `fact-level' preferences have been shown to naturally arise in applications like information extraction and can e.g.\ be declaratively specified using rules \cite{DBLP:journals/tods/FaginKRV16}.
To lift a priority relation between facts to the level of repairs, three different methods were proposed by Staworko et al.
Pareto-optimal and globally-optimal repairs are defined as those for which there is no possible substitution of facts that leads to a (Pareto / global) improvement, 
while completion-optimal repairs correspond to greedily constructing a repair based upon some compatible total order (see Section \ref{sec:reps} for formal definitions).
The complexity of reasoning with these three kinds of optimal repair has been investigated, primarily focusing on constraints given as functionality dependencies  \cite{DBLP:conf/pods/FaginKK15,DBLP:conf/icdt/KimelfeldLP17,DBLP:conf/pods/LivshitsK17}. 

In this paper, we explore the use of fact-level preferences and optimal repairs for inconsistency handling in OMQA. 
%
Our first contribution is a data complexity analysis of the central reasoning problems related to optimal repairs: conjunctive query entailment under inconsistency-tolerant semantics based upon optimal repairs, uniqueness of optimal repairs, and enumeration of all optimal repairs. 
Our results provide a nearly complete picture of the data complexity for the three types of optimal repair and for ontologies formulated in \dllitecore and \dlliterhorn; however, as we make precise later, many of our results transfer to other data-tractable ontology languages. Not surprisingly, we find that reasoning with optimal repairs is generally more challenging than for standard repairs. In particular, query entailment under variants of the AR, IAR, and brave semantics based upon any of the three notions of optimal repair is intractable in data complexity, whereas with standard repairs, the IAR and brave semantics allow for tractable query answering. 

Our second contribution is to establish connections with abstract argumentation frameworks (AFs). More precisely, we show that every prioritized KB with only binary conflicts naturally corresponds to a preference-based AF \cite{DBLP:conf/comma/KaciTV18}. To handle general prioritized KBs, we need preference-based set-based AFs (SETAFs), which we introduce as a natural extension of SETAFs \cite{DBLP:journals/ijar/FlourisB19}. This correspondence enables us to compare repairs and extensions. We determine that Pareto-optimal repairs are precisely the stable extensions of the corresponding preference-based (SET)AF, and under reasonable assumptions, also coincide with the preferred extensions. To establish the latter result, we prove a technically challenging result about symmetric preference-based SETAFs that we believe is of independent interest. Globally-optimal and completion-optimal repairs correspond to proper subsets of the stable extensions, but do not at present have any analog in the argumentation setting. 

The argumentation connection situates optimal repairs within a broader context and lays the foundations for importing ideas and results from the argumentation literature.
Indeed, our third contribution is to propose a new notion of grounded repair, directly inspired by grounded extensions from argumentation. 
We show that the (unique) grounded repair is contained in the intersection of Pareto-optimal repairs. 
As the grounded repair can be computed in polynomial time from the dataset and conflicts, it yields a tractable approximation of the Pareto variant of the IAR semantics for the considered DLs. Moreover, we show that it is more productive than the recently proposed Elect semantics \cite{DBLP:conf/lpnmr/BelabbesBC19,DBLP:conf/ksem/BelabbesB19}. These advantages motivate us to take a closer look 
at the computational properties of the grounded semantics. We prove in particular a matching \ptime\ lower bound and show how the semantics can be computed via the well-founded semantics of logic programs. 

Proofs of all results are provided in the appendix.

%
%
%
%
%
%
%

\section{Preliminaries}

Even if many of our results apply to more general settings, our focus is on description logic  (DL) knowledge bases, and we will in particular consider the \dllitecore and \dlliterhorn\ dialects of the \dllite family.

\mypar{Syntax}
A DL \emph{knowledge base (KB)} $\Kmc=\tup{\Tmc,\Amc}$ consists of an ABox $\Amc$ and a TBox $\Tmc$, both constructed from a set \cn\ of \emph{concept names} (unary predicates), 
a set of \rn\ of \emph{role names} (binary predicates), and a set \inds\ of \emph{individuals} (constants). 
The \emph{ABox} (dataset) consists of a finite number of \emph{concept assertions} of the form $A(a)$ and \emph{role assertions} of the form $R(a,b)$, where $A \in \cn$,  $R \in \rn$, $a,b \in \inds$. 
The \emph{TBox} (ontology) consists of a set of axioms whose form depends on the DL in question. 
In \dllitecore, TBox axioms are \emph{concept inclusions} $B\sqsubseteq C$ built according to the following grammar, where $A\in \cn$ and $R\in \rn$:
$$ B := A \mid \exists S,\quad  C:= B \mid \neg B,\quad S:= R \mid R^-.$$

\dlliterhorn extends \dllitecore with \emph{role inclusions} of the form $S\sqsubseteq Q$ where $Q:= S \mid \neg S$ and concept inclusions of the form $B_1\sqcap\dots\sqcap B_n\sqsubseteq C$.

\mypar{Semantics} An \emph{interpretation} has the form $\Imc = (\Delta^{\Imc}, \cdot^{\Imc})$, 
where $\Delta^{\Imc}$ is a non-empty set and $\cdot^{\Imc}$ 
maps each $A \in \cn$ to $A^{\Imc} \subseteq \Delta^{\Imc}$, 
each $R \in \rn$ to $R^{\Imc} \subseteq \Delta^{\Imc} \times \Delta^{\Imc}$, and each $a \in \inds$ to $a^{\Imc} \in \Delta^{\Imc}$. 
The function $\cdot^{\Imc}$ is straightforwardly extended to general concepts and roles,
e.g. $(\neg B_1)^{\Imc}= \Delta^{\Imc}\setminus B_1^{\Imc}$, $(B_1\sqcap B_2)^{\Imc}= B_1^{\Imc}\cap B_2^\Imc$, $(R^{-})^{\Imc}=\{(c,d) \mid (d,c) \in R^{\Imc}\}$ and $(\exists Q)^{\Imc}=\{c \mid \exists d:\, (c,d) \in Q^{\Imc}\}$. 
An interpretation $\Imc$ satisfies 
$G \sqsubseteq H$ if $G^{\Imc} \subseteq H^{\Imc}$; it satisfies $A(a)$ (resp. $R(a,b)$) if
$a^{\Imc} \in A^{\Imc}$ (resp. $(a^{\Imc},b^{\Imc})\in R^{\Imc}$). 
 We call $\Imc$ a \emph{model} of  $\Kmc=\tup{\Tmc,\Amc}$ if $\Imc$ satisfies all  axioms
in $\Tmc$ and assertions in \Amc. 
A KB \Kmc\ is \emph{consistent} if it has a model; otherwise it is inconsistent.
We say that an ABox $\Amc$ is \emph{$\Tmc$-consistent} if the KB $\Kmc=\tup{\Tmc,\Amc}$ is consistent.

\mypar{Queries}
A \emph{first-order (FO) query} is a first-order logic formula whose atoms are built using the predicate symbols in $\cn \cup \rn$ and constants in $\inds$. 
We will focus on the subclass of 
\emph{conjunctive queries} (CQs) which take the form 
$q(\vec{x})=\exists \vec{y} \, \psi(\vec{x},\vec{y})$, where  
$\psi$ is a conjunction of atoms of the forms $A(t)$ or $R(t,t')$ whose terms are either variables from $\vec{x} \cup \vec{y}$ or individuals. 
When a CQ consists of a single atom, we called it an instance query (IQ). 
A query without free variables is called \emph{Boolean}, and we refer to Boolean CQs as BCQs. 
A Boolean query $q$ is \emph{satisfied} by an interpretation~$\Imc$, written $\Imc \models q$, 
iff $q$ is satisfied in~$\Imc$ according to standard first-order logic semantics; $q$ is \emph{entailed} by a KB $\Kmc$, written $\Kmc \models q$, iff $\Imc\models q$ for every model~$\Imc$ of $\Kmc$. 
Henceforth, when we speak of a (Boolean) query, without specifying the type of query, we mean a (Boolean) CQ.

\section{Querying Inconsistent Prioritized KBs}\label{sec:reps}
An inconsistent KB entails every Boolean query under the standard semantics. 
Several inconsistency-tolerant semantics have been defined to obtain meaningful answers in this context (see \cite{DBLP:conf/rweb/BienvenuB16} for a survey). 
Most of these semantics are based on the notion of \emph{repairs}, which correspond to the different ways of restoring consistency by minimally removing some assertions.

\begin{definition}[Repair]
An \emph{ABox repair} of a KB $\Kmc=\tup{\Tmc,\Amc}$, or \emph{repair} for short, is an inclusion-maximal $\Tmc$-consistent subset of $\Amc$. The set of 
repairs of $\Kmc$ is denoted by $\reps{\Kmc}$.
\end{definition}

Another central notion is that of \emph{conflicts}, which are the minimal subsets of assertions that contradict the TBox. 
\begin{definition}[Conflict]
A \emph{conflict} of a KB $\Kmc=\tup{\Tmc,\Amc}$, is a minimal $\Tmc$-inconsistent subset of $\Amc$. The set of all conflicts of $\Kmc$ is denoted by $\conflicts{\Kmc}$ and can be represented as a \emph{conflict hypergraph} whose vertices are assertions and whose hyperedges are the conflicts.
\end{definition}

To simplify the presentation, we assume throughout the paper that \emph{the ABox does not contain any assertion that is inconsistent with the TBox}, \ie every conflict contains at least two assertions. Note that self-contradictory assertions do not occur in any repair, and such assertions can be readily identified and removed using existing reasoning algorithms.

In the database setting, \citeauthor{DBLP:journals/amai/StaworkoCM12} \shortcite{DBLP:journals/amai/StaworkoCM12} defined three notions of \emph{optimal repairs} based on a \emph{priority relation} over facts, which expresses preferences between conflicting facts. 
Such a relation could e.g.\ represent expert knowledge of how to resolve certain kinds of conflicts or 
be learned from manual conflict resolution \cite{DBLP:journals/ijar/MartinezPPSS14,DBLP:conf/www/TanonBS19}. 
We recast these notions in the KB setting.

\begin{definition}[Priority relation, completion]
Given a KB $\Kmc=\tup{\Tmc,\Amc}$, a \emph{priority relation} $\succ$ for $\Kmc$ is an acyclic binary relation over the assertions of $\Amc$ such that if $\alpha\succ\beta$, then there exists $C\in\conflicts{\Kmc}$ such that $\{\alpha,\beta\}\subseteq C$. 
A priority relation $\succ$ is \emph{total} iff for all $\alpha\neq\beta$, if there exists $C\in\conflicts{\Kmc}$ such that $\{\alpha,\beta\}\subseteq C$, then $\alpha\succ\beta$ or $\beta\succ\alpha$. 
A \emph{completion} of $\succ$ is a total priority relation $\succ'\ \supseteq \ \,\succ$.
\end{definition}

\begin{definition}[Prioritized KB]
A \emph{prioritized KB} is a pair $(\Kmc, \succ)$ consisting of a KB $\Kmc$ and priority relation $\succ$ for $\Kmc$.
The notation $\Kmc_\succ$ will be used as shorthand for $(\Kmc, \succ)$.   
\end{definition}

\begin{definition}[Optimal repairs]
Consider a prioritized KB $\Kmc_\succ$ with $\Kmc=\tup{\Tmc,\Amc}$, and let $\Amc' \in \reps{\Kmc}$. 
\begin{itemize}
\item A \emph{Pareto improvement} of $\Amc'$ (\wrt $\succ$) is  a $\Tmc$-consistent $\Amc''\subseteq\Amc$ such that there exists $\beta\in\Amc''\setminus\Amc'$ such that $\beta\succ\alpha$ for every $\alpha\in\Amc'\setminus\Amc''$.
\item A \emph{global improvement} of $\Amc'$ (\wrt $\succ$) is a $\Tmc$-consistent $\Amc''\subseteq\Amc$ such that $\Amc''\neq\Amc'$ and for every $\alpha\in\Amc'\setminus\Amc''$, there exists $\beta\in\Amc''\setminus\Amc'$ such that $\beta\succ\alpha$.
\end{itemize}
The repair $\Amc'$ is a:
\begin{itemize}
\item \emph{Pareto-optimal repair of $\Kmc_\succ$} if there is no Pareto improvement of $\Amc'$ w.r.t.\ $\succ$. 
\item \emph{Globally-optimal repair of $\Kmc_\succ$} if there is no global improvement of $\Amc'$ w.r.t.\ $\succ$.
\item \emph{Completion-optimal repair of $\Kmc_\succ$} 
if $\Amc'$ is a globally-optimal repair of $\Kmc_{\succ'}$, for some completion $\succ'$ of $\succ$.
\end{itemize}
We denote by $\greps{\Kmc_\succ}$, $\preps{\Kmc_\succ}$ and $\creps{\Kmc_\succ}$ the sets of globally-, Pareto-, and completion-optimal repairs.
\end{definition}

\citeauthor{DBLP:journals/amai/StaworkoCM12} \shortcite{DBLP:journals/amai/StaworkoCM12} showed the following relation between optimal repairs: $$\creps{\Kmc_\succ}\subseteq \greps{\Kmc_\succ}\subseteq \preps{\Kmc_\succ}\subseteq \reps{\Kmc}.$$
They also show that $\Amc'$ is a completion-optimal repair iff it can be obtained by the following greedy procedure:
while some assertion has not been considered, 
pick an assertion 
that is maximal \wrt $\succ$ among those not yet considered, 
and add it to the current set if it does not lead to a contradiction.

\begin{example}\label{ex:prio-repairs}
Consider the prioritized KB $\Kmc_\succ$ given on Figure \ref{fig:confgraph}. 
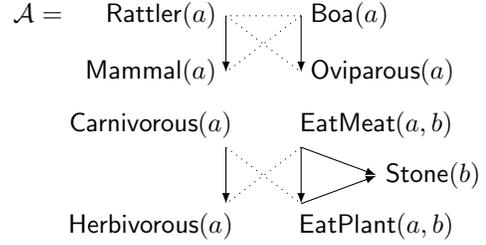
\begin{figure*}
\begin{subfigure}[b]{0.5\textwidth}
\begin{align*}
\Tmc=\{& 
\mn{Boa}\sqsubseteq\mn{Snake}, \,
\mn{Boa}\sqsubseteq\mn{Ovoviviparous},
\\&
\mn{Rattler}\sqsubseteq\mn{Snake}, \,
\mn{Rattler}\sqsubseteq\mn{Ovoviviparous},
\\&
\mn{EatPlant}\sqsubseteq\mn{Eat},  \,
\mn{EatMeat}\sqsubseteq\mn{Eat}, 
\\&
\mn{Mammal}\sqsubseteq\neg\mn{Snake}, \,
\mn{Oviparous}\sqsubseteq\neg \mn{Ovoviviparous},\\&
\mn{Boa}\sqsubseteq\neg\mn{Rattler}, \,
\mn{Carnivorous}\sqsubseteq\neg\mn{Herbivorous},\\&
\mn{Carnivorous}\sqsubseteq\neg\exists\mn{EatPlant}, \,
\mn{Herbivorous}\sqsubseteq\neg\exists\mn{EatMeat}, 
\\&
\exists\mn{EatPlant}^-\sqsubseteq\neg \exists\mn{EatMeat}^-, \,
\exists\mn{Eat}^-\sqsubseteq \neg \mn{Stone}
\}
\end{align*}
\end{subfigure}
\begin{subfigure}[t]{0.5\textwidth}
\centering
\begin{tikzpicture} 
\node [above] at (-2.5,0.55) {$\Amc=$};
\node [left] at (0,0.75) {$\mn{Rattler}(a)$};
\node [right] at (1,0.75) {$\mn{Boa}(a)$};
\node [left] at (0,0) {$\mn{Mammal}(a)$};
\node [right] at (1,0) {$\mn{Oviparous}(a)$};

\draw[<-] (0,0) -- (0,0.75);
\draw[->] (1,0.75) -- (1,0);
\draw[dotted] (0,0.75) -- (1,0);
\draw[dotted] (1,0.75) -- (0,0.75);
\draw[dotted] (1,0.75) -- (0,0);

\node [above] at (-1,-1) {$\mn{Carnivorous}(a)$};
\node [above] at (2,-1) {$\mn{EatMeat}(a,b)$};
\node [below] at (-1,-1.75) {$\mn{Herbivorous}(a)$};
\node [below] at (2,-1.75) {$\mn{EatPlant}(a,b)$};
\node [right] at (2,-1.37) {$\mn{Stone}(b)$};

\draw[->] (0,-1) -- (0,-1.75);
\draw[->] (1,-1) -- (1,-1.75);
\draw[dotted] (0,-1) -- (1,-1.75);
\draw[dotted] (1,-1) -- (0,-1.75);
\draw[->] (1,-1) -- (2,-1.37);
\draw[->] (1,-1.75) -- (2,-1.38);
\end{tikzpicture}
\end{subfigure}
\caption{Prioritized KB example. $\Tmc$ formalizes knowledge about animal groups and $\Amc$ gives information about individuals $a$ and $b$. An arrow from $\alpha$ to $\beta$ indicates that $\alpha\succ\beta$ while a dotted line indicates a conflict without priority between assertions.}\label{fig:confgraph}
\end{figure*}
The sets of optimal repairs of $\Kmc_\succ$ are as follows.
\begin{align*}
\creps{\Kmc_\succ\!}\!=&\{\!
\{\mn{Rattler}(a), \mn{Carnivorous}(a), \mn{EatMeat}(a,b)\!\},\\&
\phantom{\{}\{\mn{Boa}(a), \mn{Carnivorous}(a), \mn{EatMeat}(a,b)\}\!
\}\\
\greps{\Kmc_\succ\!}\!=&\creps{\Kmc_\succ}\cup\{
\{
\mn{Mammal}(a),
 \mn{Oviparous}(a), 
 \\&
\quad\quad\quad\quad 
\mn{Carnivorous}(a), \mn{EatMeat}(a,b)\}
\}\\
\preps{\Kmc_\succ\!}\!=&\greps{\Kmc_\succ}\cup\{\\&
\{\mn{Rattler}(a), \mn{Herbivorous}(a), \mn{EatPlant}(a,b)\},\\&
\{\mn{Boa}(a), \mn{Herbivorous}(a), \mn{EatPlant}(a,b)\},\\&
\{
\mn{Mammal}(a),
 \mn{Oviparous}(a), \\&
 \phantom{\{}\mn{Herbivorous}(a), \mn{EatPlant}(a,b)\}\!
\}
\end{align*}
$\reps{\Kmc}$ additionally contains repairs with $\mn{Stone}(b)$ instead of the $\mn{EatPlant}(a,b)$ or $\mn{EatMeat}(a,b)$ assertions.
\end{example}

We next introduce the three natural inconsistency-tolerant semantics based upon repairs. 
The AR semantics is arguably the most natural and well-known semantics and requires that a query be entailed from 
 \emph{every repair} of the KB. The IAR semantics is a more cautious semantics that evaluates queries over the \emph{intersection of all repairs}. Finally, the brave semantics is the most adventurous semantics, returning yes to queries that are entailed by \emph{some repair}. 

\begin{definition}[AR, IAR and brave semantics]
For $\text{X}\in\{\text{C},\text{G},\text{P},\text{S}\}$, a prioritized KB $\Kmc_\succ$ entails a query $q$ under 
\begin{itemize}
	\item\emph{X-AR semantics}, 
	written $\Kmc_\succ \armodels{X} q$,
	 \begin{flushright} iff $\tup{\Tmc,\Amc'}\models q$ for \emph{every } $\Amc'\in\xreps{\Kmc_\succ}$;\end{flushright} 	
	\item\emph{X-IAR semantics}, written $\Kmc_\succ \iarmodels{X} q$,\begin{flushright} 
	iff $\tup{\Tmc,\Amc^{\cap}} \models q$ where $\Amc^{\cap}=\bigcap_{\Amc'\in\xreps{\Kmc_\succ}}\Amc'$;\end{flushright}
	\item\emph{X-brave semantics}, written $\Kmc_\succ \bravemodels{X} q$,\begin{flushright}
	 iff $\tup{\Tmc,\Amc' }\models q$ for \emph{some } $\Amc'\in\xreps{\Kmc_\succ}$. 	\end{flushright}
\end{itemize}
\end{definition}

We remark that these semantics are related as follows:
$$\Kmc_\succ \iarmodels{X} q \Rightarrow \Kmc_\succ \armodels{X} q  \Rightarrow \Kmc_\succ \bravemodels{X} q.$$ 


\begin{example}[Ex. \ref{ex:prio-repairs} cont'd]
Considering the different semantics based upon completion-optimal repairs, we observe that 
\begin{itemize}
\item $\Kmc_\succ\iarmodels{C}\mn{Carnivorous}(a)$, 
\item $\Kmc_\succ\armodels{C}\mn{Snake}(a)$ but $\Kmc_\succ\not\iarmodels{C}\mn{Snake}(a)$, 
\item $\Kmc_\succ\bravemodels{C}\mn{Boa}(a)$ but $\Kmc_\succ\not\armodels{C}\mn{Boa}(a)$.
\end{itemize}
If we consider now AR semantics for the different kinds of optimal repairs, we find that
\begin{itemize}
\item $\Kmc_\succ\armodels{C}\mn{Snake}(a)$ but $\Kmc_\succ\not\armodels{G}\mn{Snake}(a)$,
\item $\Kmc_\succ\armodels{G}\mn{Carnivorous}(a)$ but $\Kmc_\succ\not\armodels{P}\mn{Carnivorous}(a)$,
\item $\Kmc_\succ\armodels{P}\mn{Eat}(a,b)$ but $\Kmc_\succ\not\armodels{S}\mn{Eat}(a,b)$.
\end{itemize}
\end{example}

\section{Complexity Analysis}\label{sec:complexity}

In this section, we analyze the data complexity of the central reasoning tasks related to optimal repairs. 

\subsection{Overview}
We define the problems and complexity classes we consider as well as some relevant properties of DL languages.

\mypar{Problems} 
Given a prioritized KB $\Kmc_\succ$ and a type of repairs $\text{X}\in\{\text{C},\text{G},\text{P},\text{S}\}$,  
\emph{repair checking} ({\sc isRep}) is deciding whether a set of assertions belongs to $\xreps{\Kmc_\succ}$;
 \emph{AR (resp. IAR, brave) entailment} ({\sc AR, IAR, brave}) is deciding whether $\Kmc_\succ$ entails a given BCQ under X-AR (resp. X-IAR, X-brave) semantics; 
 \emph{uniqueness} ({\sc unique}) is deciding whether $|\xreps{\Kmc_\succ}|=1$; 
and \emph{enumeration} ({\sc enum}) is enumerating all elements of $\mi{XRep}(\Kmc_\succ)$.

\mypar{Complexity classes} All of the complexity results stated in this paper are for the \emph{data complexity} measure, 
i.e., the complexity is with respect to the size of $\Amc$, with $\Tmc$ and $q$ treated as fixed (hence of constant size).

For decision problems, we will refer to the following complexity classes: 
 \aczero is the class of problems that can be solved by a uniform family of cicuits of constant depth and polynomial-size, with unbounded-fanin AND and OR gates; 
 \ptime is the class of problems solvable in polynomial time;
 \np is the class of problems solvable in non-deterministic polynomial time; 
\sigmaptwo is the class of problems solvable in non-deterministic polynomial time with access to an \np oracle;
\conp (resp. \piptwo) is the class of problems whose complement is in~\np (resp. \sigmaptwo). 
These classes are related as follow: $\aczero\subset\ptime\subseteq\np\subseteq\sigmaptwo$ and $\aczero\subset\ptime\subseteq \text{\conp} \subseteq\piptwo$. 
It is widely believed that all of the inclusions are proper. 

For enumeration problems, we will use the following classes: 
\totalp is the class of enumeration problems that can be solved in polynomial time in $n+m$ where $n$ is the input size and $m$ the output size; 
 \delayp is the class of enumeration problems that can be solved by an algorithm such that the delay between the $k^{th}$ and the $k+1^{th}$ solution is bounded by a polynomial in the input size; 
\incrdelayp is the class of enumeration problems that can be solved by an algorithm such that the delay between the $k^{th}$ and the $k+1^{th}$ solution is bounded by a polynomial in the size of the input and $k$. 
These classes are related as follow: $\delayp\subseteq \incrdelayp\subseteq \totalp$. If the decision problem associated to an enumeration problem is \np-hard, then the enumeration problem is not in $\totalp$.

\mypar{Results} Table~\ref{tab:complexity} presents an overview of the complexity results for \dllite. However, since most of the results hold for some DLs outside the \dllite family, we will state them for DLs that satisfy some of the following properties: 
\begin{itemize}
\item \boundedconf : the size of conflicts is bounded;
\item \binaryconf : conflicts have size at most two;
\item \polyCQ : polynomial time BCQ entailment; 
\item \polycons : polynomial time consistency checking. 
\end{itemize}
The most commonly considered DL-Lite dialects, such as \dllitecore\ and \dlliter\ (the basis for OWL 2 QL), satisfy all of these properties, while \dlliterhorn satisfies all except \binaryconf (the size of conflicts is bounded by the number of concept and role names occurring in the TBox).

\begin{table}[t]
\setlength{\tabcolsep}{1pt}
\begin{tabular}{lcccc}
\toprule
& $\reps{\Kmc}$ & $\preps{\Kmc_\succ}$&$\greps{\Kmc_\succ}$ & $\creps{\Kmc_\succ}$\\
\midrule
\sc{isRep}&in \ptime & in \ptime& \conp-c & in \ptime\\
\sc{AR} &\conp-c  &  \conp-c &\piptwo-c & \conp-c\\
\sc{IAR} &  in \aczero & \conp-c &\piptwo-c& \conp-c\\
\sc{brave} & in \aczero  & \np-c &\sigmaptwo-c& \np-c\\
\sc{unique}&in \ptime & \conp-c  & \piptwo-c$^\ddagger$& in \ptime \\
\sc{enum} &$\delayp^{\ast}$  & not $\totalp^\dagger$ &not $\totalp^\dagger$  &$\delayp^{\ast}$ \\
\bottomrule
\end{tabular}
\caption{Data complexity overview. Lower bounds hold for \dllitecore.  Upper bounds hold for \dlliterhorn except those noted with $^\ast$ which hold only for \dllitecore. Results for the $\reps{\Kmc}$ case are stated (or are straightforward consequences of results) in 
\protect\cite{DBLP:conf/rweb/BienvenuB16,DBLP:conf/ijcai/BienvenuR13,DBLP:conf/pods/LivshitsK17}.
$^\ddagger$ in \ptime if $\succ$ is transitive 
$^\dagger$ \incrdelayp if $\succ$ is score-structured 
$^\ast$ \incrdelayp for \dlliterhorn KBs with a score-structured priority relation $\succ$ (see Section \ref{sec:trans} for the definition)} 
\label{tab:complexity}
\end{table}
\subsection{Query Entailment}
We start by establishing the complexity of repair checking and BCQ entailment under the various semantics.

The upper bounds are stated in the following theorem and can be proven by adapting existing results from databases and inconsistency-tolerant OMQA
(recall that 
proofs of all results are provided in the appendix). 
Observe that this result applies to 
all DLs satisfying \polycons and \polyCQ, which includes prominent Horn DLs like $\mathcal{EL}$ and Horn-$\mathcal{SHIQ}$. 

\begin{restatable}{theorem}{PropComplexityRepairCheckEntailment}\label{Prop:complexityRepairCheckEntailment}
For DLs satisfying \polycons, 
repair checking is in \conp for globally-optimal repairs,  in \ptime for Pareto- or completion-optimal repairs. 

For DLs satisfying \polycons and \polyCQ, 
G-AR and G-IAR entailment are in \piptwo, G-brave entailment in \sigmaptwo, 
P-AR, P-IAR, C-AR and C-IAR entailment are in \conp, and P-brave and C-brave entailment in~\np. 
\end{restatable}
For Pareto-optimal and completion-optimal repairs, \conp-hardness (resp.\ \np-hardness) of AR and IAR (resp.\ brave) IQ entailment in \dllitecore follow from the special case where the priority relation is given by priority levels, for which it is known the three families of optimal repairs we consider coincide \cite{DBLP:phd/hal/Bourgaux16}. 
The following theorem establishes the remaining lower bounds.

\begin{restatable}{theorem}{ThmHardnessGlobal}\label{ThmHardnessGlobal}
In \dllitecore, repair checking is \conp-hard for globally-optimal repairs, 
G-AR and G-IAR entailment are \piptwo-hard, and G-brave entailment \sigmaptwo-hard, even for IQs. 
\end{restatable}

\subsection{Uniqueness and Enumeration}
We now turn our attention to the problems of deciding uniqueness and enumerating repairs. Note that in the case of classical repairs, deciding whether $|\reps{\Kmc}|=1$ amounts to checking whether $\Kmc$ is consistent (recall our assumption that all self-contradictory assertions have been removed). 

As observed by \citeauthor{DBLP:conf/pods/LivshitsK17} \shortcite{DBLP:conf/pods/LivshitsK17} in the database setting,  
classical repairs correspond to the maximal independent sets (MISs) of the conflict hypergraph. 
 It is known that MIS enumeration for graphs is in $\delayp$ \cite{DBLP:journals/ipl/JohnsonP88}, while for  
hypergraphs with bounded-size hyperedges, the problem is in $\incrdelayp$ \cite{DBLP:journals/siamcomp/EiterG95,DBLP:journals/ppl/BorosEGK00}.  
Since the conflict hypergraph can be tractably computed 
for DLs with bounded-size conflicts, 
enumeration of $\reps{\Kmc}$ 
 is in \incrdelayp for DLs with bounded conflicts, and in \delayp\ for DLs with binary conflicts. 

\citeauthor{DBLP:conf/icdt/KimelfeldLP17} \shortcite{DBLP:conf/icdt/KimelfeldLP17} and \citeauthor{DBLP:conf/pods/LivshitsK17} \shortcite{DBLP:conf/pods/LivshitsK17} provide intractability results for \textsc{unique} and \textsc{enum}
in the case of globally- and Pareto-optimal repairs 
of databases with functional dependencies. Their proofs do not 
transfer to \dllite KBs, but via different proofs, we can establish analogous results in our context.

\begin{restatable}{theorem}{ThmUniqueness}\label{ThmUniqueness}
Uniqueness is \conp-complete for Pareto-optimal repairs, \piptwo-complete for globally-optimal repairs. Upper bounds hold for DLs satisfying \polycons. 
Lower bounds hold for \dllitecore.
\end{restatable}

\begin{corollary}
Enumeration of Pareto-optimal or globally-optimal repairs is not in \totalp for \dllitecore.
\end{corollary}

\citeauthor{DBLP:conf/icdt/KimelfeldLP17} \shortcite{DBLP:conf/icdt/KimelfeldLP17}
provide an algorithm which, given an arbitrary conflict hypergraph and a priority relation, decides uniqueness for completion-optimal repairs in polynomial time. Hence deciding whether $|\creps{\Kmc_\succ}|=1$ is in \ptime for DLs with bounded conflicts, since in this case the conflict hypergraph can be computed in polynomial time. 
\citeauthor{DBLP:conf/pods/LivshitsK17} \shortcite{DBLP:conf/pods/LivshitsK17} describe an algorithm that enumerates completion-optimal repairs with polynomial delay for \emph{conflict graphs}. It follows that for DLs with binary conflicts, enumeration of completion-optimal repairs is in \delayp. 
The case of completion-optimal repairs with non-binary conflicts remains open.

\subsection{Case of Transitive Priority Relations} \label{sec:trans}
A natural case to consider is when the priority relation is \emph{transitive}. This occurs in particular when the priority relation captures the relative reliability of the facts.

\begin{definition}
A priority relation is \emph{transitive} if $\alpha_1\succ\dots\succ\alpha_n$ and $\{\alpha_1,\alpha_n\}\subseteq C\in\conflicts{\Kmc}$ implies $\alpha_1\succ\alpha_n$.
\end{definition}

\citeauthor{DBLP:conf/icdt/KimelfeldLP17} \shortcite{DBLP:conf/icdt/KimelfeldLP17} proved that when the priority relation is transitive, deciding uniqueness for globally-optimal repairs can be done in polynomial time when the conflict hypergraph is given. 
This result is applicable to all DLs for which the conflict hypergraph is computable in polynomial time. 
A natural question is whether any of our other problems become easier if we assume $\succ$ is transitive. 
We found that this is not the case, although the transitivity requirement makes the lower bounds proofs more involved. 
In particular, we can no longer use the \conp-hardness proof for uniqueness to derive that 
the enumeration of globally-optimal repairs is not in \totalp, but instead show this via the 
\conp-hardness of 
deciding whether a given set of ABoxes 
 is exactly $\greps{\Kmc_\succ}$.

\begin{restatable}{theorem}{ThmTransitivity}
When $\succ$ is transitive, uniqueness is in \ptime for globally-optimal repairs and DLs satisfying \boundedconf. 
All other lower bounds in Table \ref{tab:complexity} still hold for \dllitecore. 
\end{restatable}
\mypar{Score-structured priority relations} 
An interesting special case of transitive priority relation 
arises when every assertion $\alpha$ is assigned a natural number called its \emph{score} (or \emph{priority level}) $s(\alpha)$, which induces a priority relation as follows: for every pair of assertions $\alpha$ and $\beta$ that belongs to a conflict, $\alpha\succ\beta$ if and only if $s(\alpha)>s(\beta)$. 
\citeauthor{DBLP:conf/pods/LivshitsK17}~\shortcite{DBLP:conf/pods/LivshitsK17} call such a priority relation \emph{score-structured} and show that deciding whether a priority is score-structured is in \ptime and that if $\succ$ is score-structured, $\creps{\Kmc_\succ}= \greps{\Kmc_\succ}= \preps{\Kmc_\succ}$. 
This relation was also given in 
 \cite{DBLP:conf/aaai/BienvenuBG14,DBLP:phd/hal/Bourgaux16} where it is shown that they further coincide 
 with another notion of preferred repairs directly defined from ABoxes partitioned into priority levels. 
 
For score-structured priorities, the complexity of all decision problems is the same 
as in the case of completion-optimal repairs, with the lower bounds shown in \cite{DBLP:phd/hal/Bourgaux16}. 
By contrast, for enumeration, it follows from the algorithm for enumerating score-structured priority repairs of 
\citeauthor{DBLP:conf/pods/LivshitsK17} that enumeration is in \incrdelayp for DLs with bounded conflicts and consistency checking in \ptime. 

\section{Preference-Based Argumentation}\label{sec:arg}
In preparation for the following section, we 
recall the basics of argumentation frameworks and 
previously studied extensions with preferences and collective attacks. 
We also prove some new results 
(Theorems \ref{thm:symm-coherent}, \ref{thm:SymPAFChar}, \ref{thm:StrongSymmSETAF}, and \ref{psetaf-coh}).

\subsection{Argumentation Frameworks}

We consider \emph{finite argumentation frameworks}. Argumentation frameworks over an infinite set of arguments have also been studied in the literature but finiteness is an ordinary assumption and is made for all results we recall or extend here.

\begin{definition}
An \emph{argumentation framework} (AF) is a pair $(\args,\defeat)$ where $\args$ is
a finite set of arguments and $\defeat\, \subseteq \args \times \args$ is the attack relation.
When $(\alpha, \beta)\! \in \defeat$, we say that 
$\alpha$ \emph{attacks} $\beta$, alternatively denoted by 
$\alpha \defeat \beta$.\end{definition}

We recall some commonly used notation and terminology, letting  
 $(\args , \defeat )$ be the considered AF and $A \subseteq \args$.
We use $A^+=\{\beta \mid \alpha \defeat \beta \text{ for some }\alpha \in A\}$ 
to denote the set of arguments attacked by arguments from $A$.   
We say that $A$ \emph{defends} $\gamma\in \args$ (or, $\gamma$ is \emph{defended by} $A$) iff $\{\beta \mid \beta \defeat \gamma\} \subseteq A^+ $. A set $A \subseteq \args$ is \emph{conflict-free} if $A \cap A^+ = \emptyset$. 
The \emph{characteristic function} $\chff:2^{\args} \mapsto 2^{\args}$ of $F=(\args , \defeat )$ 
is defined as follows: $\chff(A) = \{\alpha \mid \alpha \text{ is defended by } A\}$. 
A set $A \subseteq \args$ is \emph{admissible} if it is conflict-free and $A \subseteq \chff(A)$. 

Argumentation semantics are usually based upon \emph{extensions}, intended to capture coherent sets of arguments. 
We recall some prominent notions of extension:

\begin{definition}\label{def:extensions}
Let $F=(\args , \defeat )$.   
Then $E \subseteq \args$ is a: 
\begin{itemize}
\item \emph{complete extension} iff $E$ is conflict-free and $E=\chff(E)$;
\item \emph{grounded extension} iff $E$ is the $\subseteq$-minimal complete extension, or equivalently, 
the least fixpoint of $\chff$; 
\item \emph{preferred extension} iff $E$ is a $\subseteq$-maximal admissible set;
\item \emph{stable extension} iff $E^+ = \args\setminus E$.
\end{itemize}
\end{definition}

While every stable extension is a preferred extension, the converse does not hold in general \cite{DBLP:journals/ai/Dung95}. 
The term \emph{coherent} designates AFs for which these two notions coincide (i.e.\ every preferred extension is stable).  
Coherence is viewed as a desirable property, and several sufficient conditions for coherence have been identified. In particular:

\begin{definition}
An AF $(\args , \defeat )$ is \emph{symmetric}
 iff  $\defeat$ is symmetric and irreflexive.  
\end{definition}

\begin{theorem} \cite{DBLP:conf/ecsqaru/Coste-MarquisDM05}
Every symmetric AF is coherent.
\end{theorem}

\subsection{Enriching AFs with Preferences}
There has been significant interest in extending AFs to allow for preferences between arguments \cite{DBLP:journals/jar/AmgoudC02,DBLP:conf/comma/KaciTV18}. 

\begin{definition}
A \emph{preference-based argumentation framework} (PAF) is a triple $(\args,\defeat, \succ)$, where $\args$ and $\defeat$ are as before, and $\succ$ is an acyclic binary relation over $\args$ (called the \emph{preference relation}). \emph{Symmetric PAFs} are obtained by requiring 
($\args, \defeat$) to be symmetric. 
\end{definition}

It is typical to assume that the preference relation $\succ$ is transitive. This is a reasonable assumption, but for the sake of generality, the preceding definition does not impose this. 

The standard way of defining the semantics of PAFs is via reduction to (plain) AFs:

\begin{definition}\label{red-paf}
Given a PAF $(\args,\defeat, \succ)$, the \emph{corresponding AF} is $(\args, \defeat_\succ)$, 
where 
$\alpha \defeat_\succ \beta$ iff $\alpha \defeat \beta$ and $\beta \not \succ \alpha$. A subset $E \subseteq \args$ is a \emph{stable (resp.\ preferred, grounded, complete) extension of a PAF} $(\args,\defeat, \succ)$ iff it is a stable (resp.\ preferred, grounded, complete) extension of the corresponding~AF.
\end{definition}

While alternative reductions 
have been proposed in \cite{DBLP:journals/ijar/AmgoudV14} and \cite{DBLP:conf/comma/KaciTV18}, all of these reductions 
coincide for symmetric PAFs, which is the case that will interest us here.

\begin{theorem}\cite{DBLP:journals/ijar/AmgoudV14}\label{thm:symm-trans-coherent}
Every symmetric PAF with a transitive preference relation is coherent. 
\end{theorem}

We generalize the preceding result by showing that transitivity is not required:

\begin{restatable}{theorem}{thmSymmCoherent}\label{thm:symm-coherent}
Every symmetric PAF is coherent. 
\end{restatable}

Theorem \ref{thm:SymPAFChar} gives a characterization of AFs that correspond to symmetric PAFs. 
It generalizes the \emph{strict acyclicity} condition characterizing the so-called \emph{conflict+preference} AFs \cite{DBLP:conf/ecai/KaciTW06} that requires for any cycle $\alpha_1 \defeat \alpha_2\defeat \dots\defeat \alpha_n\defeat \alpha_1$ that $\alpha_1 \defeat \alpha_n\defeat \dots\defeat \alpha_2\defeat \alpha_1$. 

\begin{restatable}{theorem}{thmSymPAFChar}\label{thm:SymPAFChar}
For every AF $F=(\args, \defeat)$, the following statements are equivalent.
\begin{enumerate}
\item $F$ is the corresponding AF of some symmetric PAF.
\item For any cycle $\alpha_1 \defeat \alpha_2\defeat \dots\defeat \alpha_n\defeat \alpha_1$, there exists $(j,i)\in\{(1,2),\dots,(n-1,n),(n,1)\}$ with $\alpha_i\defeat\alpha_j$. 
\end{enumerate}
\end{restatable}

\subsection{Set-Based Argumentation Frameworks}

Another well-studied extension of AFs is to allow collective attacks,  
in which a set of arguments together form an attack 
\cite{DBLP:conf/argmas/NielsenP06,DBLP:journals/ijar/FlourisB19}.  

\begin{definition}
A \emph{set-based argumentation framework} (SETAF) is a pair $(\args , \defeat )$
where $\args$ is a finite set of arguments, and $\defeat\, \subseteq (2^{\args}\setminus \{\emptyset\}) \times \args$ is the attack relation. 
We'll write $S \defeat \alpha$ to mean $(S, \alpha)\! \in \defeat$. 
\end{definition}

When working with SETAFs, we define $A^+$ 
as $\{\beta \mid S \defeat \beta \text{ for some }S \subseteq A\}$
and say $A$ \emph{defends} $\beta $ iff $A^+\cap S \neq \emptyset$ whenever $S \defeat \beta$. 
The definitions of characteristic function, conflict-free, admissibility, different types of extensions, and coherence for SETAFs are the same as for AFs but using these modified notions of defeated and defended arguments.

Symmetric SETAFs were recently defined as follows: 

\begin{definition}\cite{symm-setaf}
A SETAF $(\args , \defeat)$ is \emph{symmetric} if the following two conditions hold:
\textbf{(Symm-1)} if $S \defeat \beta$ and $\alpha \in S$, then $S' \cup \{\beta\}\defeat \alpha$ for some $S'$, and 
\textbf{(Irr)} there is no $S \defeat \alpha$ with $\alpha \in S$. 
\end{definition}

Unfortunately, \citeauthor{symm-setaf} show their definition does not preserve
the nice properties of symmetric AFs. In particular:

\begin{theorem}
Some symmetric SETAFs are not coherent.  
\end{theorem}

We propose an alternative, and we believe quite natural, notion of symmetric SETAF, which we term  `strongly symmetric' to distinguish it from the preceding notion:

\begin{definition}
A SETAF $(\args , \defeat)$ is \emph{strongly symmetric} if 
it satisfies \textbf{(Irr)} and \textbf{(Symm-2)}:
 for every attack $S \defeat \beta$ and every $\alpha\in S$, we have $S\setminus\{\alpha\}\cup\{\beta\}\defeat \alpha$. 
\end{definition}

With our definition, symmetry implies coherence:

\begin{restatable}{theorem}{thmStrongSymmSETAF}\label{thm:StrongSymmSETAF}
Strongly symmetric SETAFs are coherent. 
\end{restatable}

\subsection{Adding Preferences to SETAFs}
It is natural to combine the preceding two extensions, 
but to the best of our knowledge, 
this combination has not yet been considered. We propose the following definition:

\begin{definition}
A \emph{preference-based set-based argumentation framework} (PSETAF) is a triple $(\args,\defeat, \succ)$, where 
$(\args,\defeat)$ is a SETAF, and $\succ$ is an acyclic binary relation over $\args$, called the preference relation.  
\end{definition}

To define the semantics of PSETAFs, we give a reduction to SETAFs, which generalizes the one for PAFs:

\begin{definition}
Given a PSETAF $(\args,\defeat, \succ)$, its corresponding SETAF is $(\args,\defeat_\succ)$, where 
the relation $\defeat_\succ \subseteq (2^{\args} \setminus \{\emptyset\}) \times \args$ is defined as follows:
$S \defeat_\succ \alpha$ iff $S \defeat \alpha$ and $\alpha \not \succ \beta$ for every $\beta \in S$. 
A subset $E \subseteq \args$ is a \emph{stable (resp.\ preferred, grounded, complete) extension of a PSETAF} $(\args,\defeat, \succ)$ iff it is a stable (resp.\ preferred, grounded, complete) extension of the corresponding SETAF.
\end{definition}

We will focus on strongly symmetric PSETAFs:

\begin{definition}
A PSETAF $(\args,\defeat, \succ)$ is \emph{strongly symmetric} if the SETAF $(\args,\defeat)$ is strongly symmetric. 
\end{definition}

The following result lifts Theorem \ref{thm:symm-trans-coherent} to PSETAFs. The proof is quite intricate, and we leave open 
whether the same result holds without the transitivity assumption:

\begin{restatable}{theorem}{thmPSETAF} \label{psetaf-coh}
Every strongly symmetric PSETAF with a transitive preference relation is coherent. 
\end{restatable}

\section{Linking Prioritized KBs and PSETAFs}\label{sec:repairs-extensions}
This section places optimal repairs into a broader context by 
exhibiting a tight relationship between optimal repairs
and extensions of argumentation frameworks.  The 
argumentation connection is then exploited to define a new kind of prioritized repair with favourable computational properties.

Each prioritized KB naturally gives rise to a PSETAF in which the ABox assertions 
play the role of arguments, and the attack relation captures the conflicts. 

\begin{definition}\label{translation}
Given a prioritized KB $\Kmc_\succ$ with 
$\Kmc=\tup{\Tmc,\Amc}$, 
the associated PSETAF is $F_{\Kmc, \succ}= (\Amc, \defeat_{\Kmc}, \succ)$,
where $\defeat_{\Kmc} = \{(C \setminus \{\alpha\}, \alpha) \mid C \in \conflicts{\Kmc}, \alpha \in C\}$. 
\end{definition}

\begin{remark}
Recall that we assume the ABoxes do not contain self-conflicting assertions, hence $C \setminus \{\alpha\} \neq \emptyset$. 
If we choose not to make this assumption, we should omit such assertions when constructing the PSETAF, i.e., 
using $\Amc \setminus \{\alpha \mid \tup{\Tmc, \{\alpha\}} \models \bot\}$ for the set of arguments
and only considering $C \in \conflicts{\Kmc}$ with $|C| \geq 2$ to construct $\defeat_{\Kmc} $. 
\end{remark}

In this section, results are formulated for arbitrary KBs, except where otherwise noted.

\subsection{Optimal Repairs vs. PSETAF Extensions}

Now that we have translated prioritized KBs into argumentation frameworks, 
it is possible to compare preferred repairs and extensions. 
The following 
theorem shows that Pareto-optimal repairs correspond precisely to stable extensions. 

\begin{restatable}{theorem}{thmstable} \label{tstable}
$\Amc' \subseteq \Amc$ is a Pareto-optimal repair of $\Kmc_\succ = (\tup{\Tmc, \Amc}, \succ)$ 
iff $\Amc'$ is a stable extension of 
$F_{\Kmc, \succ}$.
\end{restatable}

Observe that our translation always produces strongly symmetric PSETAFs, which makes it possible to transfer Theorem \ref{psetaf-coh} when
 the priority relation is transitive. 

\begin{restatable}{theorem}{thmTransCoherent}
\label{trans-coherent} 
If $\succ$ is a transitive priority relation, then
the PSETAF $F_{\Kmc, \succ}$ is coherent. 
\end{restatable}

Theorems \ref{tstable} and \ref{trans-coherent} together show that for transitive $\succ$,
Pareto-repairs coincide also with preferred extensions. 


%

For KBs which contain only binary conflicts (e.g.\  core DL-Lite dialects), the resulting PSETAF is actually a PAF,
so we can drop the transitivity requirement:

\begin{theorem}
For every prioritized KB $\Kmc_\succ$ such that $\Cmc \in \conflicts{\Kmc}$ implies $|\Cmc|=2$, 
$F_{\Kmc, \succ}$ is coherent. Thus,
$\Amc' \in \preps{\Kmc_\succ}$
iff $\Amc'$ is a preferred extension of $F_{\Kmc, \succ}$.
\end{theorem}

Note that it follows from Theorem \ref{tstable} that globally-optimal and completion-optimal repairs correspond to proper subsets of the stable extensions of $F_{\Kmc, \succ}$, but they do not at present have any analog in the argumentation setting. 
For a (SET)AF $F$ that corresponds to a (strongly) symmetric P(SET)AF $P$, it is possible to define new types of extensions that will correspond to globally- and completion-optimal repairs of a KB having $P$ for associated P(SET)AF and to import results from the prioritized KB setting. 
In particular, when $F$ is an AF that respects the conditions of Theorem \ref{thm:SymPAFChar}, we can adapt the algorithms for completion-optimal repairs to tractably enumerate a non-empty subset of the stable/preferred extensions of $F$, while enumerating its stable/preferred extensions is not in \totalp in general.

\subsection{Grounded Semantics for Prioritized KBs}

We next turn to the relationship between grounded extensions
and optimal repair semantics. We start by clarifying the situation  
for the simplest case, in which there are no preferences 
(cf.\ \cite{DBLP:conf/sum/CroitoruV13} for a similar result in a different but related setting).

\begin{restatable}{theorem}{thmGroundedNoPref}
Let $\Kmc=\tup{\Tmc, \Amc}$ be a KB, and let $F_\Kmc$ be the SETAF
corresponding to $F_{\Kmc, \succ_\emptyset}$ (with $\succ_\emptyset$ the empty relation).
Then the  
grounded extension of $F_\Kmc$
coincides with the intersection of the repairs of $\Kmc$. 
\end{restatable}

For prioritized KBs, we can use Theorem \ref{tstable} and the fact that every 
stable extension is a complete extension to relate 
grounded extensions with Pareto-optimal  
repairs: 

\begin{restatable}{theorem}{thmGroundedPareto}\label{ground-pareto}
If $\Gmc$ is the grounded extension of 
$F_{\Kmc, \succ}$,
then $\Gmc \subseteq \Bmc$ for every $\Bmc \in \preps{\Kmc_\succ}$. 
\end{restatable}

Thus, the grounded extension contains only assertions common to all
Pareto-optimal repairs. It may not contain all such assertions,
as the following example illustrates:

\begin{example}
Consider a prioritized KB $\Kmc_\succ$ in which $\conflicts{\Kmc}=\{\{\alpha,\beta\},\{\alpha,\gamma\},\{\beta,\gamma\},\{\gamma,\delta\}\}$ and where $\succ$ contains $\alpha\succ\gamma$, $\beta\succ\gamma$, and $\gamma\succ\delta$. Then $\delta$ appears in the two Pareto-optimal repairs of $\Kmc_\succ$, which are $\{\alpha,\delta\}$ and $\{\beta,\delta\}$, while the grounded extension of $F_{\Kmc, \succ}$ is $\emptyset$.

\end{example}

We propose to use the grounded extension to define a new
inconsistency-tolerant semantics for prioritized KBs: 

\begin{definition}
A query $q$  is entailed from $\Kmc_\succ$ under \emph{grounded semantics}, denoted $\Kmc_\succ \models_{GR} q$,
iff $\tup{\Tmc, \Gmc} \models q$, where $\Gmc$ is the grounded extension 
of $F_{\Kmc, \succ}$.
\end{definition}

As follows from Theorem \ref{ground-pareto}, our new semantics provides an under-approximation of the P-IAR semantics:

\begin{theorem}\label{ground-iar-pareto}
For every prioritized KB $\Kmc_\succ$ and query $q$: 
$\Kmc_\succ \models_{GR} q$ implies $\Kmc_\succ \iarmodels{P} q$. 
\end{theorem}

Importantly, since the grounded extension of a (SET)AF can be computed in polynomial time,
this semantics allows for tractable query answering in some relevant settings.

\begin{theorem}\label{thm:groundedUpper}
Let $\mathcal{L}$ be a DL satisfying \boundedconf, 
 \polyCQ and 
 \polycons. 
Then the problem of BCQ entailment over $\mathcal{L}$-KBs under grounded semantics 
is in \ptime\ w.r.t.\ data complexity. 
\end{theorem}

\begin{corollary}
For \dlliterhorn\ KBs, BCQ entailment under grounded semantics 
is in \ptime\ w.r.t.\ data complexity.
\end{corollary}

The grounded semantics can be compared to another recently proposed
semantics for prioritized KBs, the \emph{Elect semantics}, defined for KBs with binary conflicts by \citeauthor{DBLP:conf/lpnmr/BelabbesBC19} \shortcite{DBLP:conf/lpnmr/BelabbesBC19} then generalized for non-binary conflicts \cite{DBLP:conf/ksem/BelabbesB19}. 
Elect evaluates queries over a set of \emph{elected assertions} and is tractable for DLs for which computing the conflicts and BCQ entailment can be done in polynomial time \wrt data complexity. 

\begin{definition}\cite{DBLP:conf/ksem/BelabbesB19}
An assertion $\alpha\in\Amc$ is \emph{elected} iff for every $C\in\conflicts{\Kmc}$, if $\alpha\in C$, then there exists $\beta\in C$ such that $\alpha\succ\beta$. 
The set of elected assertions is denoted by $\mi{Elect}(\Kmc_\succ)$. \end{definition}

The grounded semantics is more productive than Elect, \ie is a more precise under-approximation of P-IAR.

\begin{restatable}{theorem}{thmElect}\label{thm:elect}
$\mi{Elect}(\Kmc_\succ)\subseteq \Gmc$ where $\Gmc$ is the grounded extension of $F_{\Kmc, \succ}$. 
\end{restatable}

\begin{example}\label{ex:Elect}
Consider a prioritized KB $\Kmc_\succ$ in which $\conflicts{\Kmc}=\{\{\alpha,\beta\},\{\beta,\gamma\}\}$ and where $\succ$ contains $\alpha\succ\beta$ and $\beta\succ\gamma$. 
Then $\mi{Elect}(\Kmc_\succ)=\{\alpha\}$ while $\Gmc=\{\alpha,\gamma\}$.
Indeed, $\alpha\in\Gmc$ and $\gamma$ is defended by $\{\alpha\}$ in the AF corresponding to $F_{\Kmc, \succ}$, so $\gamma\in \Gmc$.
\end{example}

Our next result establishes that reasoning under grounded semantics is \ptime-complete in the settings of Theorem \ref{thm:groundedUpper}.
\begin{restatable}{theorem}{thmGroundedHardness}
IQ entailment under grounded semantics is \ptime-hard \wrt data complexity for \dllitecore. 
\end{restatable}

Dung (\citeyear{DBLP:journals/ai/Dung95}) observed that the grounded extension can be computed via the well-founded semantics of a simple logic program.  
The following lemma generalizes this result to $k$-SETAFs, i.e., SETAFs for which $S \defeat \beta$ implies $|S| \leq k$. 

\begin{theorem}\label{lemma:accp}
Let $F=(\args,\defeat)$ be a $k$-SETAF. Then $\alpha$ is in the grounded extension of $F$ iff 
$\accp(\alpha)$ belongs to the well-founded model of the following normal logic program:
\begin{align*}
&\{\defp(x) \!\leftarrow\! \mathtt{att}_i\!(y_1\!,.., y_i,x), \! \accp(y_1),\!...,\!\accp(y_i)\! \mid \!1 \!\leq \! i \leq\! k\} \\
&\cup \{\accp(x) \leftarrow \mathtt{arg}(x), \nnot \defp(x)\} \cup \{\mathtt{arg}(\alpha) \mid \alpha \in \args\} \\
&\cup \{\mathtt{att}_i(\alpha_1, \ldots, \alpha_i, \beta) \leftarrow \, \mid \{\alpha_1, \ldots, \alpha_i\} \defeat \beta\}
\end{align*}
\end{theorem}

Intuitively, the preceding logic program computes the sets of defeated ($\defp$) and accepted ($\accp$) arguments starting from the original sets of arguments ($\mathtt{arg}$) 
and attacks ($\mathtt{att}_i$ encodes $i$-ary attacks).   
By adding 
rules to compute conflicts and populate the attack relations $\mathtt{att}_i$, 
we can show that the grounded extension of $F_{\Kmc, \succ}$ can be computed via logic programming. 
While stated for \dlliterhorn, the next theorem holds for any ontology (or constraint) language for which inconsistency can be characterized by a finite set of BCQs.

\begin{restatable}{theorem}{Wfhorn} \label{thm:wfgrnd}
For every \dlliterhorn\ TBox $\Tmc$, there exists a normal logic program $\Pi_\Tmc$ 
such that for every ABox $\Amc$, priority relation $\succ$ for $\Kmc=\tup{\Tmc, \Amc}$, and assertion $\alpha$, the following are equivalent:
\begin{itemize}
\item $\alpha$ belongs to the grounded extension of the PSETAF $F_{\Kmc, \succ}$
\item $\accp(id(\alpha))$ belongs to the well-founded model of $\Pi_\Tmc \cup \{\gamma \leftarrow \, \mid \gamma \in \Amc_{id}\}
\cup \{ \prefp(id(\alpha), id(\beta)) \leftarrow \, \mid \alpha \succ \beta\}$,
\end{itemize}
where $\Amc_{id}$ is obtained from $\Amc$ by adding an extra argument to every assertion $\beta$ containing a unique id, denoted $id(\beta)$. 
\end{restatable}

As the well-founded semantics is implemented in logic programming systems like XSB \cite{DBLP:conf/lpnmr/RaoSSWF97}, the preceding theorem presents a method for implementing grounded semantics. Alternatively, the grounded extension can be naturally under-approximated by fixing $d>0$ and considering $\chff^d(\emptyset)$ rather than the least fixpoint of $\chff$. The logic program from Theorem \ref{thm:wfgrnd} can be modified to yield a non-recursive stratified program that computes $\Gamma^d_{F_{\Kmc, \succ}}(\emptyset)$, which in turn can be expressed as a first-order ($\sim$ SQL) query and evaluated using a relational database system.
%

\section{Related Work}
Our complexity results are similar to those obtained for prioritized databases with functional dependencies \cite{DBLP:journals/amai/StaworkoCM12,DBLP:conf/icdt/KimelfeldLP17,DBLP:conf/pods/LivshitsK17}. 
Compared to this work, we additionally consider IAR and brave semantics,  study the impact of 
transitivity on all reasoning tasks, and establish connections with argumentation. 

Inconsistency-tolerant semantics based on other kinds of preferred repairs have been investigated both in the database and DL or Datalog$^\pm$ contexts, often with a focus on repairs that have maximal cardinality or weight
 \cite{DBLP:conf/icdt/LopatenkoB07,DBLP:journals/kais/DuQS13,DBLP:conf/jelia/BagetBBCMPRT16,DBLP:conf/aaai/LukasiewiczMV19}. 
However, such global optimality criteria lead to a higher computational complexity (typically $\Delta^P_2[O(log\ n)]$-hard). 
By contrast, \citeauthor{DBLP:conf/aaai/BienvenuBG14} \shortcite{DBLP:conf/aaai/BienvenuBG14} show that for preferred repairs based on priority levels and set-inclusion (corresponding to optimal repairs in the score-structured case), BCQ entailment remains in the first level of the polynomial hierarchy for DL-Lite. 
Our results show that even if the priority relation over facts is not score-structured, and not even transitive, these computational properties are retained for Pareto- and completion-optimal repairs.

When $\succ$ is score-structured, we can compare the grounded extension and the intersection of optimal repairs 
to the proposals of \citeauthor{DBLP:conf/ijcai/BenferhatBT15} \shortcite{DBLP:conf/ijcai/BenferhatBT15} for selecting a single preferred \Tmc-consistent subset of the ABox. 
In particular, the \emph{non-defeated repair} is the union of the intersections of the repairs of $S_1, S_1\cup S_2, \dots, S_1\cup\dots \cup S_n$, where $S_1,\dots, S_n$ is a partition of $\Amc$ into priority levels. 
Since Elect coincides with non-defeated semantics for score-structured priorities \cite{DBLP:conf/ksem/BelabbesB19}, 
 the non-defeated repair is included in the grounded extension by Theorem \ref{thm:elect}; as the priority relation in Example \ref{ex:Elect} is score-structured, the inclusion may be strict. 
Another proposal 
is the \emph{prioritized inclusion-based non-defeated repair}, which coincides with the intersection of the optimal repairs.  
The grounded extension thus lies 
between the non-defeated and prioritized inclusion-based non-defeated repairs. 
Other proposals are either included in non-defeated, or are not sound approximations of the intersection of the optimal repairs. 
For partially preordered ABoxes, \citeauthor{DBLP:conf/ksem/BelabbesB19} \shortcite{DBLP:conf/ksem/BelabbesB19} propose to go beyond Elect by intersecting the optimal repairs w.r.t. the score-structured priorities  obtained by extending the original partial preorder into a total preorder.  
Interestingly, the obtained set of assertions corresponds exactly to the intersection of the completion-optimal repairs of the prioritized KB underpinned by the partially preordered KB (cf.\ appendix D).

Argumentation frameworks derived from inconsistent KBs have been considered before. 
Arguments are generally defined as pairs of a support (subset of the KB) and a conclusion (consequence of the support), with various attack relations (see \eg \cite{DBLP:journals/jar/AmgoudC02,DBLP:journals/ai/GorogiannisH11}). 
A series of papers starting with \cite{DBLP:conf/sum/CroitoruV13} 
links argumentation and inconsistency-tolerant querying of KBs,
where the support of the argument is a subset of the ABox, the conclusion is a conjunction of facts entailed from the support and $\Tmc$,
and 
 $\alpha$ attacks $\beta$ iff the conclusion of $\alpha$ is \Tmc-inconsistent with the support of~$\beta$ (in some papers, the argument consists of a whole derivation sequence, not just the support and conclusion). 
Stable and preferred extensions of such AFs correspond to repairs in the following sense: every such extension contains all those arguments whose supports are included in some particular repair. 
The work by \citeauthor{DBLP:conf/prima/CroitoruTV15} \shortcite{DBLP:conf/prima/CroitoruTV15} shows that it is possible to use a preference relation over facts to define notions of optimal extensions for AFs induced by inconsistent KBs in this fashion, and observe that such extensions correspond to optimal repairs. 
While touching on similar topics, 
an essential difference between this line of work and our own lies in the definition of the AF. 
Indeed, in the AFs of Croitoru et al.
, even if we group together arguments with the same supports, we would still have an exponential number of arguments. 
By contrast, our translation can be carried out in polynomial time under reasonable assumptions, enabling us to import tractability results from argumentation to OMQA. Moreover, since we use assertions as arguments, each priority relation yields a preference relation over arguments, 
enabling a transparent reduction to preference-based (SET)AFs. \citeauthor{DBLP:conf/prima/CroitoruTV15} do not discuss any connections to preference-based, nor set-based AFs, and 
%
%
%
%
our results on the equivalence between Pareto-optimal repairs and stable and preferred extensions of P(SET)AFs 
do not follow from their results. 

\section{Conclusion and Future Work}
We have explored the problem of how to repair and query inconsistent KBs 
while taking into account a priority relation over facts. 
By leveraging connections to database repairs and abstract argumentation,
we obtained a number of novel results for our setting, while at the same time 
contributing new results and research questions to these two areas.

After importing the notions of Pareto-, globally-, and completion-optimal repairs into the OMQA setting, 
our first contribution was a data complexity study that showed, unsurprisingly, that reasoning with optimal repairs is typically intractable and more challenging than for classical repairs. Nevertheless, there are several cases which are `only' \np / \conp, which suggests that it may be interesting to devise practical SAT-based procedures, as has been successfully done for some other forms of repair \cite{DBLP:conf/aaai/BienvenuBG14}. 
It would also be relevant to implement and experiment 
our proposed grounded semantics as well as the first-order approximations we suggested. Our complexity study could be expanded to 
include further ontology languages as well as the combined complexity measure. 

In order to provide an elegant translation of prioritized KBs to argumentation frameworks, we were naturally led to consider preference-based (SET)AFs. We believe that the new results we established in Section \ref{sec:arg}, in particular, Theorem \ref{psetaf-coh}, should be of interest to the argumentation community. Moreover, the connections we established between extensions of PSETAFs and repairs of prioritized KBs motivate a more detailed study of PSETAFs (a concrete open question is whether Theorem \ref{psetaf-coh} holds in the absence of transitivity). Furthermore, this correspondence could be leveraged to explore new notions of extension for (strongly) symmetric P(SET)AFs inspired by completion- or globally-optimal repairs, or to design benchmarks for (P)(SET)AFs via a translation from inconsistent KBs.

An important question that has not yet been satisfactorily answered in the database literature is which of the three forms of optimal repair is most natural, independently of their computational costs. While we do not claim to provide a definitive answer to this question, we believe that our result showing that Pareto-optimal repairs coincide with stable extensions (and often also preferred extensions) speaks to the interest of adopting Pareto-optimal repairs. We should emphasize that while phrased for KBs, the connection between Pareto-optimal repairs and stable / preferred extensions holds equally well for databases with denial constraints (such as functional dependencies). Moreover, the tractable grounded semantics we proposed can be applied not just to KBs but also to prioritized databases. 


\section*{Acknowledgements}
This work was supported by Camille Bourgaux's CNRS PEPS grant and the ANR AI Chair INTENDED.

\bibliographystyle{kr}
\bibliography{cqa-priority}

\begin{thebibliography}{}

\bibitem[\protect\citeauthoryear{Abiteboul, Hull, and
  Vianu}{1995}]{AbiteboulHV95}
Abiteboul, S.; Hull, R.; and Vianu, V.
\newblock 1995.
\newblock {\em Foundations of Databases}.
\newblock Addison-Wesley.

\bibitem[\protect\citeauthoryear{Amgoud and
  Cayrol}{2002}]{DBLP:journals/jar/AmgoudC02}
Amgoud, L., and Cayrol, C.
\newblock 2002.
\newblock Inferring from inconsistency in preference-based argumentation
  frameworks.
\newblock {\em J. Autom. Reasoning} 29(2):125--169.

\bibitem[\protect\citeauthoryear{Amgoud and
  Vesic}{2014}]{DBLP:journals/ijar/AmgoudV14}
Amgoud, L., and Vesic, S.
\newblock 2014.
\newblock Rich preference-based argumentation frameworks.
\newblock {\em Int. J. Approx. Reasoning} 55(2):585--606.

\bibitem[\protect\citeauthoryear{Artale \bgroup et al\mbox.\egroup
  }{2009}]{DBLP:journals/jair/ArtaleCKZ09}
Artale, A.; Calvanese, D.; Kontchakov, R.; and Zakharyaschev, M.
\newblock 2009.
\newblock The {DL-Lite} family and relations.
\newblock {\em J. Artif. Intell. Res.} 36:1--69.

\bibitem[\protect\citeauthoryear{Baader \bgroup et al\mbox.\egroup
  }{2017}]{DBLP:books/daglib/0041477}
Baader, F.; Horrocks, I.; Lutz, C.; and Sattler, U.
\newblock 2017.
\newblock {\em An Introduction to Description Logic}.
\newblock Cambridge University Press.

\bibitem[\protect\citeauthoryear{Baget \bgroup et al\mbox.\egroup
  }{2016}]{DBLP:conf/jelia/BagetBBCMPRT16}
Baget, J.; Benferhat, S.; Bouraoui, Z.; Croitoru, M.; Mugnier, M.; Papini, O.;
  Rocher, S.; and Tabia, K.
\newblock 2016.
\newblock Inconsistency-tolerant query answering: Rationality properties and
  computational complexity analysis.
\newblock In {\em {JELIA}}.

\bibitem[\protect\citeauthoryear{Belabbes and
  Benferhat}{2019}]{DBLP:conf/ksem/BelabbesB19}
Belabbes, S., and Benferhat, S.
\newblock 2019.
\newblock Inconsistency handling for partially preordered ontologies: Going
  beyond {Elect}.
\newblock In {\em {KSEM}}.

\bibitem[\protect\citeauthoryear{Belabbes, Benferhat, and
  Chomicki}{2019}]{DBLP:conf/lpnmr/BelabbesBC19}
Belabbes, S.; Benferhat, S.; and Chomicki, J.
\newblock 2019.
\newblock Elect: An inconsistency handling approach for partially preordered
  lightweight ontologies.
\newblock In {\em {LPNMR}}.

\bibitem[\protect\citeauthoryear{Benferhat, Bouraoui, and
  Tabia}{2015}]{DBLP:conf/ijcai/BenferhatBT15}
Benferhat, S.; Bouraoui, Z.; and Tabia, K.
\newblock 2015.
\newblock How to select one preferred assertional-based repair from
  inconsistent and prioritized {DL-Lite} knowledge bases?
\newblock In {\em {IJCAI}}.

\bibitem[\protect\citeauthoryear{Bienvenu and
  Bourgaux}{2016}]{DBLP:conf/rweb/BienvenuB16}
Bienvenu, M., and Bourgaux, C.
\newblock 2016.
\newblock Inconsistency-tolerant querying of description logic knowledge bases.
\newblock In {\em Reasoning Web}.

\bibitem[\protect\citeauthoryear{Bienvenu and
  Ortiz}{2015}]{DBLP:conf/rweb/BienvenuO15}
Bienvenu, M., and Ortiz, M.
\newblock 2015.
\newblock Ontology-mediated query answering with data-tractable description
  logics.
\newblock In {\em Reasoning Web}.

\bibitem[\protect\citeauthoryear{Bienvenu and
  Rosati}{2013}]{DBLP:conf/ijcai/BienvenuR13}
Bienvenu, M., and Rosati, R.
\newblock 2013.
\newblock Tractable approximations of consistent query answering for robust
  ontology-based data access.
\newblock In {\em {IJCAI}}.

\bibitem[\protect\citeauthoryear{Bienvenu, Bourgaux, and
  Goasdou{\'{e}}}{2014}]{DBLP:conf/aaai/BienvenuBG14}
Bienvenu, M.; Bourgaux, C.; and Goasdou{\'{e}}, F.
\newblock 2014.
\newblock Querying inconsistent description logic knowledge bases under
  preferred repair semantics.
\newblock In {\em {AAAI}}.

\bibitem[\protect\citeauthoryear{Bienvenu, Bourgaux, and
  Goasdou{\'{e}}}{2019}]{DBLP:journals/jair/BienvenuBG19}
Bienvenu, M.; Bourgaux, C.; and Goasdou{\'{e}}, F.
\newblock 2019.
\newblock Computing and explaining query answers over inconsistent {DL-Lite}
  knowledge bases.
\newblock {\em J. Artif. Intell. Res.} 64:563--644.

\bibitem[\protect\citeauthoryear{Bienvenu}{2019}]{DBLP:conf/dlog/Bienvenu19}
Bienvenu, M.
\newblock 2019.
\newblock Inconsistency handling in ontology-mediated query answering: {A}
  progress report.
\newblock In {\em {DL}}.

\bibitem[\protect\citeauthoryear{Boros \bgroup et al\mbox.\egroup
  }{2000}]{DBLP:journals/ppl/BorosEGK00}
Boros, E.; Elbassioni, K.~M.; Gurvich, V.; and Khachiyan, L.
\newblock 2000.
\newblock An efficient incremental algorithm for generating all maximal
  independent sets in hypergraphs of bounded dimension.
\newblock {\em Parallel Processing Letters} 10(4):253--266.

\bibitem[\protect\citeauthoryear{Bourgaux}{2016}]{DBLP:phd/hal/Bourgaux16}
Bourgaux, C.
\newblock 2016.
\newblock {\em Inconsistency Handling in Ontology-Mediated Query Answering.}
\newblock Ph.D. Dissertation, University of Paris-Saclay, France.

\bibitem[\protect\citeauthoryear{Calvanese \bgroup et al\mbox.\egroup
  }{2007}]{calvaneseetal:dllite}
Calvanese, D.; De~Giacomo, G.; Lembo, D.; Lenzerini, M.; and Rosati, R.
\newblock 2007.
\newblock Tractable reasoning and efficient query answering in description
  logics: The {DL-Lite} family.
\newblock {\em J. Autom. Reasoning} 39(3):385--429.

\bibitem[\protect\citeauthoryear{Coste{-}Marquis, Devred, and
  Marquis}{2005}]{DBLP:conf/ecsqaru/Coste-MarquisDM05}
Coste{-}Marquis, S.; Devred, C.; and Marquis, P.
\newblock 2005.
\newblock Symmetric argumentation frameworks.
\newblock In {\em {ECSQARU}}.

\bibitem[\protect\citeauthoryear{Croitoru and
  Vesic}{2013}]{DBLP:conf/sum/CroitoruV13}
Croitoru, M., and Vesic, S.
\newblock 2013.
\newblock What can argumentation do for inconsistent ontology query answering?
\newblock In {\em {SUM}}.

\bibitem[\protect\citeauthoryear{Croitoru, Thomopoulos, and
  Vesic}{2015}]{DBLP:conf/prima/CroitoruTV15}
Croitoru, M.; Thomopoulos, R.; and Vesic, S.
\newblock 2015.
\newblock Introducing preference-based argumentation to inconsistent
  ontological knowledge bases.
\newblock In {\em {PRIMA}}.

\bibitem[\protect\citeauthoryear{Diller \bgroup et al\mbox.\egroup
  }{2020}]{symm-setaf}
Diller, M.; {Keshavarzi Zafarghandi}, A.; Linsbichler, T.; and Woltran, S.
\newblock 2020.
\newblock Investigating subclasses of abstract dialectical frameworks.
\newblock {\em Argument \& Computation}.
\newblock To appear.

\bibitem[\protect\citeauthoryear{Du, Qi, and
  Shen}{2013}]{DBLP:journals/kais/DuQS13}
Du, J.; Qi, G.; and Shen, Y.
\newblock 2013.
\newblock Weight-based consistent query answering over inconsistent {SHIQ}
  knowledge bases.
\newblock {\em Knowl. Inf. Syst.} 34(2):335--371.

\bibitem[\protect\citeauthoryear{Dung}{1995}]{DBLP:journals/ai/Dung95}
Dung, P.~M.
\newblock 1995.
\newblock On the acceptability of arguments and its fundamental role in
  nonmonotonic reasoning, logic programming and n-person games.
\newblock {\em Artif. Intell.} 77(2):321--358.

\bibitem[\protect\citeauthoryear{Dvor{\'{a}}k and
  Dunne}{2017}]{DBLP:journals/flap/DvorakD17}
Dvor{\'{a}}k, W., and Dunne, P.~E.
\newblock 2017.
\newblock Computational problems in formal argumentation and their complexity.
\newblock {\em {FLAP}} 4(8).

\bibitem[\protect\citeauthoryear{Eiter and
  Gottlob}{1995}]{DBLP:journals/siamcomp/EiterG95}
Eiter, T., and Gottlob, G.
\newblock 1995.
\newblock Identifying the minimal transversals of a hypergraph and related
  problems.
\newblock {\em {SIAM} J. Comput.} 24(6):1278--1304.

\bibitem[\protect\citeauthoryear{Fagin \bgroup et al\mbox.\egroup
  }{2016}]{DBLP:journals/tods/FaginKRV16}
Fagin, R.; Kimelfeld, B.; Reiss, F.; and Vansummeren, S.
\newblock 2016.
\newblock Declarative cleaning of inconsistencies in information extraction.
\newblock {\em {ACM} Trans. Database Syst.} 41(1):6:1--6:44.

\bibitem[\protect\citeauthoryear{Fagin, Kimelfeld, and
  Kolaitis}{2015}]{DBLP:conf/pods/FaginKK15}
Fagin, R.; Kimelfeld, B.; and Kolaitis, P.~G.
\newblock 2015.
\newblock Dichotomies in the complexity of preferred repairs.
\newblock In {\em {PODS}}.

\bibitem[\protect\citeauthoryear{Flouris and
  Bikakis}{2019}]{DBLP:journals/ijar/FlourisB19}
Flouris, G., and Bikakis, A.
\newblock 2019.
\newblock A comprehensive study of argumentation frameworks with sets of
  attacking arguments.
\newblock {\em Int. J. Approx. Reasoning} 109:55--86.

\bibitem[\protect\citeauthoryear{Gorogiannis and
  Hunter}{2011}]{DBLP:journals/ai/GorogiannisH11}
Gorogiannis, N., and Hunter, A.
\newblock 2011.
\newblock Instantiating abstract argumentation with classical logic arguments:
  Postulates and properties.
\newblock {\em Artif. Intell.} 175(9-10):1479--1497.

\bibitem[\protect\citeauthoryear{Johnson, Papadimitriou, and
  Yannakakis}{1988}]{DBLP:journals/ipl/JohnsonP88}
Johnson, D.~S.; Papadimitriou, C.~H.; and Yannakakis, M.
\newblock 1988.
\newblock On generating all maximal independent sets.
\newblock {\em Inf. Process. Lett.} 27(3):119--123.

\bibitem[\protect\citeauthoryear{Kaci, van~der Torre, and
  Villata}{2018}]{DBLP:conf/comma/KaciTV18}
Kaci, S.; van~der Torre, L. W.~N.; and Villata, S.
\newblock 2018.
\newblock Preference in abstract argumentation.
\newblock In {\em {COMMA}}.

\bibitem[\protect\citeauthoryear{Kaci, van~der Torre, and
  Weydert}{2006}]{DBLP:conf/ecai/KaciTW06}
Kaci, S.; van~der Torre, L. W.~N.; and Weydert, E.
\newblock 2006.
\newblock Acyclic argumentation: Attack = conflict + preference.
\newblock In {\em {ECAI}}.

\bibitem[\protect\citeauthoryear{Kimelfeld, Livshits, and
  Peterfreund}{2017}]{DBLP:conf/icdt/KimelfeldLP17}
Kimelfeld, B.; Livshits, E.; and Peterfreund, L.
\newblock 2017.
\newblock Detecting ambiguity in prioritized database repairing.
\newblock In {\em {ICDT}}.

\bibitem[\protect\citeauthoryear{Lembo \bgroup et al\mbox.\egroup
  }{2010}]{LemboLRRS10}
Lembo, D.; Lenzerini, M.; Rosati, R.; Ruzzi, M.; and Savo, D.~F.
\newblock 2010.
\newblock Inconsistency-tolerant semantics for description logics.
\newblock In {\em {RR}}.

\bibitem[\protect\citeauthoryear{Lembo \bgroup et al\mbox.\egroup
  }{2015}]{DBLP:journals/ws/LemboLRRS15}
Lembo, D.; Lenzerini, M.; Rosati, R.; Ruzzi, M.; and Savo, D.~F.
\newblock 2015.
\newblock Inconsistency-tolerant query answering in ontology-based data access.
\newblock {\em Journal Web Sem.} 33:3--29.

\bibitem[\protect\citeauthoryear{Livshits and
  Kimelfeld}{2017}]{DBLP:conf/pods/LivshitsK17}
Livshits, E., and Kimelfeld, B.
\newblock 2017.
\newblock Counting and enumerating (preferred) database repairs.
\newblock In {\em {PODS}}.

\bibitem[\protect\citeauthoryear{Lopatenko and
  Bertossi}{2007}]{DBLP:conf/icdt/LopatenkoB07}
Lopatenko, A., and Bertossi, L.~E.
\newblock 2007.
\newblock Complexity of consistent query answering in databases under
  cardinality-based and incremental repair semantics.
\newblock In {\em {ICDT}}.

\bibitem[\protect\citeauthoryear{Lukasiewicz, Malizia, and
  Vaicenavicius}{2019}]{DBLP:conf/aaai/LukasiewiczMV19}
Lukasiewicz, T.; Malizia, E.; and Vaicenavicius, A.
\newblock 2019.
\newblock Complexity of inconsistency-tolerant query answering in {Datalog}+/-
  under cardinality-based repairs.
\newblock In {\em {AAAI}}.

\bibitem[\protect\citeauthoryear{Martinez \bgroup et al\mbox.\egroup
  }{2014}]{DBLP:journals/ijar/MartinezPPSS14}
Martinez, M.~V.; Parisi, F.; Pugliese, A.; Simari, G.~I.; and Subrahmanian,
  V.~S.
\newblock 2014.
\newblock Policy-based inconsistency management in relational databases.
\newblock {\em Int. J. Approx. Reason.} 55(2):501--528.

\bibitem[\protect\citeauthoryear{Motik \bgroup et al\mbox.\egroup
  }{2012}]{profiles}
Motik, B.; Cuenca~Grau, B.; Horrocks, I.; Wu, Z.; Fokoue, A.; and Lutz, C.
\newblock 2012.
\newblock {OWL} 2 {W}eb {O}ntology {L}anguage profiles.
\newblock {W3C} {R}ecommendation.
\newblock Available at \url{http://www.w3.org/TR/owl2-profiles/}.

\bibitem[\protect\citeauthoryear{Nielsen and
  Parsons}{2006}]{DBLP:conf/argmas/NielsenP06}
Nielsen, S.~H., and Parsons, S.
\newblock 2006.
\newblock A generalization of {Dung}'s abstract framework for argumentation:
  Arguing with sets of attacking arguments.
\newblock In {\em {ArgMAS}}.

\bibitem[\protect\citeauthoryear{{OWL Working Group}}{2009}]{owl2}
{OWL Working Group}, W.
\newblock 2009.
\newblock {\em OWL~2 Web Ontology Language: Document overview}.
\newblock W3C Recommendation.
\newblock Available at \url{https://www.w3.org/TR/owl2-overview/}.

\bibitem[\protect\citeauthoryear{Poggi \bgroup et al\mbox.\egroup
  }{2008}]{DBLP:journals/jods/PoggiLCGLR08}
Poggi, A.; Lembo, D.; Calvanese, D.; {De Giacomo}, G.; Lenzerini, M.; and
  Rosati, R.
\newblock 2008.
\newblock Linking data to ontologies.
\newblock {\em Journal of Data Semantics} 10:133--173.

\bibitem[\protect\citeauthoryear{Rao \bgroup et al\mbox.\egroup
  }{1997}]{DBLP:conf/lpnmr/RaoSSWF97}
Rao, P.; Sagonas, K.; Swift, T.; Warren, D.~S.; and Freire, J.
\newblock 1997.
\newblock {XSB:} {A} system for effciently computing {WFS}.
\newblock In {\em LPNMR}.

\bibitem[\protect\citeauthoryear{Staworko, Chomicki, and
  Marcinkowski}{2012}]{DBLP:journals/amai/StaworkoCM12}
Staworko, S.; Chomicki, J.; and Marcinkowski, J.
\newblock 2012.
\newblock Prioritized repairing and consistent query answering in relational
  databases.
\newblock {\em Ann. Math. Artif. Intell.} 64(2-3):209--246.

\bibitem[\protect\citeauthoryear{Tanon, Bourgaux, and
  Suchanek}{2019}]{DBLP:conf/www/TanonBS19}
Tanon, T.~P.; Bourgaux, C.; and Suchanek, F.~M.
\newblock 2019.
\newblock Learning how to correct a knowledge base from the edit history.
\newblock In {\em {WWW}}.

\bibitem[\protect\citeauthoryear{Trivela, Stoilos, and
  Vassalos}{2018}]{DBLP:conf/aaai/TrivelaSV18}
Trivela, D.; Stoilos, G.; and Vassalos, V.
\newblock 2018.
\newblock A framework and positive results for {IAR}-answering.
\newblock In {\em {AAAI}}.

\bibitem[\protect\citeauthoryear{Tsalapati \bgroup et al\mbox.\egroup
  }{2016}]{DBLP:conf/ijcai/TsalapatiSSK16}
Tsalapati, E.; Stoilos, G.; Stamou, G.~B.; and Koletsos, G.
\newblock 2016.
\newblock Efficient query answering over expressive inconsistent description
  logics.
\newblock In {\em {IJCAI}}.

\bibitem[\protect\citeauthoryear{Xiao \bgroup et al\mbox.\egroup
  }{2018}]{DBLP:conf/ijcai/XiaoCKLPRZ18}
Xiao, G.; Calvanese, D.; Kontchakov, R.; Lembo, D.; Poggi, A.; Rosati, R.; and
  Zakharyaschev, M.
\newblock 2018.
\newblock Ontology-based data access: {A} survey.
\newblock In {\em {IJCAI}}.

\end{thebibliography}

\appendix
\clearpage
\onecolumn

\section{Proofs for Section \ref{sec:complexity}}

\PropComplexityRepairCheckEntailment*
\begin{proof}
For repair checking, the proofs are similar to the database case. Verifying that $\Amc'$ belongs to $\reps{\Kmc}$ can be done in polynomial time. Then if $\Amc'\in\reps{\Kmc}$, to show that $\Amc'$ is not globally-optimal, it suffices to guess a global improvement of $\Amc'$ and verify it in polynomial time; to show that $\Amc'$ is not Pareto-optimal, it suffices to check for each assertion $\beta\notin\Amc'$ whether $\Amc''=\Amc'\cup\{\beta\}\setminus\{\alpha\mid \beta\succ\alpha\}$ is $\Tmc$-consistent, which implies that it is a Pareto improvement of $\Amc'$. 
Finally, to check that $\Amc'$ is completion-optimal, we can use the greedy procedure described in Section \ref{sec:reps} by restricting the choice of facts to the intersection of $\Amc'$ and the not-yet-considered facts that are maximal \wrt $\succ$. 

For AR, IAR and brave entailment, we use the following standard algorithms to decide (non-)entailment. The upper complexity bounds then follow from the complexity of repair checking and BCQ entailment. 
To decide $\Kmc_\succ\not\armodels{X} q$ (resp.\ $\Kmc_\succ\bravemodels{X} q$), guess $\Amc'\in\xreps{\Kmc_\succ}$ such that $\tup{\Tmc,\Amc'}\not\models q$ (resp. $\tup{\Tmc,\Amc'}\models q$).  
To decide $\Kmc_\succ\not\iarmodels{X} q$, guess a set of assertions $\Bmc=\{\alpha_1,\dots,\alpha_n\}\subseteq\Amc$ such that $\tup{\Tmc,\Amc\setminus\Bmc}\not\models q$ together with $\Amc'_1,\dots,\Amc'_n$ in $\xreps{\Kmc_\succ}$ such that $\alpha_i\notin\Amc'_i$ for $1\leq i\leq n$. Since $\bigcap_{\Amc'\in\xreps{\Kmc_\succ}}\Amc'\subseteq \Amc\setminus\Bmc$, this implies that $\Kmc_\succ\not\iarmodels{X}q$.%
\end{proof}

Given a propositional formula $\varphi$ and a partial valuation $\nu$ of its variables, $\nu(\varphi)$ denotes the formula obtained by replacing each variable $x$ in the domain of $\nu$ by $\nu(x)$ in~$\varphi$. 

\ThmHardnessGlobal*
\begin{proof}
\mypar{\conp-hardness of globally-optimal repair checking}
We show \conp-hardness of repair checking for globally-optimal repairs by reduction from UNSAT. 
Let $\Phi=c_1\wedge \dots\wedge c_k$ be a conjunction of clauses $c_1,\dots, c_k$ over variables $x_1,\dots, x_n$. 
We define a KB $\Kmc=\tup{\Tmc,\Amc}$ and a priority relation $\succ$ over $\Amc$ as follows.
\begin{align*}
\Tmc=&\{\exists P^-\sqsubseteq\neg\exists N^-,
\exists P\sqsubseteq\neg \exists\mi{Unsat}^-,
\exists N\sqsubseteq\neg \exists\mi{Unsat}^-,\exists\mi{Block}_\exists^-\sqsubseteq\neg \exists P^-,
\exists\mi{Block}_\exists^-\sqsubseteq\neg \exists N^-\}
\\
&\cup\{
\mi{Exist}\sqsubseteq\neg\exists\mi{Block}_\exists, 
\mi{Exist}\sqsubseteq\neg\exists\mi{Unsat}
\}\\
\Amc=&\{\mi{Unsat}(a,c_i)\mid 1\leq i\leq k\}\cup\{P(c_i,x_j)\mid x_j\in c_i\}\cup\{N(c_i,x_j)\mid \neg x_j\in c_i\}\\
&\cup\{\mi{Block}_\exists(a,x_j)\mid 1\leq j\leq n\}\cup\{\mi{Exist}(a)\}
\end{align*}
\begin{align*}
\mi{Exist}(a)&\succ\mi{Block}_\exists(a,x_j)\\
\mi{Block}_\exists(a,x_j)&\succ P(c_i,x_j)\\
\mi{Block}_\exists(a,x_j)&\succ N(c_i,x_j)\\
P(c_i,x_j)&\succ \mi{Unsat}(a,c_i)\\
N(c_i,x_j)&\succ \mi{Unsat}(a,c_i)
\end{align*}

We show that $\Phi$ is satisfiable iff $\Amc'=\{\mi{Unsat}(a,c_i)\mid 1\leq i\leq k\}\cup\{\mi{Block}_\exists(a,x_j)\mid 1\leq j\leq n\}\notin\greps{\Kmc_\succ}$. It is easy to verify that $\Amc'\in\reps{\Kmc}$, \ie is an inclusion-maximal $\Tmc$-consistent subset of $\Amc$.

\noindent($\Rightarrow$) Assume that $\Phi$ is satisfiable and let $\nu$ be a valuation of $x_1,\dots,x_n$ that satisfies $\Phi$. Let $$\Amc_\nu=\{P(c_i,x_j)\mid x_j\in c_i, \nu(x_j)=\true\}\cup\{N(c_i,x_j)\mid \neg x_j\in c_i, \nu(x_j)=\false\}\cup\{\mi{Exist}(a)\}.$$ 
Since $\nu$ satisfies $\Phi$, then $\nu$ satisfies every clause $c_i$ so every $c_i$ has an outgoing $P$ or $N$ edge in $\Amc_\nu$. It is easy to see that $\Amc_\nu$ is $\Tmc$-consistent (since no $x_j$ has both a $P$ and a $N$ incoming edge). Moreover, since $\mi{Exist}(a)\succ\mi{Block}_\exists(a,x_j)$, $P(c_i,x_j)\succ \mi{Unsat}(a,c_i)$, and $N(c_i,x_j)\succ \mi{Unsat}(a,c_i)$, it follows that $\Amc_\nu$ is a global improvement of $\Amc'$, so $\Amc'\notin\greps{\Kmc_\succ}$.\smallskip\\
\noindent($\Leftarrow$) Assume that $\Amc'\notin\greps{\Kmc_\succ}$. There exists a global improvement $\Amc_\nu$ of $\Amc'$, \ie a $\Tmc$-consistent $\Amc_\nu\subseteq \Amc$, $\Amc_\nu\neq\Amc'$,  
such that for every $\alpha\in\Amc'\setminus\Amc_\nu$, there exists $\beta\in\Amc_\nu\setminus\Amc'$ such that $\beta\succ\alpha$. 
Assume for a contradiction that for every $x_j$, $\mi{Block}_\exists(a,x_j)\notin \Amc'\setminus\Amc_\nu$. Since $\Amc_\nu\neq\Amc'$ and $\Amc'$ is a maximal \Tmc-consistent subset of $\Amc$, then $\Amc'\not\subseteq\Amc_\nu$, so there exists $\alpha\in\Amc'\setminus\Amc_\nu$. Thus there exists $c_i$ such that $\mi{Unsat}(a,c_i)\in\Amc'\setminus\Amc_\nu$. 
By construction of $\succ$, there must then be some $P(c_i,x_j)$ or $N(c_i,x_j)$ in $\Amc_\nu$. 
However, since we assumed $\mi{Block}_\exists(a,x_j)\notin \Amc'\setminus\Amc_\nu$, then $\Amc_\nu$ is $\Tmc$-inconsistent. 
It follows that there exists $x_j$ such that $\mi{Block}_\exists(a,x_j)\in \Amc'\setminus\Amc_\nu$. 
By construction of $\succ$, it follows that $\mi{Exist}(a)\in\Amc_\nu$. 
Since $\Amc_\nu$ is $\Tmc$-consistent, there is thus no $\mi{Block}_\exists(a,x_j)$ and no $\mi{Unsat}(a,c_i)$ in $\Amc_\nu$.  
Hence for every $i$, $\mi{Unsat}(a,c_i)\in\Amc'\setminus\Amc_\nu$, so there must be some $P(c_i,x_j)$ or $N(c_i,x_j)$ in $\Amc_\nu$. 
Let $\nu$ be the valuation of $x_1,\dots, x_n$ such that $\nu(x_j)=\true$ iff $P(c_i,x_j)\in\Amc_\nu$. 
Since $\Amc_\nu$ is $\Tmc$-consistent, then every $x_j$ has only $P$ or $N$ incoming edges. 
It follows that if $P(c_i,x_j)\in\Amc_\nu$, then $x_j\in c_i$ and $\nu(x_j)=\true$ so $\nu(c_i)=\true$, and if $N(c_i,x_j)\in\Amc_\nu$, then $\neg x_j\in c_i$ and $\nu(x_j)=\false$ so $\nu(c_i)=\true$. 
Hence $\nu$ satisfies~$\Phi$.

\mypar{\piptwo-hardness of G-AR and G-IAR IQ entailment} 
We show \piptwo-hardness of G-AR and G-IAR IQ entailment by reduction from 2QBF. 
Let $\Psi=\forall x_1\dots x_n\exists x_{n+1}\dots x_{n+m} \Phi$ where $\Phi=c_1\wedge \dots\wedge c_k$ is a conjunction of clauses over variables $x_1,\dots, x_n, x_{n+1},\dots, x_{n+m}$. We assume that every variable occurs in both positive and negative literals (this is w.l.o.g. since $\Psi'$ obtained by replacing $\Phi$ by $\Phi'=c_1\wedge \dots\wedge c_{k}\wedge (x_1\vee\neg x_1)\wedge\dots\wedge (x_{m+n}\vee\neg x_{m+n})$ is valid iff $\Psi$ is valid). We define a KB $\Kmc=\tup{\Tmc,\Amc}$ and a priority relation $\succ$ over $\Amc$ as follows.
\begin{align*}
\Tmc=&\{\exists P^-\sqsubseteq\neg\exists N^-,
\exists P\sqsubseteq\neg \exists\mi{Unsat}^-,
\exists N\sqsubseteq\neg \exists\mi{Unsat}^-\}\cup\{\exists\mi{Block}_Q^-\sqsubseteq\neg \exists P^-,
\exists\mi{Block}_Q^-\sqsubseteq\neg \exists N^-\mid Q\in\{\exists,\forall\} \}\\&
\cup\{
\mi{Exist}\sqsubseteq\neg\exists\mi{Block}_\exists, 
\mi{Exist}\sqsubseteq\neg\exists\mi{Unsat},
\mi{Univ}\sqsubseteq\neg\exists\mi{Block}_\forall,
\mi{Exist}\sqsubseteq\neg\mi{Univ}, \mi{Univ}\sqsubseteq \neg\mi{NoUniv}
\}\\
&\cup\{\exists\mi{Block}_\exists\sqsubseteq\neg\exists\mi{Block}_\forall\}\\
\Amc=&\{\mi{Unsat}(a,c_i)\mid 1\leq i\leq k\}\cup\{P(c_i,x_j)\mid x_j\in c_i\}\cup\{N(c_i,x_j)\mid \neg x_j\in c_i\}\\
&\cup\{\mi{Block}_\forall(a,x_j)\mid 1\leq j\leq n\}\cup\{\mi{Block}_\exists(a,x_j)\mid n+1\leq j\leq n+m\}\cup\{\mi{Exist}(a),\mi{Univ}(a), \mi{NoUniv}(a)\}
\end{align*}
\begin{align*}
\mi{Block}_{\forall}(a,x_j) &\succ \mi{Block}_{\exists}(a, x_{j'})\\
\mi{Block}_{Q}(a,x_j) &\succ P(c_i, x_j)&Q\in\{\exists,\forall\}\\
\mi{Block}_{Q}(a,x_j) &\succ N(c_i, x_j)&Q\in\{\exists,\forall\}\\
P(c_i,x_j) &\succ \mi{Unsat}(a,c_i)\\
N(c_i,x_j) &\succ \mi{Unsat}(a,c_i)\\
\mi{Exist}(a)&\succ\mi{Univ}(a)
\end{align*}
We show that $\Kmc_\succ\armodels{G}\mi{NoUniv}(a)$ (resp. $\Kmc_\succ\iarmodels{G}\mi{NoUniv}(a)$) iff $\Psi$ is valid. 

\noindent($\Rightarrow$) Assume that $\Psi$ is not valid: there exists a valuation $\nu_\forall$ of $x_1,\dots,x_n$ such that for every valuation $\nu_\exists$ of $x_{n+1},\dots,x_{n+m}$, $\nu_\forall(\nu_\exists(\Phi))$ evaluates to false. 
Let
\begin{align*}
\Amc_\forall=&
\{\mi{Univ}(a)\}\cup\{\mi{Block}_\exists(a,x_j)\mid n+1\leq j\leq n+m\}\\&
\cup\{P(c_i,x_j)\mid 1\leq j\leq n, x_j\in c_i, \nu_\forall(x_j)=\true\}\\&
\cup\{N(c_i,x_j)\mid 1\leq j\leq n, \neg x_j\in c_i, \nu_\forall(x_j)=\false\}
\\&
\cup\{\mi{Unsat}(a,c_i)\mid \nu_\forall(c_i)\text{ does not evaluate to }\true\}.
\end{align*} 
It is easy to check that $\Amc_\forall$ is a repair. 
In particular, every $x_j$ has only incoming $P$ or $N$ edges in $\Amc_\forall$ and every $c_i$ has an incoming $\mi{Unsat}$ edge only if it has no outgoing $P$ or $N$ edges, so $\Amc_\forall$ is $\Tmc$-consistent. Moreover, it is not possible to add an assertion from $\Amc\setminus\Amc_\forall$ while staying $\Tmc$-consistent. In particular, the $P$ and $N$ assertions that involve some $x_j$ with $j>n$ are in conflict with $\mi{Block}_\exists(a,x_j)\in\Amc_\forall$. We show that $\Amc_\forall\in\greps{\Kmc_\succ}$, which implies $\Kmc_\succ\not\armodels{G}\mi{NoUniv}(a)$ (and thus $\Kmc_\succ\not\iarmodels{G}\mi{NoUniv}(a)$).

Assume for a contradiction that there exists a $\Tmc$-consistent $\Amc'\subseteq\Amc$, $\Amc'\neq\Amc_\forall$, such that for every $\alpha\in\Amc_\forall\setminus\Amc'$, there exists $\beta\in\Amc'\setminus\Amc_\forall$ such that $\beta\succ\alpha$. 
If $\mi{Exist}(a)\notin\Amc'$, then $\Amc'$ contains $\mi{Univ}(a)$ by  construction of $\succ$. Since $\mi{Univ}(a)$ is inconsistent with the $\mi{Block}_\forall(a,x_j)$, then (i) for every $j\leq n$, $\mi{Block}_\forall(a,x_j)\notin\Amc'$ and (ii) $\Amc'$ contains all the $\mi{Block}_\exists(a,x_j)$ by  construction of $\succ$. It follows from (ii) that for every $j>n$, $x_j$ has no incoming $P$ or $N$ edge in $\Amc'$, and from (i) that for every $j\leq n$, $\Amc'$ contains the same $P(c_i,x_j)$ and $N(c_i,x_j)$ as $\Amc_\forall$  by construction of $\succ$. Hence, by construction of $\succ$, $\Amc'$ contains the same $\mi{Unsat}(a,c_i)$ as $\Amc_\forall$. It follows that $\Amc'=\Amc_\forall$, which contradicts our assumptions on $\Amc'$. Hence $\mi{Exist}(a)\in\Amc'$. 
Since $\mi{Exist}(a)\in\Amc'$, then $\mi{Univ}(a)\notin\Amc'$, for every $j>n$, $\mi{Block}_\exists(a,x_j)\notin\Amc'$, and for every $i$, $\mi{Unsat}(a,c_i)\notin\Amc'$. 
By construction of $\succ$, for each $\mi{Unsat}(a,c_i)\in\Amc_\forall$, there must then be some $P(c_i,x_j)$ or $N(c_i,x_j)$ in $\Amc'\setminus\Amc_\forall$. 
Moreover, since $\Psi$ is not valid, one of these $P(c_i,x_j)$ or $N(c_i,x_j)$ in $\Amc'\setminus\Amc_\forall$ is such that $j\leq n$.
Indeed, we otherwise define the valuation $\nu_\exists$ of $x_{n+1},\dots,x_{n+m}$ by $\nu_\exists(x_j)=\true$ iff $x_j$ has an incoming $P$ edge in $\Amc'$. It is easy to see that  $\nu_\exists$ satisfies all clauses of $\Phi$ that were not already satisfied by $\nu_\forall$, so that $\nu_\exists(\nu_\forall(\Phi))$ evaluates to true.  
Assume w.l.o.g. that there is $j\leq n$ such that $P(c_i,x_j)\in\Amc'\setminus\Amc_\forall$. Since $P(c_i,x_j)\notin\Amc_\forall$, $\nu_\forall(x_j)=\false$. 
Since we assumed that both $x_j$ and $\neg x_j$ occur in $\Phi$, it follows that $N(c_{i'},x_j)\in\Amc_\forall$. 
Hence, since $\Amc'$ is \Tmc-consistent, $N(c_{i'},x_j)\in\Amc_\forall\setminus\Amc'$ so there must be $\beta\in\Amc'$ such that $\beta\succ N(c_{i'},x_j)$. However, every $\beta\in\Amc$ such that $\beta\succ N(c_{i'},x_j)$ is inconsistent with $P(c_i,x_j)\in\Amc'$. 
We obtain a contradiction so there does not exist any global improvement of $\Amc_\forall$, and $\Amc_\forall\in\greps{\Kmc_\succ}$. \smallskip\\
\noindent($\Leftarrow$) In the other direction, assume that $\Psi$ is valid: for every valuation $\nu_\forall$ of $x_1,\dots,x_n$, there exists a valuation $\nu_\exists$ of $x_{n+1},\dots,x_{n+m}$ such that $\nu_\forall(\nu_\exists(\Phi))$ evaluates to true. 
Assume for a contradiction that there exists $\Amc'\in\greps{\Kmc_\succ}$ such that $\mi{NoUniv}(a)\notin\Amc'$, \ie $\mi{Univ}(a)\in\Amc'$ by maximality of $\Amc'$. 
By \Tmc-consistency of $\Amc'$, $\mi{Exist}(a)\notin\Amc'$ and there is no $\mi{Block}_{\forall}(a,x_j)$ in $\Amc'$. Since $\mi{Block}_{\exists}(a,x_j) \succ P(c_i, x_j)$ and $\mi{Block}_{\exists}(a,x_j)\succ N(c_i, x_j)$, then $\{\mi{Block}_{\exists}(a,x_j)\mid n+1\leq j\leq n+m\}\subseteq\Amc'$. Finally, since $P(c_i,x_j) \succ \mi{Unsat}(a,c_i)$ and $N(c_i,x_j) \succ \mi{Unsat}(a,c_i)$, $\Amc'$ contains a maximal \Tmc-consistent subset of the $P(c_i,x_j)$ and $N(c_i,x_j)$ with $j\leq n$, so 
there exists a valuation $\nu_\forall$ of $x_1,\dots,x_n$ such that 
\begin{align*}
\Amc'=&\{\mi{Univ}(a)\}\cup\{\mi{Block}_{\exists}(a,x_j)\mid n+1\leq j\leq n+m\}\\
&\cup\{P(c_i,x_j)\mid j\leq n, P(c_i,x_j)\in\Amc, \nu_\forall(x_j)=\true\}\\
&\cup\{N(c_i,x_j)\mid j\leq n, N(c_i,x_j)\in\Amc, \nu_\forall(x_j)=\false\}\\
&\cup\{\mi{Unsat}(a,c_i)\mid \nu_\forall(c_i)\text{ does not evaluate to true}\}.
\end{align*}

Since $\Psi$ is valid, there exists a valuation $\nu_\exists$ of $x_{n+1},\dots,x_{n+m}$ such that $\nu_\forall(\nu_\exists(\Phi))$ evaluates to true. 
Let 
\begin{align*}
\Amc_\nu=&\{\mi{Exist}(a)\}\cup\{\mi{Block}_{\forall}(a,x_j)\mid1\leq j\leq n\}\cup\\&\{P(c_i,x_j)\mid j>n,P(c_i,x_j)\in\Amc,\nu_\exists(x_j)=\true \}\cup\\&\{N(c_i,x_j)\mid j>n,N(c_i,x_j)\in\Amc,\nu_\exists(x_j)=\false \}. 
\end{align*}
Each $c_i$ such that $\mi{Unsat}(a,c_i)\in\Amc'$ has a $P$ or $N$ outgoing edge in $\Amc_\nu$ because $\nu_\forall(\nu_\exists(\Phi))$ evaluates to true and $\mi{Unsat}(a,c_i)\in\Amc'$ only if $\nu_\forall(c_i)$ does not evaluate to true, so that $\nu_\exists(c_i)$ evaluates to true. 
Moreover, $\mi{Exist}(a)\succ\mi{Univ}(a)$, $\mi{Block}_{\forall}(a,x_j)\succ\mi{Block}_{\exists}(a,x_{j'})$, $\mi{Block}_{\forall}(a,x_j) \succ P(c_i, x_j)$, $\mi{Block}_{\forall}(a,x_j) \succ N(c_i, x_j)$, $P(c_i,x_j)\succ\mi{Unsat}(a,c_i)$, and $N(c_i,x_j)\succ\mi{Unsat}(a,c_i)$. 
It follows that $\Amc_\nu$ is a global improvement of $\Amc'$. 
We conclude that every globally-optimal repair contains $\mi{NoUniv}(a)$, \ie $\Kmc_\succ\iarmodels{G} \mi{NoUniv}(a)$.

\mypar{\sigmaptwo-hardness of G-brave IQ entailment} We use the above reduction to show that $\Kmc_\succ\bravemodels{G}\mi{Univ}(a)$ iff $\Psi$ is not valid. 

\noindent($\Rightarrow$) If $\Psi$ is valid, the globally-optimal repairs contain $\mi{NoUniv}(a)$ so do not contain $\mi{Univ}(a)$ and  $\Kmc_\succ\not\bravemodels{G}\mi{Univ}(a)$.
 
\noindent($\Leftarrow$) If $\Psi$ is not valid, $\Amc_\forall\in\greps{\Kmc_\succ}$ and $\mi{Univ}(a)\in\Amc_\forall$ so $\Kmc_\succ\bravemodels{G}\mi{Univ}(a)$. 
\end{proof}

\ThmUniqueness*
\begin{proof}
\mypar{Upper bounds}
To show that $|\xreps{\Kmc_\succ}|\neq 1$, guess $\Amc'\neq\Amc''$ and check that $\Amc'\in\xreps{\Kmc_\succ}$ and $\Amc''\in\xreps{\Kmc_\succ}$. For DLs with polynomial time consistency checking, repair checking is in \ptime for $X=P$ and in \conp for $X=G$ by Theorem~\ref{Prop:complexityRepairCheckEntailment}.

\mypar{\conp-hardness of uniqueness for Pareto-optimal repairs} 
We show \conp-hardness of uniqueness for Pareto-optimal repairs by reduction from UNSAT. 
Let $\Phi=c_1\wedge \dots\wedge c_k$ be a conjunction of clauses $c_1,\dots, c_k$ over variables $x_1,\dots, x_n$. 
We define a KB $\Kmc=\tup{\Tmc,\Amc}$ and a priority relation $\succ$ over $\Amc$ as follows.
\begin{align*}
\Tmc=&\{\exists P^-\sqsubseteq\neg\exists N^-,
\exists P\sqsubseteq\neg \exists\mi{Unsat}^-,
\exists N\sqsubseteq\neg \exists\mi{Unsat}^-,\exists\mi{Block}_\exists^-\sqsubseteq\neg \exists P^-,
\exists\mi{Block}_\exists^-\sqsubseteq\neg \exists N^-
\}
\\
&\cup\{
\mi{Exist}\sqsubseteq\neg\exists\mi{Block}_\exists, 
\mi{Exist}\sqsubseteq\neg\exists\mi{Unsat},
\mi{Exist}\sqsubseteq\neg\mi{NoExist}
\}\\
\Amc=&\{\mi{Unsat}(a,c_i)\mid 1\leq i\leq k\}\cup\{P(c_i,x_j)\mid x_j\in c_i\}\cup\{N(c_i,x_j)\mid \neg x_j\in c_i\}
\\
&\cup\{\mi{Block}_\exists(a,x_j)\mid 1\leq j\leq n\}
\cup\{\mi{Exist}(a),\mi{NoExist}(a)\}
\end{align*}
\begin{align*}
\mi{Unsat}(a,c_i)\succ& \mi{Exist}(a)\\
\mi{Exist}(a)\succ&\mi{Block}_\exists(a,x_j)\\
\mi{Block}_\exists(a,x_j)\succ& P(c_i,x_j)\\
\mi{Block}_\exists(a,x_j)\succ& N(c_i,x_j)
\end{align*}
Let $\Amc'=\{\mi{NoExist}(a)\}\cup\{\mi{Block}_\exists(a,x_j)\mid 1\leq j\leq n\}\cup\{\mi{Unsat}(a,c_i)\mid 1\leq i\leq k\}$. It is easy to check that $\Amc'\in\preps{\Kmc_\succ}$: $\Amc'$ is $\Tmc$-consistent, and it is not possible to add any assertion $\alpha\in\Amc\setminus\Amc'$ to $\Amc'$ while staying $\Tmc$-consistent without removing an assertion $\beta\in\Amc'$ such that $\alpha\not\succ\beta$. 
We show that $\Phi$ is unsatisfiable iff $|\preps{\Kmc_\succ}|=1$, \ie $\Amc'$ is the unique Pareto-optimal repair. \smallskip\\
\noindent($\Leftarrow$) Assume that $\Phi$ is satisfiable: there exists a valuation $\nu$ that satisfies $\Phi$.
Let $\Amc_\nu=\{P(c_i,x_j)\mid x_j\in c_i, \nu(x_j)=\true\}\cup\{N(c_i,x_j)\mid \neg x_j\in c_i, \nu(x_j)=\false\}
\cup\{\mi{Exist}(a)\}$. 
 It is easy to see that $\Amc_\nu$ is $\Tmc$-consistent since every $x_j$ has only incoming $P$ or $N$ edges. 
 We show that $\Amc_\nu\in\preps{\Kmc_\succ}$. 
It is not possible to add $\mi{NoExist}(a)$ without removing $\mi{Exist}(a)$ and $\mi{NoExist}(a)\not\succ\mi{Exist}(a)$. 
Since $\nu$ satisfies $\Phi$, every $c_i$ has an outgoing $P$ or $N$ edge in $\Amc_\nu$ so it is not possible to add an assertion of the form $\mi{Unsat}(a,c_i)$ without removing an assertion of the form $P(c_i,x_j)$ or $N(c_i,x_j)$ and $\mi{Unsat}(a,c_i)\not\succ P(c_i,x_j)$, $\mi{Unsat}(a,c_i)\not\succ N(c_i,x_j)$. 
It is not possible to add a $\mi{Block}_\exists(a,x_j)$ assertion without removing $\mi{Exist}(a)$ and  $\mi{Block}_\exists(a,x_j)\not\succ\mi{Exist}(a)$. 
Finally, it is not possible to add $P(c_i,x_j)$ without removing some $N(c_{i'},x_j)$ and $P(c_i,x_j)\not\succ N(c_{i'},x_j)$. Similarly it is not possible to add $N(c_i,x_j)$. 
It follows that $\Amc_\nu\in\preps{\Kmc_\succ}$ and $|\preps{\Kmc_\succ}|\geq 2$.\smallskip\\
\noindent($\Rightarrow$) Assume that there exists a Pareto-optimal repair $\Amc''\neq\Amc'$. 
Assume for a contradiction that $\mi{Exist}(a)\notin\Amc''$. 
By maximality, $\mi{NoExist}(a)\in\Amc''$. 
Moreover, $\Amc''$ contains no $P$ or $N$ 
assertions: otherwise, all $P$ or $N$ 
edges incoming in some $x_j$ could be replaced by $\mi{Block}_\exists(a,x_j)$ to obtain a Pareto-improvement of $\Amc''$. 
 By maximality, $\Amc''$ thus contains all $\mi{Block}_\exists$ and $\mi{Unsat}$ assertions. We obtain that $\Amc''=\Amc'$, which contradicts our assumption. 
It follows that $\mi{Exist}(a)\in\Amc''$. 
Hence $\mi{NoExist}(a)\notin\Amc''$ and there is no $\mi{Block}_\exists$ or $\mi{Unsat}$ assertions in $\Amc''$. Thus the remaining assertions of $\Amc''$ form 
 a maximal subset of the $P$ and $N$ assertions such that no $x_j$ has both $P$ and $N$ incoming edges. 
Let $\nu$ be the valuation of $x_1,\dots, x_n$ such that $\nu(x_j)=\true$ iff $P(x_j)\in\Amc''$.  
We show that $\nu$ satisfies $\Phi$. Otherwise, there would be a $c_i$ without $P$ or $N$ outgoing edge and $(\Amc''\setminus\{\mi{Exist}(a)\})\cup\{\mi{Unsat}(a,c_i)\}$ would be a Pareto improvement of $\Amc''$ since $\mi{Unsat}(a,c_i)\succ \mi{Exist}(a)$. It follows that $\Phi$ is satisfiable. 

\mypar{\piptwo-hardness of uniqueness for globally-optimal repairs} 
We show \piptwo-hardness of uniqueness for globally-optimal repairs by reduction from 2QBF. Let $\Psi=\forall x_1\dots x_n\exists x_{n+1}\dots x_{n+m} \Phi$ where $\Phi=c_1\wedge \dots\wedge c_k$ is a conjunction of clauses over variables $x_1,\dots, x_n, x_{n+1},\dots, x_{n+m}$. 
We define a KB $\Kmc=\tup{\Tmc,\Amc}$ and a priority relation $\succ$ over $\Amc$ as follows.
\begin{align*}
\Tmc=&\{
\exists P^-\sqsubseteq\neg\exists N^-,
\exists P\sqsubseteq\neg \exists\mi{Unsat}^-,
\exists N\sqsubseteq\neg \exists\mi{Unsat}^-\}\\
&\cup\{\exists P_\forall^-\sqsubseteq\neg\exists N_\forall^-,
\exists P_\forall^-\sqsubseteq\neg\exists N^-,
\exists P^-\sqsubseteq\neg\exists N_\forall^-\}
\\&\cup
\{\exists P_\forall\sqsubseteq\neg \exists P,
\exists  P_\forall\sqsubseteq\neg \exists N,
\exists  N_\forall\sqsubseteq\neg \exists P,
\exists  N_\forall\sqsubseteq\neg \exists N\}\\
&\cup\{\exists\mi{Block}_Q^-\sqsubseteq\neg \exists P^-,
\exists\mi{Block}_Q^-\sqsubseteq\neg \exists N^-\mid Q\in\{\exists,
\mn{all}\} \}\\
&\cup\{\exists\mi{Block}_Q^-\sqsubseteq\neg \exists P_\forall^-,
\exists\mi{Block}_Q^-\sqsubseteq\neg \exists N_\forall^-\mid Q\in\{\forall,
\mn{all}\} \}\\
&\cup\{\exists\mi{Block}_\exists\sqsubseteq\neg\exists\mi{Block}_\forall, \exists\mi{Block}_\forall\sqsubseteq\neg\exists\mi{Block}_\mn{all},\exists\mi{Block}_\mn{all}\sqsubseteq\neg\exists\mi{Block}_\exists \}\\
&\cup\{\exists\mi{Block}_Q\sqsubseteq\neg \mi{Exist}\mid Q\in\{\exists,
\mn{all}\} \}\\
&\cup\{\exists\mi{Block}_Q\sqsubseteq\neg \mi{Univ}, \exists\mi{Block}_Q\sqsubseteq\neg \exists\mi{Unsat}\mid Q\in\{\forall,
\mn{all}\} \}\\
&\cup\{\exists\mi{Block}_Q\sqsubseteq\neg \mi{NoExist}\mid Q\in\{\exists,
\forall\} \}\\
&\cup\{
\mi{Exist}\sqsubseteq\neg\mi{Univ},\mi{Univ}\sqsubseteq\neg\mi{NoExist},\mi{NoExist}\sqsubseteq\neg\mi{Exist}
\}\\
&\cup\{
\mi{NoExist}\sqsubseteq\neg\exists\mi{Unsat}, 
\mi{Exist}\sqsubseteq\neg\exists\mi{Unsat}
\}
\end{align*}
\begin{align*}
\Amc=&\{\mi{Unsat}(a,c_i)\mid 1\leq i\leq k\}\cup\{P(c_i,x_j)\mid x_j\in c_i\}\cup\{N(c_i,x_j)\mid \neg x_j\in c_i\}\\
&\cup\{P_\forall(c_i,x_j)\mid x_j\in c_i, j\leq n\}\cup\{N_\forall(c_i,x_j)\mid \neg x_j\in c_i, j\leq n\}\\
&\cup\{\mi{Block}_\forall(a,x_j),\mi{Block}_\exists(a,x_j),\mi{Block}_\mn{all}(a,x_j)\mid 1\leq j\leq n+m\}\\
&\cup\{\mi{NoExist}(a),\mi{Exist}(a),\mi{Univ}(a)\}
\end{align*}
\begin{align*}
\mi{NoExist}(a)&\succ\mi{Exist}(a)\\
\mi{Exist}(a)&\succ\mi{Univ}(a)\\
\mi{Block}_\mn{all}(a,x_j)&\succ\mi{Block}_\forall(a,x_{j'})\\
\mi{Block}_\forall(a,x_j)&\succ\mi{Block}_\exists(a,x_{j'})\\
\mi{Block}_Q(a,x_j)&\succ P(c_i,x_j) &Q\in\{\mn{all},\exists\}\\
\mi{Block}_Q(a,x_j)&\succ N(c_i,x_j)&Q\in\{\mn{all},\exists\}\\
P(c_i,x_j)&\succ \mi{Unsat}(a,c_i)\\
N(c_i,x_j)&\succ \mi{Unsat}(a,c_i)\\
P(c_i,x_j)&\succ P_\forall(c_i,x_j)\\
N(c_i,x_j)&\succ N_\forall(c_i,x_j)
\end{align*}
Let $\Amc'=\{\mi{NoExist}(a)\}\cup\{\mi{Block}_\mn{all}(a,x_j)\mid 1\leq j\leq n+m\}$. It is easy to check that $\Amc'\in\greps{\Kmc_\succ}$: $\Amc'$ is a maximal $\Tmc$-consistent subset of $\Amc$ and there is no assertion preferred to $\mi{NoExist}(a)$ or $\mi{Block}_\mn{all}(a,x_j)$. 
We show that $\Psi$ is valid iff $|\greps{\Kmc_\succ}|=1$, \ie $\Amc'$ is the unique globally-optimal repair. \medskip\\
\noindent($\Leftarrow$) Assume that $\Psi$ is not valid: there exists a valuation $\nu_\forall$ of $x_1,\dots,x_n$ such that for every valuation $\nu$ of $x_{1},\dots,x_{n+m}$ that extends $\nu_\forall$, $\nu(\Phi)$ evaluates to false. 
Let
\begin{align*}
\Amc_\forall=&
\{\mi{Univ}(a)\}\cup\{\mi{Block}_\exists(a,x_j)\mid 1\leq j\leq n+m\}\\&
\cup\{P_\forall(c_i,x_j)\mid 1\leq j\leq n, x_j\in c_i, \nu_\forall(x_j)=\true\}\\&
\cup\{N_\forall(c_i,x_j)\mid 1\leq j\leq n, \neg x_j\in c_i, \nu_\forall(x_j)=\false\}
\\&
\cup\{\mi{Unsat}(a,c_i)\mid 1\leq i\leq k\}.
\end{align*} 
It is easy to check that $\Amc_\forall$ is a repair. 
In particular, every $x_j$ has only incoming $P_\forall$ or $N_\forall$ edges in $\Amc_\forall$, so $\Amc_\forall$ is $\Tmc$-consistent. Moreover, it is not possible to add an assertion from $\Amc\setminus\Amc_\forall$ while staying $\Tmc$-consistent. In particular, the $P$ and $N$ assertions are in conflict with the $\mi{Block}_\exists$ assertions. We show that $\Amc_\forall\in\greps{\Kmc_\succ}$, which implies $|\greps{\Kmc_\succ}|\geq 2$.\smallskip\\
Assume for a contradiction that there exists a global improvement $\Amc''$ of $\Amc_\forall$, \ie a $\Tmc$-consistent $\Amc''\subseteq\Amc$, $\Amc''\neq\Amc_\forall$, such that for every $\alpha\in\Amc_\forall\setminus\Amc''$, there exists $\beta\in\Amc''\setminus\Amc_\forall$ such that $\beta\succ\alpha$. 
If $\mi{Exist}(a)\notin\Amc''$, then $\Amc''$ contains $\mi{Univ}(a)$ by  construction of $\succ$. 
Hence $\Amc''$ does not contain any $\mi{Block}_\forall(a,x_j)$, and by  construction of $\succ$ it follows that $\Amc''$ contains all $\mi{Block}_\exists(a,x_j)$. 
Thus there is no $P$ or $N$ assertion in $\Amc''$. Hence by construction of $\succ$, $\Amc''$ contains the same $P_\forall$, $N_\forall$, and $\mi{Unsat}$ assertions as $\Amc_\forall$, so $\Amc''=\Amc_\forall$, which contradicts our assumption on $\Amc''$. It follows that $\mi{Exist}(a)\in\Amc''$. 
Hence $\mi{Unsat}(a,c_i)\notin\Amc''$ for every $i$. By construction of $\succ$, each $c_i$ has thus an outgoing $P$ or $N$ edge in $\Amc''$. 
It follows that $\Amc''$ does not contain any  $P_\forall$ or $N_\forall$ assertion. Hence, by construction of $\succ$, for each $P_\forall(c_i,x_j)\in\Amc_\forall$ (resp. $N_\forall(c_i,x_j)\in\Amc_\forall$), $P(c_i,x_j)\in\Amc''$ (resp. $N(c_i,x_j)\in\Amc''$). 
Moreover, since $\Amc''$ is $\Tmc$-consistent, every $x_j$ has only $P$ or $N$ incoming edges. Let $\nu$ be the valuation of $x_{1},\dots,x_{n+m}$ defined by $\nu(x_j)=\true$ iff $x_j$ has an incoming $P$ edge in $\Amc''$. It is easy to see that  $\nu$ satisfies all clauses of $\Phi$ and that $\nu$ extends $\nu_\forall$. 
This contradicts our assumption on $\nu_\forall$, so we conclude that $\Amc_\forall\in\greps{\Kmc_\succ}$. \smallskip\\
\noindent($\Rightarrow$) In the other direction, assume that $\Psi$ is valid: for every valuation $\nu_\forall$ of $x_1,\dots,x_n$, there exists a valuation $\nu$ of $x_{1},\dots,x_{n+m}$ that extends $\nu_\forall$ and such that $\nu(\Phi)$ evaluates to true. 
Suppose for a contradiction that there is a globally-optimal repair $\Amc''\neq\Amc'$. First note that every repair of $\Amc$ contains either $\mi{Exist}(a)$, or $\mi{NoExist}(a)$, or $\mi{Univ}(a)$. Indeed, it is not possible to contradict these three assertions with a $\Tmc$-consistent subset of $\Amc$. 
We consider these three cases below.

If $\mi{Exist}(a)\in\Amc''$, then $\mi{NoExist}(a)\notin\Amc''$, $\mi{Univ}(a)\notin\Amc''$, and there is no $\mi{Block}_\exists$, $\mi{Block}_\mn{all}$ or $\mi{Unsat}$ assertions in $\Amc''$. 
There is then no $P_\forall$ or $N_\forall$ assertions in $\Amc''$: otherwise, $\Amc''\setminus(\{P_\forall(c_i,x_j)\mid P_\forall(c_i,x_j)\in\Amc''\}\cup\{N_\forall(c_i,x_j)\mid N_\forall(c_i,x_j)\in\Amc''\})\cup\{P(c_i,x_j)\mid P_\forall(c_i,x_j)\in\Amc''\}\cup\{N(c_i,x_j)\mid N_\forall(c_i,x_j)\in\Amc''\}$ is a global improvement of $\Amc''$. 
Hence, by maximality, $\Amc''$ contains all $\mi{Block}_\forall$ assertions and a maximal $\Tmc$-consistent subset of the $P$ and $N$ assertions. Hence, since $\mi{NoExist}(a)\succ\mi{Exist}(a)$, $\mi{Block}_\mn{all}(a,x_j)\succ\mi{Block}_\forall(a,x_{j'})$, $\mi{Block}_\mn{all}(a,x_j)\succ P(c_i,x_j)$ and  $\mi{Block}_\mn{all}(a,x_j)\succ N(c_i,x_j)$, $\Amc'$ is a global improvement of $\Amc''$.

If $\mi{NoExist}(a)\in\Amc''$, then $\mi{Exist}(a)\notin\Amc''$, $\mi{Univ}(a)\notin\Amc''$, and there is no $\mi{Block}_\exists$, $\mi{Block}_\forall$ or $\mi{Unsat}$ assertions in $\Amc''$. Moreover, 
there is no $P_\forall$ or $N_\forall$ assertions in $\Amc''$ (otherwise, $\Amc''$ can be improved by replacing the $P_\forall$ or $N_\forall$ assertions by the corresponding $P$ or $N$ assertions), and there is no $P$ or $N$ assertions in $\Amc''$ (otherwise, $\Amc''$ can be improved by replacing the $P$ or $N$ assertions by $\mi{Block}_\mn{all}$ assertions). 
It follows that $\Amc''=\Amc'$ by maximality of repairs.

If $\mi{Univ}(a)\in\Amc''$, then $\mi{NoExist}(a)\notin\Amc''$, $\mi{Exist}(a)\notin\Amc''$, and there is no $\mi{Block}_\forall$ or $\mi{Block}_\mn{all}$ assertions in $\Amc''$. Since $\mi{Block}_\exists(a,x_j)\succ P(c_i,x_j)$ and 
$\mi{Block}_\exists(a,x_j)\succ N(c_i,x_j)$, $\Amc''$ does not contain any $P$ or $N$ assertions (otherwise $\Amc''$ could be improved by replacing them by some $\mi{Block}_\exists$ assertions). 
Let $\nu_\forall$ be the valuation of $x_1,\dots,x_n$ such that $\nu_\forall(x_j)=\true$ iff $P_\forall(c_i,x_j)\in\Amc''$. 
Let $\nu$ be a valuation that extends $\nu_\forall$ and satisfies $\Phi$. Consider 
\begin{align*}
\Amc_\nu=&\{\mi{Exist}(a)\}\cup\{\mi{Block}_\forall(a,x_j)\mid 1\leq j\leq n+m\}
\\&
\cup\{P(c_i,x_j)\mid x_j\in c_i, \nu(x_j)=\true\}
\\&\cup\{N(c_i,x_j)\mid \neg x_j\in c_i, \nu(x_j)=\false\}.
\end{align*}
We show that $\Amc_\nu$ is a global improvement of $\Amc''$. Indeed, (i) $\mi{Exist}(a)\succ\mi{Univ}(a)$, (ii) $\mi{Block}_\forall(a,x_j)\succ\mi{Block}_\exists(a,x_{j'})$, (iii) since $\nu$ satisfies $\Phi$ every $c_i$ has an outgoing $P$ or $N$ edge in $\Amc_\nu$ and $P(c_i,x_j)\succ \mi{Unsat}(a,c_i)$, 
$N(c_i,x_j)\succ \mi{Unsat}(a,c_i)$, and (iv) since $\nu$ extends $\nu_\forall$, for every $P_\forall$ or $N_\forall$ assertion in $\Amc''$, the corresponding $P$ or $N$ assertion is in $\Amc_\nu$ and 
$P(c_i,x_j)\succ P_\forall(c_i,x_j)$, 
$N(c_i,x_j)\succ N_\forall(c_i,x_j)$. 

All cases lead to a contradiction, so there is no $\Amc''\neq \Amc'$ such that $\Amc''\in\greps{\Kmc_\succ}$, \ie $|\greps{\Kmc_\succ}|=1$.
\end{proof}

\ThmTransitivity*
\begin{proof}
\citeauthor{DBLP:conf/icdt/KimelfeldLP17} \shortcite{DBLP:conf/icdt/KimelfeldLP17} have shown that when the priority relation is transitive, deciding uniqueness for globally-optimal repairs can be done in polynomial time when the conflict hypergraph is given. 
This result is applicable to all DLs with bounded conflicts, for which the conflict hypergraph is computable in polynomial time by checking the consistency of all subsets of size bounded by the maximal conflict size. 

To show that the lower bounds of the other problems remain the same when $\succ$ is transitive, some proofs have to be adapted. 

\mypar{Hardness of enumeration of globally-optimal repairs} 
The fact that enumeration of globally-optimal repairs is not in \totalp is not a corollary of the complexity of uniqueness as in the case where $\succ$ is not transitive. We show that the associated decision problem of deciding whether a given set of ABox subsets is exactly $\greps{\Kmc_\succ}$ is \conp-hard. 
We show this by reduction from UNSAT. Let $\Phi=c_1\wedge \dots\wedge c_k$ be a conjunction of clauses over variables $x_1,\dots, x_m$. 
We define a KB $\Kmc=\tup{\Tmc,\Amc}$ and a transitive priority relation $\succ$ over $\Amc$ as follows.

\begin{align*}
\Tmc=&\{\exists P^-\sqsubseteq\neg\exists N^-,
\exists P\sqsubseteq\neg \exists\mi{Unsat}^-,
\exists N\sqsubseteq\neg \exists\mi{Unsat}^-\}\\
&\cup\{\exists \mi{Block}_Q^-\sqsubseteq\neg\exists P^-,\exists \mi{Block}_
Q^-\sqsubseteq\neg\exists N^-,\mi{Exist}\sqsubseteq\neg \exists \mi{Block}_Q
\mid Q\in\{\exists,\mn{all}\}
\}\\
&\cup\{\exists \mi{Block}_\exists\sqsubseteq\neg \exists \mi{Block}_\mn{all},
 \mi{All}\sqsubseteq\neg\exists\mi{Block}_\mn{all},
\mi{Exist}\sqsubseteq\neg\mi{All},\mi{All}\sqsubseteq\neg\mi{Sat}, 
\mi{Sat}\sqsubseteq\neg\exists\mi{Unsat}
\}\\
\Amc=&\{\mi{Unsat}(a,c_i)\mid 1\leq i\leq k\}\cup\{P(c_i,x_j)\mid x_j\in c_i\}\cup\{N(c_i,x_j)\mid \neg x_j\in c_i\}\\
&\cup\{\mi{Block}_{\exists}(a,x_j),\mi{Block}_{\mn{all}}(a,x_j)\mid 1\leq j\leq m\}\cup\{\mi{Exist}(a),\mi{Sat}(a),\mi{All}(a)\}
\end{align*}
\begin{align*}
\mi{Sat}(a)&\succ\mi{All}(a)\\
\mi{Exist}(a)&\succ\mi{Block}_{\exists}(a,x_j) \\
\mi{Block}_{\mn{all}}(a,x_j) &\succ \mi{Exist}(a)\\
\mi{Block}_{\mn{all}}(a,x_j) &\succ \mi{Block}_{\exists}(a, x_{j'})\\
\mi{Block}_Q(a,x_j) &\succ P(c_i, x_j) &Q\in\{\mn{all},\exists\}\\
\mi{Block}_Q(a,x_j) &\succ N(c_i, x_j)&Q\in\{\mn{all},\exists\}\\
P(c_i,x_j) &\succ \mi{Unsat}(a,c_i)\\
N(c_i,x_j) &\succ \mi{Unsat}(a,c_i)
\end{align*}

We show that $\Phi$ is unsatisfiable iff $\greps{\Kmc_\succ}=\{\Amc_1,\Amc_2,\Amc_3\}$ where 
\begin{align*}
\Amc_1=&\{\mi{Sat}(a)\}\cup\{\mi{Block}_{\mn{all}}(a,x_j)\mid 1\leq j\leq m\},\\ 
\Amc_2=&\{\mi{Block}_{\mn{all}}(a,x_j)\mid 1\leq j\leq m\}\cup\{\mi{Unsat}(a,c_i)\mid 1\leq i\leq k\},\\ 
\Amc_3=&\{\mi{All}(a)\}\cup\{\mi{Block}_{\exists}(a,x_j)\mid 1\leq j\leq m\}\cup\{\mi{Unsat}(a,c_i)\mid 1\leq i\leq k\}.
\end{align*}
\noindent($\Leftarrow$) Assume that $\Phi$ is satisfiable. Then $\Amc_3\notin\greps{\Kmc_\succ}$. 
Indeed, let $\nu$ be a valuation that satisfies $\Phi$ and 
\begin{align*}
\Amc_\nu=&\{\mi{Sat}(a),\mi{Exist}(a)\}\cup\{P(c_i,x_j)\mid P(c_i,x_j)\in\Amc,\nu(x_j)=\true \}\cup\{N(c_i,x_j)\mid N(c_i,x_j)\in\Amc,\nu(x_j)=\false \}. 
\end{align*}
Since $\mi{Sat}(a)\succ\mi{All}(a)$, $\mi{Exist}(a)\succ\mi{Block}_{\exists}(a,x_j)$, $P(c_i,x_j)\succ\mi{Unsat}(a,c_i)$, $N(c_i,x_j)\succ\mi{Unsat}(a,c_i)$, and each $c_i$ has a $P$ or $N$ outgoing edge in $\Amc_\nu$ because $\nu$ satisfies $\Phi$, $\Amc_\nu$ is a global improvement of $\Amc_3$. 
It follows that $\greps{\Kmc_\succ}\neq\{\Amc_1,\Amc_2,\Amc_3\}$.\smallskip\\
\noindent($\Rightarrow$) Assume that $\Phi$ is unsatisfiable. 
We first show that $\{\Amc_1,\Amc_2,\Amc_3\}\subseteq \greps{\Kmc_\succ}$. 
It is easy to check that $\Amc_1,\Amc_2,\Amc_3$ are repairs of $\Kmc_\succ$. We show that they are globally-optimal repairs.
\begin{itemize}
\item $\Amc_1\in\greps{\Kmc_\succ}$ because none of its assertions is dominated by any other. 

\item $\Amc_2\in\greps{\Kmc_\succ}$ because its $\mi{Block}_{\mn{all}}$ assertions are not dominated by any others and the only assertions of greater priority than the $\mi{Unsat}$ assertions are the $P$, $N$ which are in conflict with the $\mi{Block}_{\mn{all}}$ assertions. 

\item For $\Amc_3$, assume for a contradiction that there exists a $\Tmc$-consistent $\Amc''\subseteq\Amc$, $\Amc''\neq\Amc_3$ such that for every $\alpha\in\Amc_3\setminus\Amc''$, there exists $\beta\in\Amc''\setminus\Amc_3$ such that $\beta\succ\alpha$. 
First note that since $\Phi$ is unsatisfiable, it is not possible that $\Amc''$ contains a $\Tmc$-consistent set of $P$ and $N$ assertions such that every $c_i$ has an outgoing $P$ or $N$ edge in $\Amc''$. By construction of $\succ$, it thus follows that $\Amc''$ contains at least one $\mi{Unsat}$ assertion. 
If $\mi{All}(a)\notin \Amc''$, then by construction of $\succ$, $\mi{Sat}(a)\in \Amc''$, which contradicts the $\Tmc$-consistency of $\Amc''$ since $\mi{Sat}\sqsubseteq \neg \exists\mi{Unsat}\in\Tmc$. 
It follows that $\mi{All}(a)\in \Amc''$. By $\Tmc$-consistency of $\Amc''$, $\mi{Exist}(a)$ and the  $\mi{Block}_{\mn{all}}$ assertions are thus not in $\Amc''$. Since they are the only assertions dominating the $\mi{Block}_\exists$ assertions, it follows that $\mi{Block}_\exists$ assertions are in $\Amc'$, so that there is no $P$ or $N$ assertions in $\Amc'$, and $\Amc''=\Amc_3$.
\end{itemize}
We now show that $\greps{\Kmc_\succ}\subseteq\{\Amc_1,\Amc_2,\Amc_3\}$. Let $\Amc'\in\greps{\Kmc_\succ}$. 
\begin{itemize}
\item Assume that $\mi{Sat}(a)\in\Amc'$. Then $\Amc'$ does not contain $\mi{All}(a)$, or any $\mi{Unsat}(a,c_i)$. Assume for a contradiction that 
$\Amc'\neq\Amc_1$. 
Since $\mi{Block}_{\mn{all}}(a,x_j) \succ P(c_i, x_j)$, 
$\mi{Block}_{\mn{all}}(a,x_j) \succ N(c_i, x_j)$, $\mi{Block}_{\mn{all}}(a,x_j) \succ \mi{Block}_\exists(a,x_{j'})$, and $\mi{Block}_{\mn{all}}(a,x_j) \succ \mi{Exist}(a)$, 
then $\Amc_1$ is a global improvement of $\Amc'$. 
It follows that $\Amc'=\Amc_1$.

\item Assume that $\mi{All}(a)\in\Amc'$. Then $\Amc'$ does not contain $\mi{Exist}(a)$, $\mi{Sat}(a)$, or any $\mi{Block}_{\mn{all}}(a, x_j)$. 
Assume for a contradiction that $\Amc'\neq\Amc_3$. 
If $\Amc'\subseteq\Amc_3$, it is clear that $\Amc'$ is not a repair, so $\Amc'\not\subseteq\Amc_3$. Thus $\Amc'$ must contain some $P$ or $N$ assertions. But in this case we obtain a global improvement of $\Amc'$ by replacing the $P$ and $N$ assertions by the conflicting $\mi{Block}_{\exists}$ assertions. 

\item In the case where $\mi{All}(a)\notin\Amc'$ and $\mi{Sat}(a)\notin\Amc'$, we show that $\Amc'=\Amc_2$. 
First note that since $\mi{All}(a)\notin\Amc'$ and the $\mi{Block}_{\mn{all}}$ assertions are of greater priority than any other assertions they are in conflict with, $\Amc'$ contains all the $\mi{Block}_{\mn{all}}$ assertions. 
By \Tmc-consistency and maximality of $\Amc'$, it follows that $\Amc'=\Amc_2$.
\end{itemize}
We have thus shown that $\greps{\Kmc_\succ}=\{\Amc_1,\Amc_2,\Amc_3\}$.

\mypar{Hardness of uniqueness and enumeration for Pareto-optimal repairs}
We adapt the reduction given in the proof of Theorem~\ref{ThmUniqueness}: since in this reduction, imposing transitivity on $\succ$ would lead to $\mi{Unsat}(a,c_i)\succ P(c_i,x_j)$ and $\mi{Unsat}(a,c_i)\succ N(c_i,x_j)$, we prevent this by creating two versions (identified by indexes 1 and 2) of all concepts and roles and letting $\mi{Unsat}_1$ in conflict with $P_2$, $N_2$ instead of $P_1$, $N_1$, and similarly for $\mi{Unsat}_2$. 
Given a conjunction of clauses $\Phi=c_1\wedge \dots\wedge c_k$ over variables $x_1,\dots, x_n$, we construct $\Kmc=\tup{\Tmc,\Amc}$ and $\succ$ as follows.
\begin{align*}
\Tmc=&\{\exists P_l^-\sqsubseteq\neg\exists N_l^-,\exists\mi{Block}_l^-\sqsubseteq\neg \exists P_l^-,
\exists\mi{Block}_l^-\sqsubseteq\neg \exists N_l^-
\mid l\in\{1,2\}\}\\
&\cup\{
\exists P_l\sqsubseteq\neg \exists\mi{Unsat}_r^-,
\exists N_l\sqsubseteq\neg \exists\mi{Unsat}_r^-\mid (l,r)\in\{(1,2),(2,1)\}\}
\\
&\cup\{
\mi{Exist}_l\sqsubseteq\neg\exists\mi{Block}_l, 
\mi{Exist}_l\sqsubseteq\neg\exists\mi{Unsat}_l,
\mi{Exist}_l\sqsubseteq\neg\mi{NoExist}_l
\mid l\in\{1,2\}\}\\
\Amc=&\{\mi{Unsat}_l(a,c_i)\mid 1\leq i\leq k, l\in\{1,2\}\}\cup\{P_l(c_i,x_j)\mid x_j\in c_i, l\in\{1,2\}\}\cup\{N_l(c_i,x_j)\mid \neg x_j\in c_i, l\in\{1,2\}\}
\\
&\cup\{\mi{Block}_l(a,x_j)\mid 1\leq j\leq n, l\in\{1,2\}\}
\cup\{\mi{Exist}_l(a),\mi{NoExist}_l(a)\mid l\in\{1,2\} \}
\end{align*}
\begin{align*}
\mi{Unsat}_l(a,c_i)\succ& \mi{Exist}_l(a) & \,l\in\{1,2\}\\
\mi{Exist}_l(a)\succ&\mi{Block}_l(a,x_j) & l\in\{1,2\}\\
\mi{Block}_l(a,x_j)\succ& P_l(c_i,x_j) &l\in\{1,2\}\\
\mi{Block}_l(a,x_j)\succ& N_l(c_i,x_j) &l\in\{1,2\}
\end{align*}
Note that $\succ$ is transitive (in particular, because the $\mi{Unsat}_1$ assertions are in conflict with the $P_2$ and $N_2$ assertions but not with the $P_1$ and $N_1$, and vice versa). Let $$\Amc'=\{\mi{NoExist}_1(a), \mi{NoExist}_2(a)\}\cup\{\mi{Block}_1(a,x_j), \mi{Block}_2(a,x_j)\mid 1\leq j\leq n\}\cup\{\mi{Unsat}_1(a,c_i), \mi{Unsat}_2(a,c_i)\mid 1\leq i\leq k\}.$$
It is easy to check that $\Amc'\in\preps{\Kmc_\succ}$: $\Amc'$ is \Tmc-consistent, and it is not possible to add any assertion $\alpha\in\Amc\setminus\Amc'$ to $\Amc'$ while staying \Tmc-consistent without removing an assertion $\beta\in\Amc'$ such that $\alpha\not\succ\beta$.
We show that $\Phi$ is unsatisfiable iff $\Amc'$ is the unique Pareto-optimal repair. \smallskip\\
\noindent($\Leftarrow$) Assume that $\Phi$ is satisfiable: there exists a valuation $\nu$ that satisfies $\Phi$.
Let $\Amc_\nu=\{P_1(c_i,x_j), P_2(c_i,x_j)\mid x_j\in c_i, \nu(x_j)=\true\}\cup\{N_1(c_i,x_j), N_2(c_i,x_j)\mid \neg x_j\in c_i, \nu(x_j)=\false\}
\cup\{\mi{Exist}_1(a), \mi{Exist}_2(a)\}$. 
Using arguments similar to those used in the case where $\succ$ was not transitive, we can show that $\Amc_\nu\in\preps{\Kmc_\succ}$, so that $|\preps{\Kmc_\succ}|\geq 2$.
\smallskip\\
\noindent($\Rightarrow$) In the other direction, assume that there exists a Pareto-optimal repair $\Amc''\neq\Amc'$. 
First we show that for $l\in\{1,2\}$, if $\mi{Exist}_l(a)\notin\Amc''$, then $\Amc''$ contains $\mi{NoExist}_l(a)$ and all $\mi{Block}_l$ and $\mi{Unsat}_l$ assertions, but no $P_l$ or $N_l$ assertions. 
Indeed, if $\mi{Exist}_l(a)\notin\Amc''$, 
by maximality, $\mi{NoExist}_l(a)\in\Amc''$. 
Moreover, if $\Amc''$ contains some $P_l$ or $N_l$ 
assertions, then all $P_l$ or $N_l$ 
edges incoming in some $x_j$ can be replaced by $\mi{Block}_l(a,x_j)$ to obtain a Pareto improvement of $\Amc''$ and $\Amc''\notin\preps{\Kmc_\succ}$. Thus $\Amc''$ contains no $P_l$ or $N_l$ 
assertions. By maximality, it thus contains all $\mi{Block}_l$ and $\mi{Unsat}_l$ assertions. 

Thus if $\mi{Exist}_1(a)\notin\Amc''$ and $\mi{Exist}_2(a)\notin\Amc''$, then $\Amc''=\Amc'$. 
It follows that $\mi{Exist}_l(a)\in\Amc''$ for $l=1$ or $l=2$. Assume w.l.o.g. that $\mi{Exist}_1(a)\in\Amc''$.  
Assume for a contradiction that $\mi{Exist}_2(a)\notin\Amc''$. It follows from the preceding paragraph that $\Amc''$ contains no $P_2$ or $N_2$ assertions. 
Thus $\mi{Exist}_1(a)$ can be replaced by the $\mi{Unsat}_1$ assertions to obtain a Pareto improvement of $\Amc''$ and $\Amc''\notin\preps{\Kmc_\succ}$. It follows that for $l\in\{1,2\}$, $\mi{Exist}_l(a)\in\Amc''$. 
Hence $\mi{NoExist}_l(a)\notin\Amc''$ and there is no $\mi{Block}_l$ or $\mi{Unsat}_l$ assertions in $\Amc''$. Thus the remaining assertions of $\Amc''$ form 
 a maximal subset of the $P_1$, $N_1$ and $P_2$, $N_2$ assertions such that no $x_j$ has both $P_l$ and $N_l$ incoming edges for $l\in\{1,2\}$. The valuation of the $x_j$ defined by $\nu(x_j)=\true$ iff $x_j$ has an incomming $P_1$ edge in $\Amc''$ satisfies $\Phi$. Otherwise, there would be a $c_i$ without $P_1$ or $N_1$ outgoing edge and it would be possible to add $\mi{Unsat}_2(a,c_i)$ and remove $\mi{Exist}_2(a)$ to improve $\Amc''$ since $\mi{Unsat}_2(a,c_i)\succ \mi{Exist}_2(a)$. It follows that $\Phi$ is satisfiable.

\mypar{Hardness of P-AR, P-IAR, P-brave, C-AR, C-IAR, C-brave IQ entailment}
For Pareto-optimal and completion-optimal repairs, \conp-hardness (resp. \np-hardness) of AR and IAR (resp. brave) IQ entailment in \dllitecore follow from the case where the priority relation is given by priority levels, in which the three families of optimal repairs coincide \cite{DBLP:conf/aaai/BienvenuBG14,DBLP:phd/hal/Bourgaux16}. 
These results apply since when $\succ$ is given by priority levels, $\succ$ is transitive: if $\alpha_1\succ\dots\succ\alpha_n$, then 
$s(\alpha_1)>\dots>s(\alpha_n)$, so if $\{\alpha_1,\alpha_n\}\subseteq C\in\conflicts{\Kmc}$, then $\alpha_1\succ\alpha_n$.

\mypar{Hardness of G-AR, G-IAR, G-brave IQ entailment} The reduction used in the proof Theorem~\ref{ThmHardnessGlobal} is such that $\succ$ is transitive.

\mypar{Hardness of globally-optimal repair checking} 
We use the reduction given for the lower bound of globally-optimal repairs enumeration: we have shown that $\Amc_3\in\greps{\Kmc_\succ}$ iff $\Phi$ is unsatisfiable.
\end{proof}

\section{Proofs for Section \ref{sec:arg}}

\thmSymmCoherent*
\begin{proof}
Let $P=(\args,\defeat, \succ)$ be a symmetric PAF. 
Take some preferred extension $E$ of $P$. Then $E$ is a preferred extension of the corresponding AF 
$F=(\args, \sdefeat)$. Suppose for a contradiction that $E$ is not a stable extension of the PAF $P$, which 
means $E$ is not a stable extension of $F$, \ie $E^+\neq\args\setminus E$. Since $E$ is conflict-free, it must then be the case that there is some $\beta \in \args \setminus E$
such that $\beta \not \in E^+$. Because $\succ$ is acyclic, we can w.l.o.g. choose $\beta$ such that there is no $\beta' \in (\args \setminus E) \setminus E^+$  with  $\beta' \succ \beta$.  
Since $E$ is a preferred extension, $E \cup \{\beta\}$ cannot be admissible in $F$, and since $E \cup \{\beta\}$ is conflict-free, 
it follows that there is an undefended attack on $E \cup \{\beta\}$. 
Since $E$ defends itself, the attack must be on~$\beta$, \ie $\beta \not \in \chff(E \cup \{\beta\})$. We can thus find an attack $\gamma \sdefeat \beta$ such that $\gamma \not \in E \cup \{\beta\}$ and $\gamma \not \in  (E \cup \{\beta\})^+$. In particular, $\beta \not \sdefeat \gamma$. As $P$ is a symmetric PAF, we know that $\beta \defeat \gamma$, and since $\beta \not \sdefeat \gamma$, it must be the case that $\gamma \succ \beta$. But we now have $\gamma \in (\args \setminus E) \setminus E^+$ and $\gamma \succ \beta$, which contradicts our choice of $\beta$, concluding our argument. 
\end{proof}

\thmSymPAFChar*
\begin{proof}

Let $F=(\args, \defeat_\succ)$ be the corresponding AF of a symmetric PAF $P=(\args,\defeat, \succ)$, and assume that  $\alpha_1 \defeat_\succ \alpha_2\defeat_\succ \dots\defeat_\succ \alpha_n\defeat_\succ \alpha_1$. Since $F$ is the corresponding AF of $P$, then $\alpha_1 \defeat \alpha_2\defeat \dots\defeat \alpha_n\defeat \alpha_1$. Moreover, since $P$ is symmetric, $(\args,\defeat)$ is symmetric, and it follows that $\alpha_2\defeat \alpha_1$, ..., $\alpha_n \defeat \alpha_{n-1}$, and $\alpha_1\defeat\alpha_n$ 
Finally, since $\succ$ is acyclic, there must exist $(j,i)\in\{(1,2),\dots,(n-1,n),(n,1)\}$ such that $\alpha_j\not\succ\alpha_i$, which implies that $\alpha_i\defeat_\succ\alpha_j$.

In the other direction, let $F=(\args, \defeat)$ be an AF such that for any cycle $\alpha_1 \defeat \alpha_2\defeat \dots\defeat \alpha_n\defeat \alpha_1$, there exists $(j,i)\in\{(1,2),\dots,(n-1,n),(n,1)\}$ with $\alpha_i\defeat\alpha_j$. 
Define $P=(\args,\defeat^P, \succ)$ by $\defeat^P=\{(\alpha,\beta)\mid \alpha\defeat\beta\}\cup\{(\alpha,\beta)\mid \beta\defeat\alpha\}$ and $\alpha\succ\beta$ iff $\alpha\defeat\beta$ and $\beta\not\defeat\alpha$. 
We first show that $\succ$ is acyclic. If it was not the case, there would be a cycle $\alpha_1 \succ \alpha_2 \succ \dots \succ \alpha_n \succ \alpha_1$. By construction of $\succ$, it would then be the case that $\alpha_1 \defeat \alpha_2\defeat \dots\defeat \alpha_n\defeat \alpha_1$ and for every $(j,i)\in\{(1,2),\dots,(n-1,n),(n,1)\}$, $\alpha_i\not\defeat\alpha_j$, a contradiction. 
It is then easy to verify that $F$ is the corresponding AF of $P$: 
\begin{itemize}
\item if $\alpha \defeat \beta$ then  $\alpha \defeat^P \beta$ by construction of $\defeat^P$ and $\beta \not \succ \alpha$ otherwise $\alpha\not\defeat\beta$ by construction of $\succ$; 
\item if (i) $\alpha \defeat^P \beta$ and (ii) $\beta \not \succ \alpha$, then it follows from (i) that $\alpha \defeat \beta$ or $\beta\defeat\alpha$, and from (ii) that $\beta\not\defeat\alpha$ or $\alpha\defeat \beta$, so $\alpha\defeat \beta$.\qedhere
\end{itemize}
\end{proof}

\thmStrongSymmSETAF*
\begin{proof}
Let $F=(\args,\defeat)$ be a strongly symmetric SETAF. Suppose for a contradiction that there is a preferred extension $E$ that is not a stable extension.  Then $(\args \setminus E) \setminus E^+ \neq \emptyset$. Let $B=\{\beta_1, \ldots, \beta_n\}$
be a $\subseteq$-maximal subset of $(\args \setminus E) \setminus E^+ $ that is conflict-free. We show with the following two claims that $E \cup B$ is admissible, which contradicts the fact that $E$ is a preferred extension, and therefore allows us to conclude that $E$ is stable. \medskip\\
\textbf{Claim}: $E \cup B$ is conflict-free. \\
\emph{Proof of claim}. Suppose that $E \cup B$ is not conflict-free. As $E$ is conflict-free and $B$ is conflict-free, any attack contained in $E \cup B$ must contain elements from both $E$ and $B$. Moreover, since $F$ is strongly symmetric, we may assume that the attack targets an element of $E$. Thus, there exists $S \subseteq E \cup B$ with $S \cap B \neq \emptyset$ and $\varepsilon \in E$ such that $S \defeat \varepsilon$. Since $E$ is a preferred extension, we must have $(S \setminus E) \cap E^+ \neq \emptyset$, which implies $B \cap E^+ \neq \emptyset$. This is a contradiction, since $B \subseteq \args \setminus E^+$. \medskip\\
\textbf{Claim}: $E \cup B \subseteq \chff(E \cup B)$. \\
\emph{Proof of claim}. Suppose for a contradiction that $E \cup B \not \subseteq \chff(E \cup B)$. Then there exists an attack 
$S \defeat \alpha$ such that $\alpha \in E \cup B$ and $S \cap (E \cup B)^+ = \emptyset$. Since $E$ is a preferred extension, we know that $E \subseteq \chff(E) \subseteq \chff(E \cup B)$, so $\alpha \in B$. Since $E \cup B$ is conflict-free (by the preceding claim), there must exist $\sigma \in S \setminus (E \cup B)$. As $S \cap (E \cup B)^+ = \emptyset$, we must have 
$\sigma \in (\args \setminus E) \setminus E^+$. From the maximality of $B$, we know that $B \cup \{\sigma\}$ is not conflict-free, 
and since $B$ is conflict-free, there must be an attack that involves $\sigma$. Moreover, as $F$ is strongly symmetric, 
we can assume that $\sigma$ is the target of the attack, which yields $\sigma \in B^+$, again contradicting $S \cap (E \cup B)^+ = \emptyset$. \qedhere
\end{proof}

\thmPSETAF*
\begin{proof}
Let $P=(\args, \defeat, \succ)$ be a strongly symmetric PSETAF such that $\succ$ is transitive,
and let $F=(\args, \sdefeat)$ be the corresponding SETAF.  Suppose for a contradiction that $S_0$ 
is a preferred extension of $P$ but not a stable extension of $P$. It follows that $S_0$ is a preferred extension of $F$
but not a stable extension of $F$. To obtain a contradiction, we will construct an admissible set which is a proper superset of $S_0$. To initiate the construction, we will require the following claim.  \medskip\\
\textbf{Claim 1}. There exists $\alpha \in (\args \setminus S_0)$ such that $S_0 \cup \{\alpha\}$ is conflict-free in $F$.\\
\emph{Proof of claim}. Since $S_0$ is preferred but not stable, there is $\alpha \in (\args \setminus S_0)$ such that
$\alpha \not \in S_0^+$. Suppose for a contradiction that $S_0 \cup \{\alpha\}$ is not conflict-free. As $\alpha \not \in S_0^+$,  it follows that there is some $S' \subseteq S_0$ and $\sigma \in S_0$ such that $S' \cup \{\alpha\} \sdefeat \sigma$.  However, since $S_0$ is preferred, it must defend itself against this attack, and since $S_0$ is conflict-free, the only possibility is that $\alpha \in S_0^+$, which yields a contradiction. (end proof of claim)\medskip\\
We now proceed to construct a finite sequence $S_0, S_1, S_2, \ldots, S_n$ of subsets of $\args$ as follows:
\begin{itemize}
\item  While there exists $\alpha$ such that (i) $\alpha \in (\args \setminus S_i)$ and (ii) $S_i \cup \{\alpha\}$ is conflict-free
\begin{itemize}
\item Pick some $\alpha_i$ that satisfies (i) and (ii) and such that there is no $\alpha_i'$ satisfying (i) and (ii) with $\alpha_i' \succ \alpha_i$ (such an $\alpha_i$ must exist due to the acyclicity of $\succ$). 
\item Set $S_{i+1} = S_i \cup \{\alpha_i\}$
\end{itemize}
\end{itemize}
Since for every $i \geq 0$, we have $S_i \subseteq \args$ and $S_i \subsetneq S_{i+1}$, the construction will terminate after a finite number of iterations. We observe that each $S_i$ is conflict-free. We will show that the final set $S_n$ is conflict-free and defends itself against all attacks, which contradicts the fact that $S_0$ is a preferred extension. The main ingredient in showing this is the following claim: \medskip\\
\textbf{Claim 2}. For each $0 \leq i \leq n$, the following property holds: 
\begin{quote}
For each attack $T \sdefeat \sigma$  on $S_i$ such that $T \cap S_i^+ = \emptyset$ (i.e.\ $S_i$ does not defend against this attack): 
\begin{description}
\item[(a)] if $i \geq 1$ and $\tau \in T \setminus S_{i}$, then the set  $S_{i-1} \cup \{\tau\}$ is conflict-free;
\item[(b)] for every $\tau \in T\setminus S_i$, $\tau \not \succ \sigma'$ for all $\sigma' \in S_i \setminus S_0$;
\item[(c)] there exists $\tau \in T\setminus S_i$ such that $\tau \not \succ \tau'$ for every $\tau' \in T \cup \{\sigma\} \setminus \{\tau\} $.
\end{description}
\end{quote}
\emph{Proof of claim}. The base case ($i=0$) is trivial, as $S_0$ is a preferred extension, and so there is no undefended attack. We will suppose that we have already shown the property to hold for $i=0, \ldots, k$ and prove it still holds for $i=k+1$. 
Let us consider some attack $T \sdefeat \sigma$ with $\sigma \in S_{k+1}$
such that $T \cap S_{k+1}^+ = \emptyset$. 
We need to show that conditions (a), (b), and (c) are satisfied. We prove each condition in turn. 
\medskip\\
\textbf{Claim 2.1}: Condition (a) holds: 
for every $\tau \in T \setminus S_{k+1}$, the set  $S_k \cup \{\tau\}$ is conflict-free. \\
\emph{Proof of claim.} Suppose for a contradiction that $S_k \cup \{\tau\}$ is not conflict-free. The attack contained in $S_k \cup \{\tau\}$ cannot target $\tau$, since otherwise we would have $\tau \in T \cap S_k^+$, contradicting the fact that 
$T\cap S_{k+1}^+=\emptyset$. The conflict must however involve $\tau$ since $S_{k}$ is conflict-free. Thus, there must exist $S^* \subseteq S_k$ and $\sigma^* \in S_k$ such that  \mbox{$S^* \cup \{\tau\} \sdefeat \sigma^*$}. 
If 
$ (S^* \cup \{\tau\})\cap S_{k+1}^+ \neq \emptyset$, then we must have $\tau \in S_{k+1}^+$, which again gives a contradiction. 
Thus, $ (S^* \cup \{\tau\})\cap S_{k+1}^+ = \emptyset$, which implies  $ (S^* \cup \{\tau\})\cap S_{k}^+ = \emptyset$. 
We can therefore apply the induction hypothesis for $i=k$. As $(S^* \cup \{\tau\}) \setminus S_k = \{\tau\}$, it follows from condition (c) applied to the undefended attack $S^* \cup \{\tau\} \sdefeat \sigma^*$
that  $\tau \not \succ \beta$ for all $\beta \in S^* \cup \{\sigma^*\}$. Since $P$ is strongly symmetric, it follows that we also have the attack $S^* \cup \{\sigma^*\} \sdefeat \tau$. But this means that $T \cap S_k^+ \neq \emptyset$, so $T \cap S_{k+1}^+ \neq \emptyset$,  a contradiction. (end proof of claim)
\medskip\\
\textbf{Claim 2.2}: Condition (b) holds: 
 for every $\tau \in T\setminus S_{k+1}$, $\tau \not \succ \sigma'$ for all $\sigma' \in S_{k+1} \setminus S_0$.\\
\emph{Proof of claim.} Suppose for a contradiction that the condition does not hold. Since $S_{k+1} \setminus S_0= \{\alpha_0, \ldots, \alpha_{k-1}, \alpha_{k}\}$, this means that there exists $\tau \in T\setminus S_{k+1}$ and $0 \leq j \leq k$ such that 
$\tau \succ \alpha_j$. Since $\alpha_j$ was selected for addition to $S_j$ despite $\tau \succ \alpha_j$,
it must be the case that $S_j \cup \{\tau\}$ is not conflict-free. As $S_j \subseteq S_k$, the set $S_k \cup \{\tau\}$ must also contain some attack, 
which contradicts Claim 2.1, which states that $S_k \cup \{\tau\}$ is conflict-free. 
(end proof of claim)\medskip\\
\textbf{Claim 2.3}: Condition (c) holds: 
there exists $\tau \in T\setminus S_{k+1}$ such that $\tau \not \succ \tau'$ for every $\tau' \in T \cup \{\sigma\} \setminus \{\tau\}$. \\
\emph{Proof of claim.} Suppose for a contradiction that condition (c) does not hold. Then for every $\tau \in T\setminus S_{k+1}$, there exists $\tau' \in (T \cup \{\sigma\} \setminus \{\tau\})$ such that $\tau \succ \tau'$. However, due to Claim 2.2, we know that $\tau' \not \in S_{k+1} \setminus S_0$. Moreover, we also know that $\sigma \in S_{k+1} \setminus S_0$, since $T \cap S_{k+1}^+ = \emptyset$, while every attack on $S_0$ is defended ($S_0$ being a preferred extension). Thus, we have:
\begin{equation}\label{prop1}
\text{For every $\tau \in T\setminus S_{k+1}$, there exists $\tau' \in (T \setminus (S_{k+1} \setminus S_0)) \setminus \{\tau\}$ such that $\tau \succ \tau'$.}
\end{equation}
Define $S_0^*$ as the restriction of $S_0$ to those elements $s_0 \in S_0$ such that $s_0 \not \succ \tau'$ for every 
$\tau' \in S_{k+1} \setminus S_0$. We aim to show that if $\tau \in T\setminus S_{k+1}$ and $s_0 \in S_0 \setminus S_0^*$, then $\tau \not \succ s_0$. Suppose for a contradiction that we have $\tau \succ s_0$ for such a $\tau$ and $s_0$. 
From $s_0 \in S_0 \setminus S_0^*$, we know that there exists $\tau'  \in S_{k+1} \setminus S_0$ such that $s_0 \succ \tau'$. 
But then by transitivity of $\succ$, we must also have $\tau \succ \tau'$, which contradicts Claim 2.2.
We can thus refine the preceding statement as follows:
\begin{equation}\label{prop2}
\text{For every $\tau \in T\setminus S_{k+1}$, there exists $\tau' \in (T \setminus (S_{k+1}\cup \{\tau\})) \cup (T \cap S_0^*)$ such that $\tau \succ \tau'$.}
\end{equation}
Next consider $\sigma_0 \in T \cap S_0^*$, and suppose for a contradiction that there is no $\tau' \in T \cup \{\sigma\}$ such that $\sigma_0 \succ \tau'$.
Since $P$ is strongly symmetric and $T \sdefeat \sigma$, we must have $T \setminus \{\sigma_0\} \cup \{\sigma\} \sdefeat \sigma_0$. We know that 
$(T \setminus \{\sigma_0\} \cup \{\sigma\}) \cap S_0^+ \neq \emptyset$
since $S_0$ is a preferred extension of $F$. As $\sigma \in S_{k+1}$ and $S_{k+1}$ is conflict-free, 
it must be the case that $(T \setminus S_{k+1}) \cap S_0^+ \neq \emptyset$. But this means that $T \cap S_{k+1}^+ \neq \emptyset$, a contradiction. When combined with the definition of $S_0^*$ and the transitivity of $\succ$, we get:
\begin{equation}\label{prop3}
\text{For every $\sigma_0 \in T \cap S_0^*$, there exists $\tau' \in T \setminus (S_{k+1} \setminus S_0^*)$ such that $\sigma_0 \succ \tau'$.}
\end{equation}
Since $S_{k+1}$ is conflict-free, we know that $T\setminus S_{k+1} \neq \emptyset$. 
This means that the set $(T \setminus S_{k+1}) \cup (T \cap S_0^*)$ is non-empty. Moreover, 
statements (\ref{prop2}) and (\ref{prop3}) show that for every $\tau \in (T \setminus S_{k+1}) \cup (T \cap S_0^*)$ 
there exists $\tau' \in (T \setminus S_{k+1}) \cup (T \cap S_0^*)$ such that $\tau \succ \tau'$. 
This implies the existence of a cycle in the preference relation $\succ$, contradicting the fact that $\succ$ is acyclic. (end proof of claim)\medskip\\
\indent We have thus established Claim 2. Now consider the final set $S_n$, which is, by construction, conflict-free. Suppose for a contradiction that $S_n \not \subseteq \chff(S_n)$. Then there is an attack $U \sdefeat \sigma_n$ on some $\sigma_n\in S_n$
such that $U \cap S_n^+ = \emptyset$. Since $S_n$ is conflict-free, $U \setminus S_n \neq \emptyset$. Consider some $\upsilon \in U \setminus S_n$. Since $\upsilon \not \in S_n$, the set $S_n \cup \{\upsilon\}$ is not conflict-free, and so there must be some attack that involves but does not target $\upsilon$. Let $S_n' \cup \{\upsilon\} \sdefeat \sigma_n'$ be such an attack, where $S_n' \subseteq S_n$ and $\sigma_n' \in S_n$. Since $S_n$ is conflict-free and $\upsilon \not \in S_n^+$, this attack is not defended by $S_n$, i.e. $(S_n' \cup \{\upsilon\}) \cap S_n^+ = \emptyset$. We can thus apply Claim 2, condition (c), to obtain $\upsilon \not \succ \beta$ for all $\beta \in S_n' \cup \{\sigma_n'\}$. 
It follows that we also have the attack $S_n' \cup \{\sigma_n'\} \sdefeat \upsilon$, which contradicts $U \cap S_n^+ = \emptyset$.

We have thus exhibited a set $S_n \supsetneq S_0$ that is conflict-free and such that $S_n \subseteq \chff(S_n)$, which contradicts that fact that $S_0$ is a preferred extension. Having obtained the desired contradiction, we can conclude that every preferred extension of a strongly symmetric PSETAF with a transitive preference relation is also a stable extension. 
\end{proof}

\section{Proofs for Section \ref{sec:repairs-extensions}}

\begin{lemma}\label{lem:conflict-free}
Let $\Kmc_\succ$ be a prioritized KB with $\Kmc=\tup{\Tmc, \Amc}$ and $\Amc'\subseteq\Amc$. The set $\Amc'$ is \Tmc-consistent iff $\Amc'$ is conflict-free in the associated PSETAF $F_{\Kmc, \succ}$.
\end{lemma}
\begin{proof}
Assume for a contradiction that $\Amc'$ is \Tmc-consistent and not conflict-free. There exists $\alpha\in \Amc'\cap\Amc'^+$, so there exists $A\subseteq \Amc'$ such that $A\defeat_\Kmc\alpha$. However, it follows that $(A\cup\{\alpha\})\in\conflicts{\Kmc}$, contradicting the \Tmc-consistency of $\Amc'$.

If $\Amc'$ is \Tmc-inconsistent, there exists 
 $C\in\conflicts{\Kmc}$ such that $C\subseteq\Amc'$. 
Since $\succ$ is acyclic and $C$ is finite, there exists $\gamma\in C$ such that for every $\delta\in C$, $\gamma\not\succ\delta$. It follows that $C\setminus\{\gamma\}\defeat_\Kmc\gamma$. Thus $\gamma\in C\cap C^+$, \ie $C$ is not conflict-free, hence neither is~$\Amc'$. 
\end{proof}

\thmstable*
\begin{proof}
Assume for a contradiction that $\Amc'\in\preps{\Kmc_\succ}$  and is not a stable extension of $F_{\Kmc, \succ}$.  
By Lemma \ref{lem:conflict-free}, $\Amc'$ is conflict-free so ${\Amc'}^+\subseteq \Amc\setminus\Amc'$. 
Since $\Amc'$ is not stable, there must then exist $\alpha\in \Amc\setminus\Amc'$ such that $\alpha\notin {\Amc'}^+$. Since $\Amc'$ is a repair, $\Amc'\cup\{\alpha\}$ is \Tmc-inconsistent, and since $\alpha\notin {\Amc'}^+$, then for every $B\subseteq\Amc'$ such that $B\cup\{\alpha\}$ is \Tmc-inconsistent, there exists $\beta\in B$ such that $\alpha\succ\beta$. 
Then $\Amc'\setminus\{\beta\mid \alpha\succ\beta\}\cup\{\alpha\}$ is a Pareto improvement of $\Amc'$. 
It follows that $\Amc'\notin\preps{\Kmc}$.

In the other direction, assume for a contradiction that $\Amc'$ is a stable extension of $F_{\Kmc, \succ}$ and $\Amc'\notin\preps{\Kmc_\succ}$. 
By Lemma~\ref{lem:conflict-free}, since $\Amc'$ is stable, $\Amc'$ is $\Tmc$-consistent. 
Since $\Amc'\notin\preps{\Kmc_\succ}$, it follows that 
there exists $\alpha\in\Amc\setminus\Amc'$ such that $\alpha$ is \Tmc-consistent and for every $B\subseteq\Amc'$ such that $B\cup\{\alpha\}$ is \Tmc-inconsistent, there exists $\beta\in B$ such that $\alpha\succ\beta$. 
Since $\Amc'$ is stable, there exists $A\subseteq \Amc'$ such that $A\defeat_\Kmc\alpha$ (otherwise $\alpha\notin{\Amc'}^+$). 
But then $A\cup\{\alpha\}$ is \Tmc-inconsistent, so there exists $\beta\in A$ such that $\alpha\succ\beta$, which contradicts $A\defeat_\Kmc\alpha$. 
\end{proof}

\thmTransCoherent*
\begin{proof}
Let $\Kmc_\succ$ be a prioritized KB with 
$\Kmc=\tup{\Tmc,\Amc}$,  and a transitive priority relation $\succ$.
The associated PSETAF is $F_{\Kmc, \succ}= (\Amc, \defeat_{\Kmc}, \succ)$,
where $\defeat_{\Kmc} = \{(C \setminus \{\alpha\}, \alpha) \mid C \in \conflicts{\Kmc}, \alpha \in C\}$.

Let $F_\mn{cl}= (\Amc, \defeat_{\Kmc}, \succ')$ where $\succ'$ is the transitive closure of $\succ$. The transitive closure of an acyclic binary relation is acyclic, so $\succ'$ is an acyclic binary relation over $\Amc$, so $F_\mn{cl}$ is a PSETAF. 
Moreover, observe that $F_\mn{cl}$ is strongly symmetric and $\succ'$ is transitive, so by Theorem \ref{psetaf-coh}, $F_\mn{cl}$ is coherent. 

We show that $F_{\Kmc, \succ}$ and $F_\mn{cl}$ have the same corresponding SETAF, so that  $F_{\Kmc, \succ}$ is coherent. 
The corresponding SETAF of $F_{\Kmc, \succ}$ is $(\Amc,\defeat_{\Kmc,\succ})$, where 
$S \defeat_{\Kmc,\succ} \alpha$ iff $S \defeat_\Kmc \alpha$ and $\alpha \not \succ \beta$ for every $\beta \in S$, and the corresponding SETAF of $F_\mn{cl}$ is $(\Amc,\defeat_{\Kmc,\succ'})$, where 
$S \defeat_{\Kmc,\succ'} \alpha$ iff $S \defeat_\Kmc \alpha$ and $\alpha \not \succ' \beta$ for every $\beta \in S$. 
Since $\succ\subseteq\succ'$, we thus have $\defeat_{\Kmc,\succ'}\subseteq\defeat_{\Kmc,\succ}$. 
Assume for a contradiction that there exists $S\defeat_{\Kmc,\succ}\alpha$ such that $S\not\defeat_{\Kmc,\succ'}\alpha$. There must be $\beta\in S$ such that $\alpha\succ'\beta$. Since $\succ'$ is the transitive closure of $\succ$, it follows that there exists $\gamma_1,\dots,\gamma_n$ such that $\alpha\succ\gamma_1\succ\dots\succ\gamma_n\succ\beta$. But then, since $\{\alpha,\beta\}\subseteq S\cup\{\alpha\}\in\conflicts{\Kmc}$ and $\succ$ is a transitive priority relation, it follows that $\alpha\succ\beta$. Hence $S\not\defeat_{\Kmc,\succ}\alpha$. We conclude that $\defeat_{\Kmc,\succ'}=\defeat_{\Kmc,\succ}$ and that$F_{\Kmc, \succ}$ and $F_\mn{cl}$ have the same corresponding SETAF. 
\end{proof}

\thmGroundedPareto*
\begin{proof}
Since $\Bmc\in\preps{\Kmc_\succ}$, 
by Theorem \ref{tstable}, $\Bmc$ is a stable extension of $F_{\Kmc, \succ}$. 
Since every stable extension is also a complete extension, and $\Gmc$ is the $\subseteq$-minimal complete extension of $F_{\Kmc, \succ}$, it follows that $\Gmc\subseteq\Bmc$.
\end{proof}

\thmGroundedNoPref*
\begin{proof}
Let $\Gmc$ be the grounded extension of $F_\Kmc$ and $\Amc^\cap$ be the intersection of the repairs of $\Kmc$. 
Let $\alpha\in\Amc^\cap$. There is no $C\in\conflicts{\Kmc}$ such that $\alpha\in C$ (otherwise $C\setminus\{\alpha\}$ would be \Tmc-consistent and could be extended to a repair of $\Kmc$ that would not contain $\alpha$). 
Thus there is no $A\defeat_\Kmc\alpha$. It follows that $\alpha\in \chf{F_\Kmc}(\emptyset)\subseteq\Gmc$. Hence $\Amc^\cap\subseteq\Gmc$. 
In the other direction, let $\Amc'\in\reps{\Kmc}$. 
Since $\succ_\emptyset$ is the empty relation, then $\reps{\Kmc}=\preps{\Kmc_{\succ_\emptyset}}$, so $\Amc'\in\preps{\Kmc_{\succ_\emptyset}}$ and by Theorem \ref{ground-pareto}, $\Gmc\subseteq\Amc'$. It follows that $\Gmc\subseteq\Amc^\cap$.
\end{proof}

\thmElect*
\begin{proof}
Let $\alpha\in \mi{Elect}(\Kmc_\succ)$ and $F=(\Amc, \defeat_{\Kmc,\succ})$ be the SETAF corresponding to $F_{\Kmc,\succ}$. 
Since for every $C\in\conflicts{\Kmc}$ such that $\alpha\in C$, there exists $\beta\in C$ such that $\alpha\succ\beta$, there is no $C\in\conflicts{\Kmc}$ such that $C\setminus\{\alpha\}\defeat_{\Kmc,\succ}\alpha$, so no $S\subseteq \Amc$ such that $S\defeat_{\Kmc,\succ}\alpha$. 
It follows that $\alpha\in \chf{F}(\emptyset)\subseteq \Gmc$.
\end{proof}

\thmGroundedHardness*
\begin{proof}
We adapt the proof of \ptime-hardness for the problem of deciding whether an argument belongs to the grounded extension of an argumentation framework given in \cite{DBLP:journals/flap/DvorakD17} (Reduction 3.8). We modify the original log-space reduction from deciding whether a variable is in the minimal model of a definite Horn theory so that the argument framework does not contain cycles, which are not allowed in our setting. 
Let $\varphi=\{r_l\mid 1\leq l\leq n\}$ be a set of definite Horn clauses over a set of variables $X$, \ie $r_l= b_{l,1}\wedge\dots\wedge b_{l,i_l}\rightarrow h_l$, with $b_{l,j},h_l \in X$, and let $z\in X$ be the variable we want to decide whether it belongs to the minimal model of $\varphi$. We assume w.l.o.g. that for every rule $r_l$, $h_l\neq b_{l,j}$ for all $j$. 

We construct a \dllitecore KB $\Kmc=\tup{\Tmc,\Amc}$ and priority relation as follows.
\begin{align*}
\Tmc=&\{ \exists \mi{FalseInBod}^-\sqsubseteq \neg \exists \mi{SatBodHead}, \exists \mi{SatBodHead}^-\sqsubseteq \neg\exists \mi{FalseInBod}, \exists \mi{SatBodHead}^-\sqsubseteq \neg F, F\sqsubseteq \neg T\}\\
\Amc=&\{\mi{FalseInBod}(b_{l,j},r_l)\mid r_l\in\varphi, 1\leq j\leq i_l \}\cup\{\mi{SatBodHead}(r_l, h_l)\mid r_l\in\varphi\}\cup\{F(z), T(z)\}
\end{align*}
$$\mi{SatBodHead}(r_l, x)\succ \mi{FalseInBod}(x,r_{l'})\quad\quad \mi{SatBodHead}(r_l, z)\succ F(z)$$ 
Clearly, $\succ$ is acyclic. 
We show that $T(z)$ belongs to the grounded extension $\Gmc$ of $F_{\Kmc,\succ}$ iff $z$ belongs to the minimal model of~$\varphi$. 
The corresponding AF $F=(\Amc, \defeat_{\Kmc,\succ})$ of $F_{\Kmc,\succ}$ is such that 
\begin{align*}
\defeat_{\Kmc,\succ}=&\{(\mi{SatBodHead}(\_, x), \mi{FalseInBod}(x,\_))\mid x\in X\}\cup\{( \mi{SatBodHead}(\_,z),F(z))\}\cup\\&
\{(\mi{SatBodHead}(r_l, \_), \mi{FalseInBod}(\_,r_l)), (\mi{FalseInBod}(\_,r_l), \mi{SatBodHead}(r_l, \_))\mid r_l\in\varphi\}\cup\\
&\{(F(z), T(z)), (T(z),F(z))\}
\end{align*}
The intuition behind the reduction is that an argument $\mi{SatBodHead}(r_l, x)$ is in the grounded extension $\Gmc$ iff all variables in the body of $r_l$ are in the minimal model of $\varphi$, so that $x$ is also in this model, and that an argument $\mi{FalseInBod}(x,\_)$ or $F(x)$, stating that $x$ is not in the model, is attacked by $\Gmc$ only if $x$ is in the model. 
That is, when computing $\Gmc$ via iteratively applying the characteristic function $\chf{F}$ we simulate the following algorithm for deciding whether $z$ is in the minimal model of $\varphi$. The algorithm starts with the rules with empty body and adds their rule heads to the minimal model: 
$\chf{F}(\emptyset)=\{\mi{SatBodHead}(r_l, x)\mid\text{there is no }\mi{FalseInBod}(\_,r_l)\in\Amc\}$. 
Then it iteratively considers all rules with the body already being part of the minimal model and adds their heads to the minimal model until a fixed-point is reached: if $G_i$ is the result at step $i$, 
let $G_i^+=\{\mi{FalseInBod}(y,\_)\mid \mi{SatBodHead}(\_, y)\in G_i\}\cup\{F(z) \text{ if there is some }\mi{SatBodHead}(\_, z)\in G_i \}$ be the set of arguments that are attacked by $G_i$, 
$\chf{F}(G_i)=\{\mi{SatBodHead}(r_l, x)\mid\text{there is no }\mi{FalseInBod}(\_,r_l)\in \Amc\setminus G_i^+\}\cup\{T(z)\text{ if }F(z)\in G_i^+ \}$. 
Hence $z$ is in the minimal model of the Horn-formula $\varphi$ iff $T(z)\in \Gmc$, \ie $\Kmc_\succ\models_{GR} T(z)$.
\end{proof}

We recall the definition of well-founded semantics and introduce some notation and terminology. Our presentation is loosely based upon Chapter 15 from \cite{AbiteboulHV95}, but tailored (and simplified) to suit our purposes. We refer readers to the latter chapter for a more detailed introduction to the semantics of rule languages with negation. 

A \emph{normal logic program} (aka Datalog$^\neg$ program, henceforth abbreviated to program) is a finite set of \emph{rules} of the form $\rho= \mathtt{h} \leftarrow \bp_1, \ldots, \bp_n, \nnot\ \bp_{n+1}, \ldots, \nnot\ \bp_{n+m}$, where each $\bp_i$ is an atom $\mathtt{p}(t_1, \ldots,t_n)$ whose every term $t_i$ is either a constant or a variable. The atom $\mathtt{h}$ is called the head of the rule $\rho$, and the body of $\rho$ is $\bp_1, \ldots, \bp_n, \nnot\ \bp_{n+1}, \ldots, \nnot\ \bp_{n+m}$. The atoms $\bp_1, \ldots, \bp_n$ are called \emph{positive body atoms}, while the atoms $\bp_{n+1}, \ldots, \bp_{n+m}$ (which are preceded by $\nnot$) are called \emph{negative body atoms}. As is common, we require rules to satisfy the following safety condition: every variable that appears in the head atom of a rule must occur in some positive body atom of that rule. Note that it is allowed for the body of a rule to be empty (empty ruleheads can also be considered, but will not be used in what follows). 

A \emph{definite rule} is a rule that does not contain any negative body atoms, and a \emph{definite program} is a finite set of definite rules. A \emph{ground atom (resp.\ rule, program)} is an atom (resp.\ rule, program) that does not contain any variables. We will sometimes use the term \emph{fact} in place of \emph{ground atom}. 
Given any program $\Pi$, the \emph{grounding of $\Pi$}, denoted $grnd(\Pi)$, consists of all ground rules that can be obtained from a rule in $\Pi$ by instantiating its variables with constants occurring in $\Pi$. 

Given a ground program $\Pi$, we let $B_\Pi$ consist of all facts that appear (positively or negatively) in $\Pi$. A \emph{2-valued instance} $\mathbf{I}$ for $\Pi$ assigns to  each fact in $B_\Pi$ a value that is either true or false.  We denote by $\bot$ the 2-valued instance that assigns false to all facts in $B_\Pi$. A \emph{3-valued instance} for $\Pi$ is defined similarly, except now there are three possible values: true, unknown, or false. It will be convenient to view 2-valued instances as sets of (positive) facts, where a fact is assigned true if it belongs to the set, and otherwise, is assigned false. Similarly, a 3-valued instance can be given as a consistent set of literals (possibly negated facts), where facts appearing positively in the set are assigned true, those whose negation is in the set are assigned false, and all other facts are treated as unknown. We use the notation $\mathbf{I}(\fact)$ to refer to the truth value that $\mathbf{I}$ assigns to the fact  $\fact$.

Given two 3-valued instances $\mathbf{I}$ and $\mathbf{J}$, we write $\mathbf{I} < \mathbf{J}$ to mean that $\mathbf{I}(\fact) < \mathbf{J}(\fact)$ for ever ground atom $\fact$, where the truth values are ordered as follows: false $<$ unknown $<$ true. 
A \emph{(2-valued) model} of a ground definite program $\Pi$ is a 2-valued instance for $\Pi$ such that for every rule $\rho \in \Pi$, if all body atoms in $\rho$ are assigned true, then the head atom is also assigned true. It is known that every 
ground definite program has a unique \emph{minimal model}, i.e., a model $\mathbf{I}$ such that $\mathbf{I} < \mathbf{J}$ for every model $\mathbf{J}$. We denote by $mm(\Pi)$ the minimal model of a 
ground definite program $\Pi$.

We next extend minimal models to stratified programs. A \emph{stratification} of a program $\Pi$ 
is a sequence $\Pi^{1}, \ldots, \Pi^n$ and a function $\sigma$ mapping predicates to $[1,n]$ such that the following conditions hold:
\begin{itemize}
\item $\Pi$ is the disjoint union of $\Pi^{1}, \ldots, \Pi^n$
\item for every predicate $P$ that occurs in $\Pi$, all rules that have $P$ in the head belong to $\Pi^{\sigma(P)}$
\item for every $1 \leq i \leq n$:
\begin{itemize}
\item if predicate $P$ occurs positively in the body of some rule of $\Pi^i$, then $\sigma(P) \leq i$
\item if predicate $P$ occurs negatively in the body of a rule of $\Pi^i$, then $\sigma(P) < i$
\end{itemize}
\end{itemize}
A program $\Pi$ is stratified if there exists a stratification for $\Pi$. 
Given a program $\Pi$ with stratification $\Pi^{1}, \ldots, \Pi^n$, we proceed as follows. We start by computing the grounding $grnd(\Pi)$, 
and let $gr(\Pi^{1}), \ldots, gr(\Pi^n)$ be the ground programs corresponding to $\Pi^{1}, \ldots, \Pi^n$ (note that they are grounded w.r.t.\ all constants occurring in $\Pi$, not just the constants from the stratum). We set $M_1$ equal to the minimal model of the ground definite program $gr(\Pi^{1})$. Then for $1 \leq i < n$, we let $M_{i+1} = M_i \cup M'_{i+1}$, where $M'_{i+1}$ is the minimal model of the ground definite program obtained from $gr(\Pi^{i+1}$) by (i) deleting every rule whose body contains $\nnot \fact$ where $\fact \in M_i$, and (ii) removing all negative facts from the remaining rule bodies. We call $M_n$ the minimal model of $\Pi$. 

We now proceed to the definition of the well-founded semantics, which can be applied to arbitary programs. For every ground program $\Pi$ and 2-valued instance\footnote{This construction in \cite{AbiteboulHV95} was originally defined for 3-valued instances, but as is noted there, only 2-valued instances are needed for the fixpoint construction of the well-founded semantics. Focusing on 2-valued instances simplifies some definitions.} $\mathbf{I}$ for $\Pi$, we denote by $\Pi|_\mathbf{I}$ the ground definite program that is obtained from $\Pi$ as follows: 
(i) delete every rule whose body contains $\nnot \fact$ where $\fact \in \mathbf{I}$, 
(ii) remove all negative body atoms from the remaining rules.
We then consider the following infinite sequence of 2-valued instances:
$$\mathbf{I}_0 = \bot \qquad \mathbf{I}_{i+1} = mm(\Pi|_{\mathbf{I}_i})$$
It is known (and can be readily verified) that for every $i > 0$, we have:
$$\mathbf{I}_0 <  \mathbf{I}_2 < \ldots <\mathbf{I}_{2i} <  \mathbf{I}_{2i+2} <  \mathbf{I}_{2i+1} < \mathbf{I}_{2i-1} < \ldots < \mathbf{I}_{1}$$
As there are a finite number of possible instances, the sequences $(\mathbf{I}_{2i})_{i \geq 0}$ and $(\mathbf{I}_{2i+1})_{i \geq 0}$ both reach a fixpoint. 
We let $\mathbf{I}_*$ denote the limit of the increasing sequence $(\mathbf{I}_{2i})_{i \geq 0}$, 
and let $\mathbf{I}^*$ denote the limit of the decreasing sequence $(\mathbf{I}_{2i+1})_{i \geq 0}$. 
Then the (unique) \emph{well-founded model} of $\Pi$ 
is the 3-valued instance that assigns true to all facts $\fact$ such that $\fact \in \mathbf{I}_*$, 
assigns false to all facts $\fact$ such that $\fact \not \in \mathbf{I}^*$, and unknown to all other facts. 
If $\Pi$ is an arbitrary (not necessarily ground) program, then the well-founded model 
of $\Pi$ is defined as the well-founded model of $grnd(\Pi)$. 

It is known that if $\Pi$ is stratified, then the well-founded model of $\Pi$ is equal to the minimal model of $\Pi$. More precisely, if $\Pi$ has a stratification with $p$ levels, then 
$\mathbf{I}_{2p}$ is the minimal model, and we have $\mathbf{I}_{2p} = \mathbf{I}_{2p+i}$ for all $i \geq 0$.

\medskip

We recall the following result from \cite{DBLP:journals/ai/Dung95} (very slightly modified to suit our later extensions) that shows how to compute the grounded extension of an AF via well-founded semantics. 
It makes reference to the following two rules:
$$r_{\accp}: \accp(x) \leftarrow \mathtt{arg}(x), \nnot \defp(x) \qquad r_{\defp}: \defp(x) \leftarrow \attackp(y,x), \accp(y)$$

\begin{theorem}
Let $F=(\args,\defeat)$ be an AF. Then $E \subseteq \args$ is the grounded extension of $F$
iff 
$$\{\accp(\alpha), \neg \defp(\alpha) \mid \alpha \in E\} 
\cup \{\defp(\beta), \neg \accp(\beta) \mid \beta \in E^+\} \cup \{\attackp(\alpha, \beta)  \mid \alpha \defeat \beta\} \cup \{\mathtt{arg}(\alpha) \mid \alpha \in \args\}$$
is the well-founded model of the program $\{r_{\accp}, r_{\defp}\} \cup \{\attackp(\alpha, \beta) \leftarrow \, \mid \alpha \defeat \beta\}\cup \{\mathtt{arg}(\alpha) \leftarrow \, \mid \alpha \in \args\}$.
\end{theorem}

The following theorem generalizes the preceding result to $k$-SETAFs, i.e., SETAFs for which $S \defeat \beta$ implies $|S| \leq k$. 
We will use relations $\attackp_1, \ldots, \attackp_k$ to store attacks, where $\attackp_i(\alpha_1, \ldots, \alpha_i, \beta)$ specifies that there is an attack $\{\alpha_1, \ldots, \alpha_i\} \defeat \beta$. We will use the following variants of rule $r_{\defp}$, where $i$ ranges from $1$ to $k$:
$$r_{\defp}^{i}: \defp(x) \leftarrow \attackp_i(y_1, \ldots, y_i,x), \accp(y_1), \ldots, \accp(y_i)$$

\textbf{Theorem \ref{lemma:accp}.} 
\emph{Let $F=(\args,\defeat)$ be a $k$-SETAF. Then $\alpha$ is in the grounded extension of $F$ iff 
$\accp(\alpha)$ belongs to the well-founded model of the following normal logic program:}
\begin{align*}
&\{\defp(x) \leftarrow \mathtt{att}_i(y_1,.., y_i,x),  \accp(y_1),...,\accp(y_i) \mid 1 \leq  i \leq k\} \cup \{\accp(x) \leftarrow \mathtt{arg}(x), \nnot \defp(x)\} \\
&\cup \{\mathtt{arg}(\alpha) \mid \alpha \in \args\}  \cup \{\mathtt{att}_i(\alpha_1, \ldots, \alpha_i, \beta) \leftarrow \, \mid \{\alpha_1, \ldots, \alpha_i\} \defeat \beta\}
\end{align*}
\begin{proof} 
Let $\Pi_F= \{r_{\accp}, r_{\defp}^1, \ldots, r_{\defp}^k\} \cup \{\mathtt{arg}(\alpha) \mid \alpha \in \args\}  \cup \{\attackp_i(\alpha_1, \ldots, \alpha_i, \beta) \leftarrow \, \mid \{\alpha_1, \ldots, \alpha_i\} \defeat \beta\}$, and let us consider the sequence of sets produced during the computation of the well-founded model of $\Pi_F$.
A straightforward inductive argument shows that the following statements holds for every $j \geq 0$:
\begin{align*}
\mathbf{I}_{2j +2} &= S_{\mathtt{arg}} \cup S_{\attackp} \cup \{\accp(\alpha) \mid \alpha \in \chff^{j+1}(\emptyset)\} \cup \{\defp(\beta) \mid \beta \in (\chff^{j+1}(\emptyset))^+ \}\\
\mathbf{I}_{2j +1} &= S_{\mathtt{arg}} \cup  S_{\attackp} \cup \{\accp(\alpha) \mid \alpha \in \args \setminus (\chff^{j}(\emptyset))^+ \}
\cup \{\defp(\beta) \mid \beta \in (\args \setminus (\chff^{j}(\emptyset))^+)^+ \}
\end{align*}
where $S_{\mathtt{arg}}= \{\mathtt{arg}(\alpha) \mid \alpha \in \args\} $, $S_{\attackp}=\{\attackp_i(\alpha_1, \ldots, \alpha_i, \beta)\mid \{\alpha_1, \ldots, \alpha_i\} \defeat \beta\}$ and $(\chff^{0}(\emptyset))^+=\emptyset$.
Recall that these are 2-valued instances, so we list the positive facts, with all other facts being assigned false. 

Now let $E$ be the grounded extension of $F$. 
Then there is $\ell $ such that $E = \chff^{\ell}(\emptyset)$, and $\chff^{\ell+j}(\emptyset)=\chff^{\ell}(\emptyset)$ for all $j \geq 1$. 
Then we have:
\begin{align*}
\mathbf{I}_* &= S_{\mathtt{arg}} \cup  S_{\attackp} \cup \{\accp(\alpha) \mid \alpha \in \chff^{\ell}(\emptyset)\} \cup \{\defp(\beta) \mid \beta \in (\chff^{\ell}(\emptyset))^+ \}\\
&=  S_{\mathtt{arg}} \cup S_{\attackp} \cup \{\accp(\alpha) \mid \alpha \in E\} \cup \{\defp(\beta) \mid \beta \in E^+ \}\\
\mathbf{I}^* &= S_{\mathtt{arg}} \cup  S_{\attackp} \cup \{\accp(\alpha) \mid \alpha \in \args \setminus (\chff^{\ell}(\emptyset))^+ \}
\cup \{\defp(\beta) \mid \beta \in (\args \setminus (\chff^{\ell}(\emptyset))^+)^+ \}\\
&= S_{\mathtt{arg}} \cup  S_{\attackp} \cup \{\accp(\alpha) \mid \alpha \in \args \setminus E^+ \}
\cup \{\defp(\beta) \mid \beta \in (\args \setminus E^+)^+ \}
\end{align*}
Note that if $\beta \not \in (\args \setminus E^+)^+$, then for every $S \defeat \beta$, there exists $\gamma \in S$ such that $\gamma \in E^+$. It follows that $\beta \in \chff(E)=E$. We therefore have that
the well-founded model of $\Pi_F$ is equal to
$$S_{\mathtt{arg}} \cup  S_{\attackp} \cup \{\accp(\alpha), \neg \defp(\alpha) \mid \alpha \in E\} \cup \{\defp(\beta), \neg \accp(\beta) \mid \beta \in E^+ \}$$
which completes the proof. 
\end{proof}

Next we need to define rules to compute conflicts, in order to derive which attack facts should be included. We assume that we have a finite set $Q_{\mathsf{unsat}}$ of BCQs such that $\Amc$ is $\Tmc$-consistent iff no $q \in Q_{\mathsf{unsat}}$ evaluates to true on $\Amc$ (or more formally, on the interpretation $\Imc_\Amc$ based upon $\Amc$). Note that when $\Tmc$ is formulated in \dlliterhorn, such a set $Q_{\mathsf{unsat}}$ can be computed using standard query rewriting procedures. In what follows, we let $N$ be the maximal number of conjuncts in any BCQ in $Q_{\mathsf{unsat}}$. 

To be able to identify conflicts, and later repairs, we assume that each assertion has an additional final argument that contains a unique id for that assertion, e.g. $A(a,u_5), R(a,b,u_{17})$. Given a (standard) ABox $\Amc$, we will use $\Amc_{id}$ to denote a corresponding set of annotated assertions, and for every $\alpha \in \Amc$, we let $id(\alpha)$ be the id of the corresponding atom in $\Amc_{id}$. We will store the priority relation in a binary relation $\prefp$ over assertion ids, with $\prefp(id(\alpha), id(\beta))$ meaning that $\alpha \succ \beta$. We let $Q_{\mathsf{unsat}}^*$ consist of all BCQs $q^*$ that can be obtained from some $q \in Q_{\mathsf{unsat}}$ as follows: (i) order the atoms in $q$, and (ii) add a final argument $z_j$ to the $j$th atom (where $z_j$ is a fresh variable not occurring in $Q_{\mathsf{unsat}}$ and distinct from $z_k$ with $k \neq j$). 

We introduce predicates $\confp_i$ for $i=1, \ldots, N$, which will store conflicts that involve exactly $i$ assertions. Concretely, for each BCQ $q^* \in Q_{\mathsf{unsat}}^*$ with a single atom 
we have the following rule (which can be omitted if we assume no self-contradicting assertions):
$$ \confp_0(z_1) \leftarrow q^*(z_1)$$
and for every BCQ $q^* \in Q_{\mathsf{unsat}}^*$ with $k+1$ atoms, we have:
$$ \confp_{k}(z_1, \ldots, z_k, z_{k+1}) \leftarrow q^*(z_1, \ldots, z_k, z_{k+1}), \bigwedge_{\substack{0 \leq j < k\\
(z_{\ell_1}, \ldots, z_{\ell_j}, z_{\ell_{j+1}}) \in \{z_1, \ldots, z_k, z_{k+1}\}^{j+1}}} \nnot \confp_{j}(z_{\ell_1}, \ldots, z_{\ell_j}, z_{\ell_{j+1}})
 $$
Here we write $q^*(z_1, \ldots, z_k)$ to highlight the shared variables $z_j$, but $q^*$ will contain also the original variables in the BCQ $q$. Observe that negated body atoms are used to ensure subset minimality. We can then compute the attack relations by using the following rule for every $1 \leq i < N$:
$$r_{\mathtt{att}}^{i}: \attackp_i(z_1, \ldots, z_i, z_{i+1}) \leftarrow \confp_{i}(z_1, \ldots, z_i, z_{i+1}), \nnot \prefp(z_{i+1}, z_1), \ldots, \nnot \prefp(z_{i+1}, z_i)$$
where 
We generate the set of arguments (corresponding to ids of assertions that are not self-contradictory) by means of the following rules, for every concept name $A$ and role name $R$:
$$\mathtt{arg}(z_1) \leftarrow A(x,z_1), \nnot \confp_0(z_1) \qquad  \mathtt{arg}(z_1) \leftarrow R(x,y,z_1), \nnot \confp_0(z_1) $$
Let $\Pi_{\confp}$ be the preceding set of rules. We observe that $\Pi_{\confp}$ is stratified. It is not hard to see that these rules allow us to compute the conflicts, and hence the attacks in the SETAF corresponding to the PSETAF $F_{\Kmc, \succ}$. Specifically, we have:

\begin{lemma} \label{lemma:attackp}Let $F=(\args, \defeat_\succ)$ be the corresponding SETAF for the PSETAF $F_{\Kmc, \succ}$ based upon the prioritized KB $\Kmc_\succ$. The following statements are equivalent:
\begin{enumerate}
\item $\{\alpha_1, \ldots, \alpha_n\} \defeat_\succ \beta$ 
\item the atom $\attackp_n(id(\alpha_1), \ldots, id(\alpha_{n}), id(\beta))$ belongs to the minimal model of  
$\Pi_{\confp} \cup \{\gamma \leftarrow \, \mid \gamma \in \Amc_{id}\} \cup \{ \prefp(id(\alpha), id(\beta))\mid \alpha \succ \beta\}$.
\end{enumerate}
\end{lemma}

Intuitively, we want to construct a program that first computes the minimal model of the stratified program $\Pi_{\confp}$, then uses the attack and conflict facts from this model as input for the `argumentation' program $$\Pi_{\mathtt{arg}}=\{r_{\accp}, r_{\defp}^1, \ldots, r_{\defp}^N\}$$ 
 However, if we simply apply well-founded semantics to the combined program, we might not get the desired result: while we are `waiting' for the computation of the stratified rules to finish, there could be some incorrectly derived $\confp_{i}$ facts, which could lead to wrongly derived $\attackp$ and $\defp$ facts, in turn wrongly eliminating some $\accp$ facts. 

The following lemma provides a means of safely combining the two programs to achieve the desired result. The basic idea is to delay the application of rules in the second program until the initial stratified program has stabilized. We note that the lemma is tailored to our purposes, and some of the hypotheses could be relaxed at the cost of a more involved proof. 

\begin{lemma}\label{lem:combine}
Let $\Pi_1$ and $\Pi_2$ be programs verifying the following conditions: (i) $\Pi_1$ is stratified and does not use any predicate / fact that occurs in the head of a rule from $\Pi_2$,  
and (ii) $\Pi_2$ may only use predicates / facts from $\Pi_1$ in the body of rules and in positive form. 
Then we can construct a program $\Pi_{1 \rightarrow 2}$ such that the following models contain precisely the same (positive) facts over the relations in $\Pi_2$:
\begin{itemize}
\item the well-founded model of $\Pi_{1\rightarrow 2}$
\item the well-founded model of $\Pi_2 \cup \{\alpha \leftarrow \, \mid \alpha \in mm(\Pi_1)\}$
\end{itemize}
\end{lemma}
\begin{proof}
Suppose that $\Pi_1$ has a stratification with $p$ levels. 
Consider the following auxiliary program:
$$\Pi_p = \{\oddp(x) \leftarrow \nextp(y,x), \nnot \oddp(y)\} \cup \{\nextp(i, i+1) \leftarrow \, \mid 0 \leq i <2p\}$$
We assume that the predicate names $\oddp, \nextp$ and constants $0, 1, \ldots, 2p$  do not occur in either $\Pi_1$ or $\Pi_2$ (otherwise, we can choose other fresh names / constants). 
A straightforward induction shows that the following instances will be produced during the construction of the well-founded model of $\Pi_p$:
\begin{align*}
\mathbf{I}_{2j +2} &= \{\oddp(i) \mid 1\leq i\leq 2p, i < 2j+2, \text{ and } i \text{ is odd } \} \cup   \{\nextp(i, i+1) \mid 0 \leq i <2p\}\\
\mathbf{I}_{2j +1} &=  
 \{\oddp(i) \mid 1\leq i\leq 2p \text{ and either } i \text{ is odd or } i > 2j+1 \}
 \cup \{\nextp(i, i+1) \mid 0 \leq i <p\}
\end{align*}
In particular, this means that the instance $\mathbf{I}_{2p}$, and every later instance in the sequence, contains 
$$\nnot \oddp(0), \oddp(1), \nnot \oddp(2), \ldots,  \oddp(2p-1), \nnot \oddp(2p)$$
but none of the earlier instances contains precisely this combination of positive and negative facts. 

Let us now modify the rules in $\Pi_2$ as follows: add $\nnot \oddp(0), \oddp(1), \nnot \oddp(2), \ldots, \oddp(2p-1), \nnot \oddp(2p)$ to the body of each rule $r$ (denoting the result $r'$). Let $\Pi_2'$ be the resulting set of rules. 

Now consider the program $\Pi_{1\rightarrow 2}=\Pi_1 \cup \Pi_2' \cup \Pi_p$ and let $\mathbf{J_0}, \mathbf{J}_1, \mathbf{J}_2, \ldots$ be the sequence of instances obtained during the construction of the well-founded model of $grnd(\Pi_1 \cup \Pi_2' \cup \Pi_p)$. 
We consider the three parts of the program:
\begin{itemize}
\item Sub-program $\Pi_1$: Due to the assumptions on $\Pi_1$, and the safety condition on rules that ensures all head variables occur in a positive body atom, it is not hard to show that for every $i \geq 2p$, $\mathbf{J}_i$ contains the same positive facts as the minimal model of $\Pi_1$, when restricted to predicates and facts from $\Pi_1$. Specifically, we have:
\begin{itemize}
\item every $\alpha \in mm(\Pi_1)$ belongs to $\mathbf{J}_i$
\item if $\alpha \in \mathbf{J}_i$ and $\alpha$ uses a predicate from $\Pi_1$, then $\alpha \in mm(\Pi_1)$
\end{itemize}
\item Sub-program $\Pi_p$: The rule $\oddp(x) \leftarrow \nextp(y,x), \nnot \oddp(y)$ is guarded by the predicate $\nextp$ which only occurs in $\Pi_p$ and which can only match the concrete facts given in $\Pi_p$. It follows that $\mathbf{J}_i$
and $\mathbf{I}_i$ contain precisely the same positive $\oddp$ and $\nextp$ facts for every $i \geq 0$. 
In particular, the set of literals $\{\nnot \oddp(0), \oddp(1), \nnot \oddp(2), \ldots,  \oddp(2p-1), \nnot \oddp(2p)\}$ belongs to 
$\mathbf{J}_i$ for every $i \geq 2p$, and it is not contained in $\mathbf{J}_i$ for $i < 2p$. 
\item Sub-program $\Pi_2'$:  Every rule in $\Pi_2'$ contains $\nnot \oddp(0), \oddp(1), \nnot \oddp(2), \ldots, \oddp(2p-1),  \nnot\oddp(2p)$ in its body. It follows from the preceding point that this set of literals is not contained in $\mathbf{J}_i$ for $i < 2p$. 
Thus, for $i < 2p$, the rules of $\Pi_2'$ cannot fire. For every $i \geq 2p$, the set of literals $\{\nnot \oddp(0), \oddp(1), \nnot \oddp(2), \ldots,  \oddp(2p-1), \nnot\oddp(2p)\}$ belongs to $\mathbf{J}_i$, so the rules in $\Pi_2'$ behave henceforth like the original rules in $\Pi_2$. Moreover, as noted above, when $i \geq 2p$, the instance $\mathbf{J}_i$ contains the same facts
as $mm(\Pi_1)$ w.r.t.\ predicates / facts from $\Pi_1$, so the same rule applications that were possible using $\{\alpha \leftarrow \, \mid \alpha \in mm(\Pi_1)\}$ can be reproduced in $\Pi_{1\rightarrow 2}$ starting from timepoint $2p$. Conversely, due to the safety condition on $\Pi_2$, all applications of rules from $\Pi_2'$ used to generate the instances $\mathbf{J}_{2p}, \mathbf{J}_{2p+1},  \ldots$ can be reproduced by the program $\Pi_2 \cup \{\alpha \leftarrow \, \mid \alpha \in mm(\Pi_1)\}$. It follows that the well-founded model of $\Pi_{1\rightarrow 2}$ and the well-founded model of $\Pi_2 \cup \{\alpha \leftarrow \, \mid \alpha \in mm(\Pi_1)\}$ contain precisely the same positive facts when restricted  to  the relations / facts in $\Pi_2$. \qedhere
\end{itemize}
\end{proof}

It is easily verified that $\Pi_{\confp}$ and $\Pi_{\mathtt{arg}}$ satisfy the conditions of programs $\Pi_1$ and $\Pi_2$ of the preceding lemma. We can therefore combine Theorem \ref{lemma:accp} and Lemmas \ref{lemma:attackp} and \ref{lem:combine} to obtain the following result. While stated for \dlliterhorn, the result holds for any constraint / ontology language for which conflicts can be identified by a finite set of CQs that depends solely on the ontology / constraints.

\Wfhorn*

A natural under-approximation of the grounded extension is obtained by considering $\chff^d(\emptyset)$ for some fixed number $d$. We describe how to adapt the preceding program to compute such an approximation. Essentially, we need to replace 
the rules $r_{\accp}, r_{\defp}^{i}$ by the following stratified set of rules (recall that $N$ is the maximal cardinality of conflicts):
\begin{align*}
\accp_1(x) &\leftarrow \mathtt{arg}(x)&&\\
\defp_1(x) &\leftarrow \attackp_i(y_1, \ldots, y_i, x), \accp_{1}(y_1), \ldots, \accp_{1}(y_i) &&1 \leq i < N\\
\accp_{j+1}(x) &\leftarrow \mathtt{arg}(x), \nnot \defp_{j}(x) && 1 \leq j < 2d \\
\defp_{j+1}(x) &\leftarrow \attackp_i(y_1, \ldots, y_i, x), \accp_{j+1}(y_1), \ldots, \accp_{j+1}(y_i) &&1 \leq j < 2d, 1 \leq i < N
\end{align*}
It is easily verified that the preceding rules behave exactly like $r_{\accp}, r_{\defp}^{i}$ for the first $2d$ iterations. To make this more formal, let $\Pi_{\mathtt{arg}}^d$ be the preceding rules. Consider some ABox $\Amc$, a priority relation $\succ$ between facts in $\Amc$, and define $S_{\Amc, \succ}$ as follows:
$$S_{\Amc, \succ} =  \{\gamma \leftarrow \, \mid \gamma \in \Amc_{id}\}
\cup \{ \prefp(id(\alpha), id(\beta)) \leftarrow \, \mid \alpha \succ \beta\} $$
Now let $\mathbf{K}_0, \mathbf{K}_1, \ldots$ and $\mathbf{L}_0, \mathbf{L}_1, \ldots$ be respectively the sequence of instances generated during the computation of the well-founded models of $\{r_{\accp}\} \cup \{r_{\defp}^{i}, 1 \leq i \leq N\} \cup S_{\Amc, \succ}$ and $\Pi_{\mathtt{arg}}^d \cup S_{\Amc, \succ}$ respectively. Then for every $1 \leq j \leq 2d$, we have the following:
\begin{itemize}
\item $\accp_j(\alpha) \in \mathbf{L}_j$ iff $\accp(\alpha) \in  \mathbf{K}_j$
\item $\defp_j(\alpha) \in \mathbf{L}_j$ iff $\defp(\alpha) \in  \mathbf{K}_j$
\end{itemize}
Thus, by combining the $\Pi_{\mathtt{arg}}^d $ with the program $\Pi_{\confp}$, we obtain a stratified program whose minimal model contains $\accp_{2d}(id(\alpha))$ iff $\alpha \in \Gamma_{F_{\Kmc, \succ}}^d(\emptyset)$. Moreover, by examining the program, we remark that there is no recursion inside the strata of the program. It is known that every non-recursive stratified datalog query can be equivalently expressed as a first-order query, essentially by `unfolding' the rules starting from a chosen goal predicate. We thus obtain:

\begin{theorem}
For every \dlliterhorn\ TBox $\Tmc$ and $d >0$, there exists a first-order query $\Psi_{\Tmc,d}$ with one free variable 
such that for every ABox $\Amc$, priority relation $\succ$ for $\Kmc=\tup{\Tmc, \Amc}$, and assertion $\alpha$, the following are equivalent:
\begin{itemize}
\item $\alpha$ belongs to $\Gamma_{F_{\Kmc, \succ}}^d(\emptyset)$ 
\item $\Psi_{\Tmc,d}(id(\alpha))$ evaluates to true over the finite first-order interpretation corresponding to $\Amc_{id} 
\cup \{ \prefp(id(\alpha), id(\beta))\mid \alpha \succ \beta\}$. 
\end{itemize}
\end{theorem}

\section{Comparison with Related Work}

Let $\trianglerighteq$ be a partial preorder (i.e. a reflexive and transitive binary relation) over the ABox $\Amc$ of a KB $\Kmc$.  
We show that the preferred repair $\mi{Partial}_{PR}(\Kmc,\trianglerighteq)$ defined in \cite{DBLP:conf/ksem/BelabbesB19} is precisely the intersection of the completion-optimal repairs of $\Kmc_{\succ_\trianglerighteq}$ where $\alpha\succ_\trianglerighteq\beta$ iff $\{\alpha,\beta\}\subseteq C\in\conflicts{\Kmc}$ and $\alpha\triangleright\beta$ (i.e. $\alpha\trianglerighteq\beta$ and $\beta\ntrianglerighteq\alpha$). 

We first recall the definition of $\mi{Partial}_{PR}(\Kmc,\trianglerighteq)$. Observe that given a total preorder $\geq$ over $\Amc$, $\succ_\geq$ is score-structured, so the three families of optimal repairs collapse. We denote by $\mi{OptRep}(\Kmc, \geq)$ the set $\creps{\Kmc_{\succ_\geq}}=\greps{\Kmc_{\succ_\geq}}=\preps{\Kmc_{\succ_\geq}}$. The repairs in $\mi{OptRep}(\Kmc, \geq)$ correspond to the preferred repairs used in \cite{DBLP:conf/ksem/BelabbesB19}, which are the $\subseteq_P$-repairs of \cite{DBLP:conf/aaai/BienvenuBG14} defined with the priority levels naturally associated with $\geq$.

\begin{definition}
Let $\geq$ be a total preorder that extends $\trianglerighteq$, i.e. such that for every $\alpha,\beta\in\Amc$, if $\alpha\trianglerighteq\beta$, then $\alpha\geq \beta$, and if $\alpha\triangleright\beta$, then $\alpha> \beta$.  
The intersection of preferred repairs associated with $(\Kmc,\geq)$ is $\mi{IPR}(\Kmc,\geq) = \bigcap_{\Amc' \in \mi{OptRep}(\Kmc, \geq)}\Amc'$.

The preferred repair associated with $(\Kmc,\trianglerighteq)$ is:
$\mi{Partial}_{PR}(\Kmc,\trianglerighteq)= \bigcap_{\geq\text{total extension of }\trianglerighteq} \mi{IPR}(\Kmc,\geq)$.
\end{definition}

\begin{theorem}
For every KB $\Kmc=\tup{\Tmc,\Amc}$ and partial preorder $\trianglerighteq$ over $\Amc$, 
$\mi{Partial}_{PR}(\Kmc,\trianglerighteq)=\bigcap_{\Amc' \in \creps{\Kmc, \succ_\trianglerighteq}}\Amc'$, where $\alpha\succ_\trianglerighteq\beta$ iff $\{\alpha,\beta\}\subseteq C\in\conflicts{\Kmc}$ and $\alpha\triangleright\beta$.
\end{theorem}
\begin{proof}
Let $\Amc^{\cap,C,\succ_\trianglerighteq}=\bigcap_{\Amc' \in \creps{\Kmc, \succ_\trianglerighteq}}\Amc'$ be the intersection of the completion-optimal repairs of $\Kmc_{\succ_\trianglerighteq}$. 

Let $\alpha_0\in \Amc^{\cap,C,\succ_\trianglerighteq}$, let $\geq$ be a total preorder that extends $\trianglerighteq$ and $\Amc'\in\mi{OptRep}(\Kmc, \geq)$. 
For every $\alpha,\beta$, 
if $\alpha{\succ_\trianglerighteq}\beta$, then $\{\alpha,\beta\}\subseteq C\in\conflicts{\Kmc}$ and $\alpha\triangleright\beta$, so since $\geq$ extends $\trianglerighteq$, $\alpha > \beta$, i.e. $\alpha\succ_\geq\beta$. 
It follows that $\succ_\geq$ extends $\succ_\trianglerighteq$. 
Hence, the completion-optimal repairs of $\Kmc_{\succ_\geq}$ form a subset of the completion-optimal repairs of $\Kmc_{\succ_\trianglerighteq}$, so their intersections are related as follows: $\bigcap_{\Amc' \in \creps{\Kmc, \succ_\trianglerighteq}}\Amc'\subseteq \bigcap_{\Amc' \in \creps{\Kmc, \succ_\geq}}\Amc'$. 
It follows that $\alpha_0\in \bigcap_{\Amc' \in \creps{\Kmc, \succ_\geq}}\Amc'=\bigcap_{\Amc' \in \mi{OptRep}(\Kmc, \geq)}\Amc'=\mi{IPR}(\Kmc,\geq)$. 
Hence, for every total preorder $\geq$ that extends $\trianglerighteq$, $\alpha_0\in \mi{IPR}(\Kmc,\geq)$. It follows that $\alpha_0\in \mi{Partial}_{PR}(\Kmc,\trianglerighteq)$. 
We conclude that $ \Amc^{\cap,C,\succ_\trianglerighteq}\subseteq \mi{Partial}_{PR}(\Kmc,\trianglerighteq)$.

In the other direction, let $\alpha_0\in \mi{Partial}_{PR}(\Kmc,\trianglerighteq)$ and let $\succ'$ be a completion of $\succ_\trianglerighteq$ and $\Amc_{\succ'}$ be the unique optimal repair defined by $\succ'$. 
Let $\geq$ be a total preorder that extends the preorder $\geq'$ defined as follows: $\alpha\geq'\beta$ if (i) $\alpha\trianglerighteq\beta$ or (ii) $\alpha\succ'\beta$. 
We show that $\geq$ extends $\trianglerighteq$. 
If $\alpha\trianglerighteq\beta$ then $\alpha\geq'\beta$ so $\alpha\geq\beta$. 
Assume that $\alpha\triangleright\beta$. 
If $\alpha$ and $\beta$ are in a conflict, $\alpha\succ_\trianglerighteq\beta$ so $\alpha\succ'\beta$ and $\beta\not\succ'\alpha$ by definition of a completion of a priority relation. 
If $\alpha$ and $\beta$ are not conflicting, $\beta\not\succ'\alpha$ trivially. 
Since $\beta\not\trianglerighteq\alpha$ and $\beta\not\succ'\alpha$, it follows that $\beta\not\geq'\alpha$. 
Hence $\alpha>'\beta$, so $\alpha>\beta$ since $\geq$ extends $\geq'$. 
It follows that $\geq$ extends $\trianglerighteq$. 
Thus $\alpha_0\in\mi{IPR}(\Kmc,\geq)$.
We show that $\Amc_{\succ'}\in \mi{OptRep}(\Kmc, \geq)$, which implies that $\mi{IPR}(\Kmc,\geq)\subseteq\Amc_{\succ'}$, so that $\alpha_0\in\Amc_{\succ'}$. 
Since $\mi{OptRep}(\Kmc, \geq)=\creps{\Kmc_{\succ_\geq}}$, it is sufficient to show that $\succ'$ is a completion of $\succ_\geq$. 
Let $\alpha\succ_\geq\beta$. By definition of $\succ_\geq$, $\alpha$ and $\beta$ are in a conflict and $\alpha>\beta$. 
Since $\succ'$ is a total priority relation and $\alpha$, $\beta$ are in a conflict, it must be the case that $\alpha\succ'\beta$ or $\beta\succ'\alpha$. 
However, $\beta\not\succ'\alpha$. Indeed, otherwise $\beta\geq'\alpha$ by construction of $\geq'$, and since $\geq$ extends $\geq'$, it follows that $\beta\geq\alpha$, contradicting $\alpha>\beta$. Hence $\alpha\succ'\beta$.  
It follows that $\succ'$ is indeed a completion of $\succ_\geq$, so that $\alpha_0\in\Amc_{\succ'}$.  
We conclude that $\mi{Partial}_{PR}(\Kmc,\trianglerighteq) \subseteq \Amc^{\cap,C,\succ_\trianglerighteq}$.
\end{proof}

\end{document}